\documentclass[msom,nonblindrev]{informs3} 

\OneAndAHalfSpacedXII 

\usepackage{titletoc}
\usepackage{natbib}

 \bibpunct[, ]{(}{)}{,}{a}{}{,}%
  \def\bibfont{\fontsize{11}{12}\selectfont} 
 \def\bibsep{\smallskipamount} 
 \def\bibhang{24pt}%

\usepackage{hyperref}
\hypersetup{
	colorlinks=true,
	linkcolor=black,
	filecolor=black,
	urlcolor=black,
	citecolor=blue,
}

\usepackage{endnotes}
\let\footnote=\endnote

\newcommand{\NDC}{\ensuremath{\mathrm{NDC}}}
\newcommand{\MDC}{\ensuremath{\mathrm{MDC}}}
\newcommand{\VDD}{\ensuremath{\mathrm{VDD}}}

\newcommand{\data}{\ensuremath{\mathrm{data}}}
\newcommand{\UP}{\ensuremath{\mathrm{UP}}}

\usepackage{makecell}
\usepackage{caption} 
\usepackage{mathtools}
\usepackage{booktabs} 
\usepackage{amssymb}
\usepackage{pifont}
\usepackage[margin=1in]{geometry}
\usepackage{fix-cm}
\usepackage{multirow}
\usepackage{tabularx}
\usepackage{amsmath}
\let\oldproof\proof
\let\oldendproof\endproof
\renewenvironment{proof}{\oldproof\normalfont}{\oldendproof}

\usepackage{algorithm}
\usepackage{algpseudocode}
\usepackage{enumitem}
\usepackage{geometry}

\usepackage{tikz}
\usepackage{pgfplots}
\usepackage{subcaption}
\pgfplotsset{compat=1.17}

\usetikzlibrary{positioning,decorations.pathreplacing,arrows.meta,fit,backgrounds}
\TheoremsNumberedThrough     
\ECRepeatTheorems

\def\qed{\hspace*{10pt}\hfill{$\square$}\hfilneg\par}
\EquationsNumberedThrough    

\newcommand{\ben}{\begin{equation}}   \newcommand{\een}{\end{equation}}
\newcommand{\bea}{\begin{eqnarray}}   \newcommand{\eea}{\end{eqnarray}}
\newcommand{\beq}{\begin{eqnarray*}}  \newcommand{\eeq}{\end{eqnarray*}}
\newcommand{\PaperTitle}{The Source of Data Value: Consumer Disclosure, Access Design, and Platform Monetization}
\MANUSCRIPTNO{} 
\pgfplotsset{
    phase diagram/.style={
        width=1.05\linewidth,
        height=6.5cm,
        xmin=0, xmax=2.2,
        ymin=0, ymax=1.5,
        xlabel={$v$ (Valuation)},
        ylabel={$c$ (Data Cost)},
        xtick={0, 0.5, 1, 1.5, 2},
        xticklabels={0, $\frac{1}{2}$, $t$, $\frac{3}{2}$, $2t$},
        ytick={0, 0.5, 1, 1.5},
        yticklabels={0, $\frac{1}{2}$, $t$, $\frac{3}{2}$},
        axis lines=left,
        tick label style={font=\scriptsize},
        label style={font=\footnotesize},
        enlargelimits=false,
        clip=false 
    },
    phase diagram sw/.style={
        width=0.7\linewidth, 
        height=9cm,
        xmin=0, xmax=2.2,
        ymin=0, ymax=1.5,
        xlabel={$v$ (Valuation)},
        ylabel={$c$ (Data Cost)},
        xtick={0, 0.5, 1, 1.5, 2},
        xticklabels={0, $\frac{1}{2}$, $t$, $\frac{3}{2}$, $2t$},
        ytick={0, 0.5, 1, 1.5},
        yticklabels={0, $\frac{1}{2}$, $t$, $\frac{3}{2}$},
        axis lines=left,
        tick label style={font=\small},
        label style={font=\normalsize},
        enlargelimits=false,
        clip=true
    }
}
\begin{document}


\RUNAUTHOR{Liu, Lu, and Wang}

\RUNTITLE{\PaperTitle}

\TITLE{\PaperTitle}

\ARTICLEAUTHORS{%
\AUTHOR{Hexiao Liu}
\AFF{School of Management, Zhejiang University, \EMAIL{liuhexiao123@163.com}}
\AUTHOR{Lijian Lu}
\AFF{HKUST Business School, Hong Kong University of Science and Technology, \EMAIL{lijianlu@ust.hk}}
\AUTHOR{Mingzheng Wang}
\AFF{School of Management, Zhejiang University, \EMAIL{wangmzh@zju.edu.cn}}
} 

\ABSTRACT{%
Digital marketplaces increasingly generate revenue from two interrelated activities: transaction commissions and seller-facing data services. A seller may purchase data to boost revenue or to avoid competing with an informed rival. This study investigates how consumer disclosure rights influence these motives and the platform's incentive to provide data.
We develop a game-theoretic model of a commissioned platform, two differentiated sellers, and privacy-sensitive consumers. The platform's data value arises from a \emph{sales-revenue effect} due to bilateral data use and a \emph{competition-threat effect} stemming from asymmetric access. Under mandatory collection, the platform can monetize defensive adoption, even if bilateral personalization decreases seller revenue. Under voluntary disclosure,  anonymous purchasing and strictly positive privacy costs make unilateral disclosure unattractive in the baseline. Exclusive access then yields no usable information, restoring sellers' no-data rejection payoffs and eliminating defensive fees. Profitable provision consequently requires nonnegative bilateral revenue gains. Under full fee extraction, sellers retain their no-data net profits. This revenue alignment does not ensure welfare improvement: additional transactions can benefit consumers, whereas segmentation in an already covered market can raise anonymous prices and reduce welfare. A cap on anonymous offers at their no-data benchmark protects individual consumers and confines positive-mass provision to the transaction-creation region. Extensions examine access choice, privacy-cost heterogeneity, imperfect identification, and limited contractual credibility. 
Platforms should evaluate data services against both bilateral-use and seller-rejection outcomes. Sellers need to differentiate between revenue improvement and defensive investments. Regulators should distinguish equal-access rules from effective consumer control and monitor pricing in anonymous channels. Changes in privacy governance not only affect data availability but also alter the sources and distribution of its commercial value.
}


\HISTORY{}

\maketitle

%


\section{Introduction}
Digital marketplaces intermediate product sales and control information that helps sellers target consumers. Amazon's Brand Tailored Promotions, for example, enables eligible brands to offer discounts of 10\%--50\% to selected audiences, including repeat purchasers, high-spending customers, and shoppers who have abandoned their carts \citep{toogood2026amazonbrandtailored}. Amazon also earns transaction-linked referral fees; its US schedule lists 15\% for Home and Kitchen, Office Products, and Toys and Games \citep{AmazonFees2026}. These practices place targeting capabilities and transaction intermediation within the same business. These practices motivate an operational question: \emph{how does consumers' ability to withhold information change what competing sellers will pay for a marketplace's data service?}

These considerations are crucial for platform investment, seller profitability, and the assessment of personalization. Sellers adoption can reflect two motives: improving pricing to increase revenue, or avoiding the loss they would face if a rival alone obtained access. A service can therefore attract paying subscribers even when widespread use reduces their product revenues. For a commissioned marketplace, this creates a potential conflict between transaction revenue and access fees.  A seller can pay to escape an unfavorable continuation even if universal adoption leaves the industry worse off. Evaluating the service requires two counterfactuals: bilateral use versus no data, and the outcome a seller faces after rejecting access. Adoption, conversion, and fee receipts alone cannot distinguish revenue improvement from defensive expenditure.

Consumer data rights make this distinction consequential beyond the seller's procurement decision. The GDPR's consent provisions and California's rights to opt out of specified data sales and sharing place different constraints on the use of personal information \citep{eu2016gdpr,california2020cpra}. A common operational issue is whether consumers can decline a particular use while continuing to purchase. Effective anonymous purchasing can preserve participation while denying a sole subscriber usable pricing information, thereby changing the outside option against which sellers evaluate a data contract. Privacy governance can thus affect not only information availability but also the contractual basis of data monetization. We isolate this mechanism where sellers depend on marketplace-provided information, and our governance regimes are economic benchmarks rather than complete representations of either statute.

\subsection{Our Approach and Contributions}
We develop a game-theoretic model of a platform that intermediates between
two horizontally differentiated sellers, earns an exogenous commission
share of product revenue, and chooses whether to offer a fee-based data
service. Subscribers can personalize prices using available information;
otherwise, they charge uniform prices. Consumers choose entry, disclosure,
and purchases, accounting for privacy costs. We compare no data collection
(NDC), mandatory data collection (MDC), and voluntary data disclosure (VDD), which preserves anonymous purchasing. The platform commits to sequential access offers and rejection continuations, whereas sellers cannot commit
to personalized offers before disclosure. Consumers choose entry and
disclosure before observing realized uniform prices. The reported provision
regions and welfare comparisons use this timing and the stated equilibrium
selection. Amazon's promotion tool motivates platform-controlled targeting
rather than these precise contractual and informational assumptions.

Our first contribution is to delineate a platform-specific breakdown of data's commercial value. We distinguish the \emph{sales-revenue effect}---the change in sellers’ product revenue when both use the data service rather than neither---from the \emph{competition-threat effect}---the revenue loss in the exclusive-access continuation after rejection, whether due to competition from an informed rival or reduced consumer participation. 
This decomposition clarifies how downstream outcomes determine access fees and commission income. Under MDC, sufficiently large and enforceable rejection losses can make provision profitable even when bilateral personalization reduces sellers’ revenue. The commitment analysis separates the existence of these losses from their extractability: when exclusive access requires a subsidy, weaker subsidy credibility reduces defensive fee extraction without altering the underlying revenue states.

Our second contribution establishes when consumer disclosure rights can mitigate the competition-threat component by altering the informational dynamics associated with unilateral adoption. Under VDD, disclosing to a sole subscriber
yields no additional purchasing surplus but incurs a strictly positive
privacy cost. Consumers therefore withhold information while retaining
anonymous purchasing, restoring the rejecting seller's no-data revenue.
VDD eliminates the threat effect, so provision requires nonnegative
bilateral seller-revenue gains. Allowing the platform to choose its access
structure also leaves exclusive access with zero incremental monetization
value. The substantive contribution is this upstream consequence of
disclosure: consumer control changes the disagreement payoffs supporting
platform contracts, rather than merely reducing the volume of data sold.

Third, we distinguish \emph{revenue alignment from consumer protection}.
VDD can expand coverage and benefit the platform and consumers in
aggregate, while full fee extraction leaves each seller at its no-data
net profit. In an already covered market, however, the selected
segmentation equilibrium can raise anonymous prices and gross seller
revenue without creating transactions, reducing consumer surplus and
welfare. An enforceable cap on anonymous prices at the no-data benchmark
addresses this remaining channel by preserving every consumer's no-data
purchasing option. With the platform free to withhold the service, the cap
protects consumers and leaves the platform and sellers weakly better off
than under no data, although it can reduce platform profit relative to
uncapped VDD. Disclosure rights and anonymous-price protection therefore
address distinct concerns: defensive access payments and revenue
extraction through the anonymous channel.

Extensions examine endogenous access, contractual credibility, strictly
positive heterogeneous privacy costs, and reveal-or-no-signal
identification. A conditional appendix bound relates threat value to the
quantity and quality of usable unilateral information under explicit
regularity assumptions, including restoration of the no-data continuation
when that information disappears. Numerical counterfactuals illustrate
provision incentives and distributional trade-offs; they are model-based
scenarios, not estimates of service-specific demand or causal effects.

\subsection{Literature Review}\label{sec:literature}
Our analysis connects research on data monetization, personalized-price
competition, and consumer privacy through the role of disclosure in the
outside options supporting platform data contracts. 

\textbf{Data value and platform monetization.} A growing body of literature examines how information changes decisions and competition, making its commercial value depend on how it is allocated and used \citep[e.g.,][]{bimpikis2019information,montes2019value,drakopoulos2023providing,wernerfelt2025estimating,guo2025selling,cong2026effects,ray2026bargaining}. For instance, \citet{bimpikis2019information} study information sales to competing firms, 
\citet{galperti2023privacy} show that privacy constraints can redistribute the marginal values of data records in an information-design problem, while \cite{sun2024value} empirically quantify the value of tracking data and data-informed decisions. Concurrently, related platform studies examine information
provision and monetization alongside downstream operations \citep[e.g.,][]{cosguner2017behavioral,choe2024softening,liu2025data,yao2025reputation,bhargava2026strategic,yan2026collect,markovich2026competing,zheng2026data}. 

\citet{montes2019value} provide a close contracting antecedent: information
prices depend on the payoff difference between acquiring data and facing
an informed rival, and the supplier compares exclusive with nonexclusive
sales. Their consumers incur costs to conceal information, whereas ours
incur costs when information is disclosed and becomes usable. Combined
with anonymous purchasing and noncommitment to personalized offers, this
difference can make unilateral access uninformative and restore the
rejecting seller's no-data payoff. We use this change in the disagreement
outcome to characterize the fees and provision incentives of a marketplace
that also earns transaction commissions.

\citet{hu2025information} already combine commission and information
revenue on a platform that competes with a third-party seller.
\citet{choe2025bright} study privacy management in a digital business with
usage-based revenue and an exogenous return per unit of shared data,
showing how anonymous purchasing and the revenue structure affect welfare.
Our platform solely intermediates between two independent sellers, and
its data-service revenue is determined by their adoption decisions and
disagreement payoffs. The distinguishing issue is how consumer control
changes the returns supporting data provision, rather than the presence
of two revenue sources itself.

\textbf{Personalized pricing and competition.} 
A classic result is that personalized pricing intensifies competition and can create a prisoner's dilemma: each firm prefers to personalize unilaterally, although both would benefit from uniform pricing \citep[e.g.,][]{thisse1988strategic,chen2001individual,choe2018pricing}. Subsequent work shows how coverage, differentiation, targeting accuracy, and consumer
heterogeneity change its profitability 
\citep{pazgal2008behavior,shin2010customer,rhee2017behavior,
amaldoss2019charm,rhodes2024personalized,li2024personalized,
li2025endogenous,du2025each}. Our contribution is complementary: holding product positions fixed, we study how an information-supplying platform  prices access to the resulting competitive environment. The threat effect is an explicit disagreement-payoff component of platform monetization, not a claim that defensive personalization is novel. The commissioned intermediary also differs from a stand-alone data broker because data use affects its operating revenue from product transactions.

A closely related literature examines how consumer data control shapes personalized pricing. Strategic privacy choices can benefit sellers by segmenting demand and softening competition \citep{chen2020competitive,kim2025behavior}, but may also intensify competition, induce price concessions, or degrade information quality \citep{ali2023voluntary,bimpikis2023data,li2023beating,cong2026effects}. Existing studies generally assume that sellers directly obtain consumer information or acquire it from an external data holder \citep{montes2019value,argenziano2026data}. We instead study sellers that depend on the marketplace for both customer access and the data capability needed to personalize. This platform dependence is precisely what makes seller-side data provision and monetization operational decisions rather than simply seller-side pricing choices.

\textbf{Consumer privacy, disclosure, and anonymity.}
Recently, a substantial literature examines how consumers respond strategically to personalized pricing by hiding, erasing, manipulating, or voluntarily disclosing information \citep{dinev2006extended,montes2019value,ichihashi2020online,fu2026privacy,ke2023privacy,kim2025behavior,li2025dark,lin2025choice,wu2025secondary}. 
\citet{chen2026personalized} provide the
closest unilateral-disclosure mechanism: a sole tracking firm cannot
attract privacy-sensitive consumers through personalized offers because
it cannot commit to compensating them for sunk privacy costs. We build
on this mechanism to determine platform access fees and provision. Their
inclusion of privacy-insensitive consumers also clarifies the scope of
our complete threat-elimination result, which maintains strictly positive
privacy costs throughout the population.

\citet{ali2023voluntary} show how verifiable disclosure can amplify
competition: relatively flexible consumers reveal information, while
those with stronger preferences pool. Their firm-specific disclosure
technology permits selective revelation to different sellers; our
consumers authorize a platform-controlled pricing capability and incur
intrinsic privacy costs. We connect the resulting disclosure choices to
upstream access contracts. \citet{li2023beating} instead study data
manipulation and its implications for informativeness, collection scope,
and transparency. In our baseline, consumers withhold verifiable
information rather than contaminate it, so the relevant margin is whether
a usable pricing signal becomes available.

Two further comparisons distinguish the welfare and operational
contributions. \citet{cong2026effects} analyze data management across
collection and application markets, including pricing externalities and
price commitments that protect firm profitability. Our anonymous-price
benchmark instead preserves consumers' no-data purchasing options within
a platform access-contracting environment. \citet{wang2026impact} examine
a platform that competes as a first-party seller and designs fuzzy
segments in response to anonymous users. Here, the platform does not sell
products, disclosure is endogenous, and the central decisions concern
data-service provision and access fees. These distinctions explain why
protecting sellers' rejection payoffs can eliminate defensive monetization
without eliminating consumer harm from segmentation.

{\bf Roadmap.} The remainder of the paper is organized as follows. Section~\ref{sec:model} presents the marketplace model and data-value
framework. Section~\ref{sec:equilibrium} derives equilibria and stakeholder
outcomes across governance regimes. Section~\ref{sec:extension} examines
extensions and model-based counterfactuals, and
Section~\ref{sec:conclusion} concludes with managerial and policy
implications. Technical details and proofs appear in the appendix.

\section{Model and The General Value of Platform Data}\label{sec:model}
\textbf{Operational setting and data-governance regimes.}
A platform connects a unit mass of consumers with two horizontally differentiated sellers, indexed by $i\in\{1,2\}$. It earns an exogenous commission rate $\gamma\in[0,1)$ on product sales. Throughout, \emph{sales revenue} means gross product revenue before commissions and access fees. A seller's net profit is its retained share $(1-\gamma)$ of this revenue minus its data-service payment. The platform's profit combines commissions and access fees; data-service operating costs are normalized to zero in the baseline. Keeping these quantities distinct is essential: an increase in seller revenue need not increase the seller's net profit or social welfare.

The platform controls consumer information that sellers cannot independently reproduce in the baseline. It can provide a data service that enables subscribing sellers to infer individual preferences and personalize prices. Nonsubscribers can offer only uniform prices. The service is a platform-controlled capability, not necessarily a sale of raw records: an audience interface, permissioned analytics environment, API, or platform-executed offer could implement the required information access. The maintained economic restriction is that sellers lack a comparable independent source of the price signal. This is most relevant to sellers dependent on the marketplace for both customers and targeting information. 

The data-governance regime is exogenous; the platform optimizes provision and contracts within that regime. Under \emph{no data collection (NDC)}, consumer information cannot be used for the seller-facing service, so sellers rely on anonymous uniform pricing. Under \emph{mandatory data collection (MDC)}, data from entering consumers are available for the service and consumers cannot preserve a meaningful anonymous purchasing channel. This regime also represents take-it-or-leave-it data collection in which continued marketplace participation effectively requires data use. Under \emph{voluntary data disclosure (VDD)}, consumers can authorize this use or decline it while continuing to purchase through the uniform-price channel. ``Anonymous'' means unidentifiable for the modeled personalized-price decision; it does not mean that a marketplace can process payment or fulfill orders without any operational information. When the data service is withheld or neither seller adopts, no pricing-related privacy cost is incurred and the NDC outcome applies.

\noindent\textbf{Research objectives.}
We examine three linked decisions: when a platform should provide a seller-facing data service; how access fees and commissions translate downstream data use into platform profit; and how consumer disclosure rights change seller adoption, market coverage, and the distribution of surplus. The first objective distinguishes seller-revenue improvement from defensive willingness to pay, helping managers assess whether demand for a data product indicates better marketplace performance. The second identifies when the platform's data business supports or undermines its transaction business. The third evaluates anonymous purchasing as an operational design choice, separating protection against an informed rival from changes in trade and privacy losses. These objectives require joint analysis of platform contracts, seller pricing, and consumer choices rather than a fixed per-record return or an exogenous disclosure rate.

\noindent\textbf{Sequence of events and information.}
The game separates marketplace access rules from subsequent pricing and consumer decisions. The platform first chooses whether to offer the service. If offered under regime $g\in\{\MDC,\VDD\}$, it publicly commits to an offer order, a history-dependent fee schedule, and the access configuration following each acceptance or rejection. Each seller is first with probability one-half. Labeling the realized offerees 1 and 2, the schedule is $F^g=(F_1^g,F_{2|AC}^g,F_{2|RE}^g)$, where seller 1 decides first, and seller 2 observes that decision and receives the applicable offer. All fees may be real-valued. In particular, a negative rejection-contingent fee subsidizes activation of exclusive access. Bilateral adoption is the baseline equilibrium-path configuration; Section~\ref{sec:endogenous_exclusive_access} also allows exclusive provision on the equilibrium path. Figure~\ref{fig:sequence} summarizes the subsequent decisions.

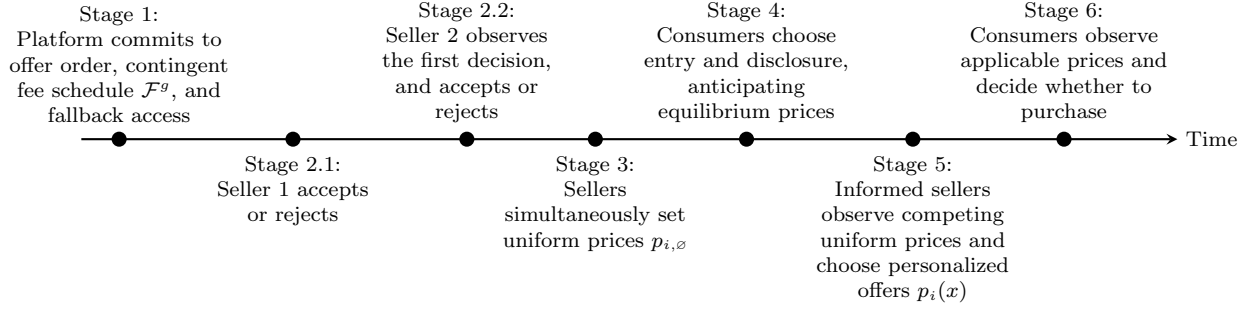
\begin{figure}[!ht]
\vspace{-0.15in}
\scriptsize 
\caption{Sequence of Events under VDD}\label{fig:sequence}

\begin{tikzpicture}[>=stealth,thick]

\draw[->] (0,0) -- (14.5,0) node[right]{Time};

\node[circle,fill=black,inner sep=2pt,
 label=above:{\parbox{2.9cm}{\centering Stage 1:\\ Platform commits to offer order, contingent fee schedule $\mathcal F^g$, and fallback access}}] 
(s1) at (0.5,0) {};

\node[circle,fill=black,inner sep=2pt,
 label=below:{\parbox{2.3cm}{\centering Stage 2.1:\\ Seller 1 accepts or rejects}}] 
 (s2a) at (2.8,0) {};

\node[circle,fill=black,inner sep=2pt,
 label=above:{\parbox{2.3cm}{\centering Stage 2.2:\\ Seller 2 observes the first decision, and accepts or rejects}}]
 (s2b) at (5.1,0) {};

\node[circle,fill=black,inner sep=2pt,
 label=below:{\parbox{2.8cm}{\centering Stage 3:\\ Sellers simultaneously set uniform prices $p_{i,\varnothing}$}}] 
 (s3) at (6.8,0) {};

\node[circle,fill=black,inner sep=2pt,
 label=above:{\parbox{2.8cm}{\centering Stage 4:\\ Consumers choose entry and disclosure, anticipating equilibrium prices}}] 
 (s4a) at (8.8,0) {};

\node[circle,fill=black,inner sep=2pt,
label=below:{\parbox{3cm}{\centering Stage 5:\\ Informed sellers observe competing uniform prices and choose personalized offers $p_i(x)$}}] 
 (s4b) at (11,0) {};

\node[circle,fill=black,inner sep=2pt,
label=above:{\parbox{3cm}{\centering Stage 6:\\ Consumers observe applicable prices and decide whether to purchase}}] 
 (sf) at (13,0) {};

\end{tikzpicture}
\vspace{-0.2in}
\end{figure}

Consumers do not observe realized uniform prices before entry and disclosure. Informed sellers subsequently observe the rival's realized uniform price when choosing personalized offers, as in the Appendix's best-response calculations. A hidden uniform-price deviation therefore leaves entry and disclosure fixed, while personalized offers can respond in asymmetric-access continuations. Rational expectations require anticipated and realized equilibrium prices to coincide. Sellers cannot commit to personalized offers before information becomes usable. Consent after observing posted anonymous prices requires a different analysis. Sequential uniform and personalized pricing is related to \cite{chen2026personalized,montes2019value,hu2025information}; the precise information sets above are maintained assumptions, not claims about universal marketplace practice. 

Full implementation requires honoring fees and access rights after every history, including subsidized rejection continuations. Sellers accept when indifferent, and consumers observe active access. Backward induction solves contracting; pricing and consumer choices satisfy sequential optimality and rational-expectations consistency. The resulting fees exhaust participation constraints under this protocol, providing an extraction benchmark rather than an unrestricted optimal mechanism or a description of observed marketplace contracts. 

\noindent\textbf{Two sources of data value.} Summarize the selected downstream continuation under regime $g$ by four gross-revenue states: $R_N^g$ when neither seller adopts the service; $R_B^g$ when both adopt; $R_I^g$ for the sole adopter under exclusive access; and $R_U^g$ for its nonadopting rival. These states incorporate entry, disclosure, and pricing responses and exclude commissions and fees. Define
\begin{equation}
\text{\emph{sales-revenue effect} } G^g := R_B^g-R_N^g,
\quad
\text{and \emph{competition-threat effect} } T^g := R_N^g-R_U^g.
\label{eq:GT}
\end{equation}
The \emph{sales-revenue effect} $G^g$ measures whether bilateral data use increases or lowers each seller's product revenue relative to the no-data outcome. It captures productive gains from better targeting or transaction creation, but it may be negative when personalization intensifies price competition or discourages participation. The \emph{competition-threat effect} $T^g$ measures the seller's revenue loss in the exclusive access continuation following rejection. This loss may arise from being the uninformed seller when the rival alone has usable data (competitive channel), or from reduced consumer participation (participation channel). If $T^g>0$, rejecting access can be costly even when bilateral use lowers both sellers' revenues. It is therefore a disagreement-payoff effect: a seller may value access not because bilateral data use is profitable, but because refusing access is worse.

This distinction is operationally important. Seller adoption is not, by itself, evidence that a data product creates downstream revenue. A marketplace may observe strong willingness to pay even when bilateral personalization reduces seller sales, provided that asymmetric access would reduce the rejecting seller's revenue even further. Conversely, data may improve seller revenue but be difficult for the platform to monetize if sellers can obtain equivalent information elsewhere or if governance prevents asymmetric informational advantage. The four-state representation makes these possibilities explicit without relying on the Hotelling pricing environment used later.

\subsection{Seller Adoption and Platform Data Value}\label{subsec:value}
We first retain the full-implementation benchmark of the specified contracting protocol. It identifies the fees consistent with bilateral adoption given the selected revenue states and promised rejection continuations.

\begin{lemma}[Seller adoption under full implementation]\label{lem:fees}
If the platform can commit to the announced rejection continuation, it can induce both sellers to adopt by charging: $$F_1^g=F_{2\mid AC}^g=(1-\gamma)(G^g+T^g), \quad F_{2\mid RE}^g=(1-\gamma)(R_I^g-R_N^g).$$
\end{lemma}

Lemma~\ref{lem:fees} clarifies why seller willingness to pay can exceed the direct revenue improvement from bilateral data use. The on-path fee includes both the bilateral revenue change and the rejection loss. The off-path fee makes the second seller willing to activate exclusive access after the first rejects. When this requires a subsidy, commitment to pay it is part of implementation, even if the continuation generates little trade.

\begin{proposition}[Platform data value]\label{prop:data_value_decomposition}
Under full implementation, the platform's incremental profit from providing the data service rather than withholding it is:
\begin{equation}
    \begin{aligned}\label{eq:data_value_decomposition}
V_{\mathrm{data}}^{g}
&\;=\;\underbrace{2\gamma G^g}_{\text{commission-based value}}
+\underbrace{
\left[2(1-\gamma)G^g+2(1-\gamma)T^g\right]}_{\text{data monetization value}} \;=\;2G^g+2(1-\gamma)T^g.
\end{aligned}
\end{equation}
\end{proposition}



The decomposition provides a practical diagnostic for identifying the source of data monetization. The sales-revenue effect enters through commissions and fees, while the threat effect supports fees through sellers' retained revenue. A platform may profitably offer a data service even when bilateral data use reduces seller revenue ($G^g<0$), provided sellers are sufficiently willing to pay to avoid an asymmetric informational disadvantage ($
(1-\gamma)T^g>-G^g$). Adoption does not establish improved revenue, and $G^g>0$ need not imply higher welfare if it comes from repricing existing trade. Under Lemma~\ref{lem:fees}, each seller retains $(1-\gamma)R_U^g$. When $T^g=0$, the platform extracts the bilateral revenue gain and seller net profit equals the no-data benchmark.

We next relax commitment only when the announced rejection continuation requires $F_{2|RE}^g<0$. In that event, the platform may revise the continuation to no access; all other offers and access rights remain enforceable. Let $\eta^g\in[0,1]$ be seller 1's ex ante perceived probability that the platform honors the subsidy. Define $\xi^g(\eta^g)=1$ if $R_I^g\ge R_N^g$, and $\xi^g(\eta^g)=\eta^g$ otherwise. Thus, a continuation requiring no subsidy remains enforceable. Here credibility is a reduced-form primitive and service costs remain zero. Appendix~\ref{app:sec_commitment_microfoundation} provides the implementation-cost microfoundation.

\begin{proposition}[Data value under limited commitment]
\label{prop:limited_commitment}
With limited implementation credibility:
{\rm (A)} The platform can induce bilateral adoption with $F_1^g=(1-\gamma)[G^g+\xi^g(\eta^g)T^g]$, $F_{2\mid AC}^g=(1-\gamma)(G^g+T^g)$, and $F_{2\mid RE}^g=(1-\gamma)(R_I^g-R_N^g)$. {\rm (B)} The platform's data value is $V_{data}^g(\eta^g)=2G^g+[1+\xi^g(\eta^g)](1-\gamma)T^g$. {\rm (C)} When a subsidy is required ($R_I^g<R_N^g$), the data value weakly increases with the credibility if $T^g\ge 0$ and decreases otherwise.
\end{proposition}

Limited commitment operates through seller 1's rejection payoff, which becomes $(1-\gamma)[R_N^g-\xi^g(\eta^g)T^g]$. Once seller 1 accepts, seller 2 faces an informed rival upon rejection regardless of subsidy credibility. When $T^g>0$, the first offeree pays weakly less and retains weakly more profit. The randomized offer order preserves ex ante symmetry.
Credibility affects extraction from the given revenue states, separating the \emph{existence} of competitive vulnerability from the platform's \emph{ability to monetize} it. In the subsidy region, greater credibility increases data value when $T^g>0$. Hence, a platform seeking to monetize the threat effect has an incentive to strengthen commitment through enforceable arrangements, such as escrow payments, third-party audits, or verifiable access records. Appendix~\ref{sec:threat_elimination} further provides a conditional bound on competition-threat value, clarifying when privacy or access governance can eliminate it.

\subsection{Hotelling Microfoundation}\label{subsec:hotelling}
The remainder of the paper leverages a Hotelling microfoundation to endogenize revenue states, consumer disclosure, prices, and market coverage, thereby establishing a tractable link between the general data-value framework and observable marketplace outcomes. Consumers have unit demand and preferences $x\sim U[0,1]$. Sellers 1 and 2 are located at 0 and 1. Purchasing from them incurs mismatch costs $tx$ and $t(1-x)$, respectively. All consumers have product valuation $v>0$, and marginal production costs are zero. Dividing all monetary variables by $t>0$ normalizes the mismatch parameter to one; $v$ and the privacy cost below are therefore normalized monetary quantities. Let $\delta(x)$ indicate whether consumer $x$'s information is used in the pricing subgame. The baseline privacy cost is $c>0$, incurred once when that information becomes usable, regardless of which seller is eventually chosen. A seller with usable data charges $p_i(x)$; otherwise it charges $p_{i,\varnothing}$. Writing $p_i$ for the applicable price, the consumer's ex ante net utilities from purchases are
\[
U_1(x)=v-p_1-x-c\delta(x),
\qquad
U_2(x)=v-p_2-(1-x)-c\delta(x).
\]
Staying outside the gives zero.  Once information has been used, privacy cost is sunk: not purchasing gives $-c\delta(x)$ at the purchase stage. Consumers therefore compare expected net entry utility with zero initially and applicable offers with the sunk-cost no-purchase alternative subsequently. Positive privacy cost and the absence of advance personalized-price commitment are central to the unilateral-disclosure incentive.

Consumers enter and purchase when indifferent. Under VDD they opt in when indifferent between disclosure and nondisclosure. When both products yield the same maximal purchase utility and exactly one seller has usable information about that consumer, she chooses that seller. Otherwise, she chooses the lower-mismatch seller, with the midpoint tie assigned to seller 1. We study pure-strategy rational-expectations equilibria, imposing symmetry for symmetric access. If multiple such pricing equilibria exist, we apply the appendix's selection: maximal gross revenue for each seller under bilateral access and maximal gross revenue for the uninformed seller under exclusive access. Access fees are sunk, so this also maximizes the corresponding continuation profit. The selected equilibrium determines reported revenues and provision regions. The alternatives in Appendix~\ref{app:sec_MDC_subgame} and Appendix~\ref{app:sec_VDD_subgame} preserve $T^{MDC}>0$ and $T^{VDD}=0$, while revenues and some provision decisions can change. In particular, the high-valuation segmentation outcome below depends on selection.

Market coverage is the mass of consumers who purchase, not the mass who enter or disclose. Consumer surplus integrates realized utility net of prices and any incurred privacy costs. In the zero-cost baseline, social welfare equals consumer surplus plus platform profit and both sellers' net profits; commissions and access fees cancel as transfers. These definitions separate transaction creation from changes in revenue or its distribution.

\subsection{Scope, Extensions, and Limitations}\label{sec:scope}
The Hotelling analysis below uses a deliberately parsimonious baseline to isolate the link between consumer disclosure, asymmetric data access, and platform monetization. It assumes homogeneous privacy costs, perfect identification, bilateral access on the equilibrium path and full implementation of announced access continuations. Section~\ref{sec:extension} relaxs the first three assumptions by introducing two types of privacy costs, a reveal-or-no-signal information structure and endogenous exclusive data provision. The reveal-or-no-signal information structure keeps consumers anonymous after failed identification and charges the privacy cost only upon successful identification; it does not serve as a general model for noisy or biased predictions. Appendix~\ref{app:sec_commitment_microfoundation} endogenizes the limited-commitment friction introduced above through implementation costs and the cost of reneging on a promised continuation. These extensions differentiate changes in downstream revenue states from the platform's capacity to extract defensive payments. However, the qualitative results do not suggest that all baseline thresholds or welfare rankings withstand arbitrary heterogeneity, information structures, or commitment frictions.

Several constraints remain. The platform remains an intermediary, with symmetric fixed-position sellers, exogenous commission, and no independent signals or direct nonprice benefits from disclosure. The introduction of first-party information, inference about anonymous buyers, costly opt-outs, service improvements, and advance price commitments may alter unilateral disclosure or segmentation and require separate downstream analysis.

\section{Product-Market Equilibrium}\label{sec:equilibrium}
This section characterizes and compares equilibrium outcomes across three data-governance regimes, revealing their impact on seller adoption, platform monetization, market participation, and welfare. All proofs are in Online Appendix~\ref{app:sec_proof_main}.

\subsection{No Data Collection: Benchmark}\label{sec:no_data}
Under NDC, both sellers choose uniform prices simultaneously $p_{1,\varnothing},p_{2,\varnothing}$. Seller 1 serves consumers satisfying  $0\le x\le \min\{v-p_{1,\varnothing},\tfrac{1+p_{2,\varnothing}-p_{1,\varnothing}}{2}\}$, with a symmetric condition for seller 2. Let $r_i^{\NDC}$ denote gross sales revenue, $D^{\NDC}$ market coverage, $\Pi_p^{\NDC}$ platform profit, $CS^{\NDC}$ consumer surplus, and $SW^{\NDC}$ social welfare.

\begin{proposition}[NDC benchmark]\label{prop:without_data}
Under NDC, all consumers enter the platform:
\[
(p_{i,\varnothing}^{\NDC},r_i^{\NDC},\Pi_p^{\NDC},CS^{\NDC},SW^{\NDC})=
\begin{cases}
\left(\tfrac{v}{2},
\tfrac{v^2}{4},
\tfrac{\gamma v^2}{2},
\tfrac{v^2}{4},
\tfrac{3v^2}{4}\right),
&{\rm if \ } v< 1\\
\left(v-\tfrac12,
\tfrac{1}{2}v-\tfrac{1}{4},
\gamma( v-\tfrac{1}{2}),
\tfrac14,
v-\tfrac14\right),
&{\rm if \ } 1\le v<\tfrac32\\
\left(1,
\tfrac{1}{2},
\gamma,
v-\tfrac54,
v-\tfrac14\right),
&{\rm if \ } v\ge\tfrac32
\end{cases}.
\]
Seller profit is $\pi_i^{\NDC}=(1-\gamma)r_i^{\NDC}$ and market coverage is $D^{\NDC}=\min\{v,1\}$.
\end{proposition}

Entry is not purchase: free anonymous entry is weakly optimal even for consumers who eventually do not buy. When $v<1$, uniform pricing leaves central consumers unserved, so data-enabled offers may create additional transactions. For $1\le v<3/2$, the market is covered but marginal participation constrains prices. For $v\ge3/2$, competition pins the uniform price at one. These cases distinguish the possibility of trade creation from repricing an already covered market.

\subsection{Mandatory Data Collection: Monetizing Competitive Vulnerability} \label{sec:no_protection}
The adoption subgames are NN (neither adopts), PP (both adopt), PN (only seller 1 adopts), and NP (only seller 2 adopts). Let $r^{PPM}$ denote each seller's revenue under PP and $r_2^{PNM}$ the uninformed seller's revenue under PN. By symmetry and \eqref{eq:GT}, the sales-revenue effect is $G^{\MDC}=r^{PPM}-r^{\NDC}$, and the threat effect is $T^{\MDC}=r^{\NDC}-r_2^{PNM}$. 
\begin{proposition}[MDC equilibrium]\label{prop:equilibrium_mdc}
Under MDC: 
\begin{itemize}
    \item[\rm{(A)}] If the platform provides the data service, both sellers adopt and the equilibrium-path fees are $F_1^{\MDC}=F_{2|AC}^{\MDC} = (1-\gamma)(r^{PPM}-r_2^{PNM})$;
    \item[\rm{(B)}] $G^{\MDC}<0$ and $T^{\MDC}>0$, the platform provides the data service if and only if $\gamma\le 1+G^{\MDC}/T^{\MDC}$;
    \item[\rm{(C)}] Conditional on monetization, consumer entry occurs for $x\in\Omega^{\MDC}$, where $\Omega^{\MDC}=[\max\{0,1-v+c\},\min\{1,v-c\}]$
    and personalized prices are
    $    p_1^{\MDC}(x)=\max\{1-2x,0\}$ and $
    p_2^{\MDC}(x)=\max\{2x-1,0\}$.
\end{itemize}  
\end{proposition}

Mandatory data use lowers bilateral seller revenue, but the revenue lost after rejecting access can be larger still. Under the selected exclusive-access continuation, Appendix \ref{app:sec_MDC_subgame} shows that the rejection loss is participation-driven when $v-c<3/4$, while it arises from personalized poaching for $v-c\ge3/4$. In the former case, exclusive access need not benefit either seller: the platform sells avoidance of an unfavorable continuation. Lower commissions increase the retained revenue at stake and strengthen the defensive fee component. Data provision can consequently benefit the platform even as it reduces the commission base.

\begin{proposition}[MDC stakeholder effects]\label{prop:profit_campare_M_N}
Whenever the platform monetizes data under MDC, platform profit weakly increases relative to NDC, while market coverage weakly decreases. Consequently, seller profit and social welfare are weakly reduced. The consumer surplus can rise only in the regions in which stronger personalized-price competition dominates privacy and participation losses: $1<v<\tfrac{3}{2}$ and $v-c>1$, or $v\ge \tfrac{3}{2}$ and $c<\tfrac{1}{2}$.
\end{proposition}

Proposition~\ref{prop:profit_campare_M_N} reveals a tension between platform monetization and marketplace performance. MDC cannot expand coverage relative to NDC in this model. Low personalized prices can attract central consumers, but mandatory privacy losses can drive other consumers out. The platform benefits only because it provides the service when doing so is profitable. Neither this platform gain nor bilateral adoption establishes improved marketplace performance. Prices and access fees redistribute surplus; privacy losses and lost transactions are real welfare costs.

\subsection{Voluntary Disclosure: Eliminating Threat-Based Adoption}\label{sec:vdd}
Under VDD, consumers retain anonymous purchasing. Define $r^{PPV}$ as bilateral seller revenue and $r_2^{PNV}$ as the uninformed seller's revenue under unilateral adoption. Then, analogous to MDC, the sales-revenue effect is $G^{\VDD}=r^{PPV}-r^{\NDC}$ and the threat effect is $T^{\VDD}= r^{\NDC}-r_2^{PNV}$. We use the following provision convention: the platform withholds the service when the opt-in set is empty and otherwise provides it when incremental profit is nonnegative. The boundary cases with only a singleton opt-in set are therefore retained, although they generate no positive mass of data or strictly positive data value. 
\begin{proposition}[VDD equilibrium and threat elimination]
\label{prop:equilibrium_vdd}
Under VDD:
\begin{itemize}
    \item[{\rm(A)}] If the platform provides the service, both sellers adopt and the equilibrium-path fees are
    $F_1^{\VDD}=F_{2\mid AC}^{\VDD}=(1-\gamma)(r^{PPV}-r^{\NDC})$;
    \item[{\rm(B)}] The threat effect $T^{\VDD}=0$, the platform provides the service if and only if the opt-in set is nonempty and the sales-revenue effect is non-negative $(G^{\VDD}\ge 0)$, which occurs when $(v,c)\in \mathcal P_1\cup \mathcal P_2\cup \mathcal P_3$, where the regions $\mathcal P_i$ are defined in the proof; 
    \item[{\rm(C)}] Conditional on provision, for all $i=1,2$, the uniform prices are
\[
p_{i,\varnothing}^{\VDD}=
\begin{cases}
\max\{\tfrac v2,2v-c-1\},&{\rm if\ } (v,c)\in \mathcal P_1\cup\mathcal P_3,\\
p_{i,\varnothing}^{\NDC},&{\rm if\ } (v,c)\in \mathcal P_2,
\end{cases}
\]
and the opt-in set is 
$$\Omega^{\VDD}=\begin{cases}
    \left[
\max\left\{0,1-v+c\right\},\,
\min\left\{1,v-c\right\}
\right],& {\rm if\ } (v,c)\in \mathcal P_1\cup \mathcal P_3\\
\{\tfrac{1}{2}\},& {\rm if\ } (v,c)\in \mathcal P_2
\end{cases}.$$
The opt-in consumers receive $p_1^{\VDD}(x)=\max\{1-2x,0\}$ and
$p_2^{\VDD}(x)=\max\{2x-1,0\}$. All other consumers opt-out and face the uniform prices above. For $x\in\Omega^{\VDD}$, $p_i^{\VDD}(x)<p_{i,\varnothing}^{\VDD}$.
\end{itemize}
\end{proposition}

The central mechanism operates in the unilateral-adoption subgame, not just on the bilateral-adoption path. With only one informed seller, disclosure does not improve available purchase surplus: after information becomes usable, that seller can price against the uninformed alternative. Disclosure also incurs the strictly positive privacy cost. No consumer therefore discloses, anonymous purchasing preserves participation, and both sellers face the no-data revenue state. This restores the rejecting seller's payoff and eliminates both channels of $T^{\VDD}$. The platform's data value becomes $V_{\data}^{\VDD}=2G^{\VDD}$, independent of $\gamma$ in the zero-cost, fixed-commission benchmark. Provision is aligned with seller gross revenue, full fee extraction and the separate welfare accounting remain essential.

Under bilateral access, consumers near the midpoint may obtain enough personalized competition to disclose, while more strongly matched consumers prefer anonymity. The provision regions in Figure~\ref{fig:data_provision} have distinct economic meanings. In the interior of $\mathcal{P}_1$, personalization reaches consumers unserved under NDC and anonymity retains endpoint demand. At its zero-mass boundary and throughout $\mathcal{P}_2$, only the midpoint consumer discloses at a zero personalized price, leaving aggregates at their NDC values.  In $\mathcal{P}_3$, NNDC already covers the market, so higher revenue comes from segmentation and higher anonymous prices. This last outcome depends on the maintained revenue-maximizing selection: Remark~\ref{app:rem_vdd_PP} in Appendix~\ref{app:sec_VDD_subgame} describes a coexisting nondisclosure equilibrium that restores NDC outcomes. Threat elimination survives this multiplicity, and positive-revenue segmentation is not common to every equilibrium.

\begin{figure}[!ht]
    \centering
     \caption{Regimes of Data Monetization ($\gamma=0.15$)}
     \includegraphics[width=0.52\textwidth]{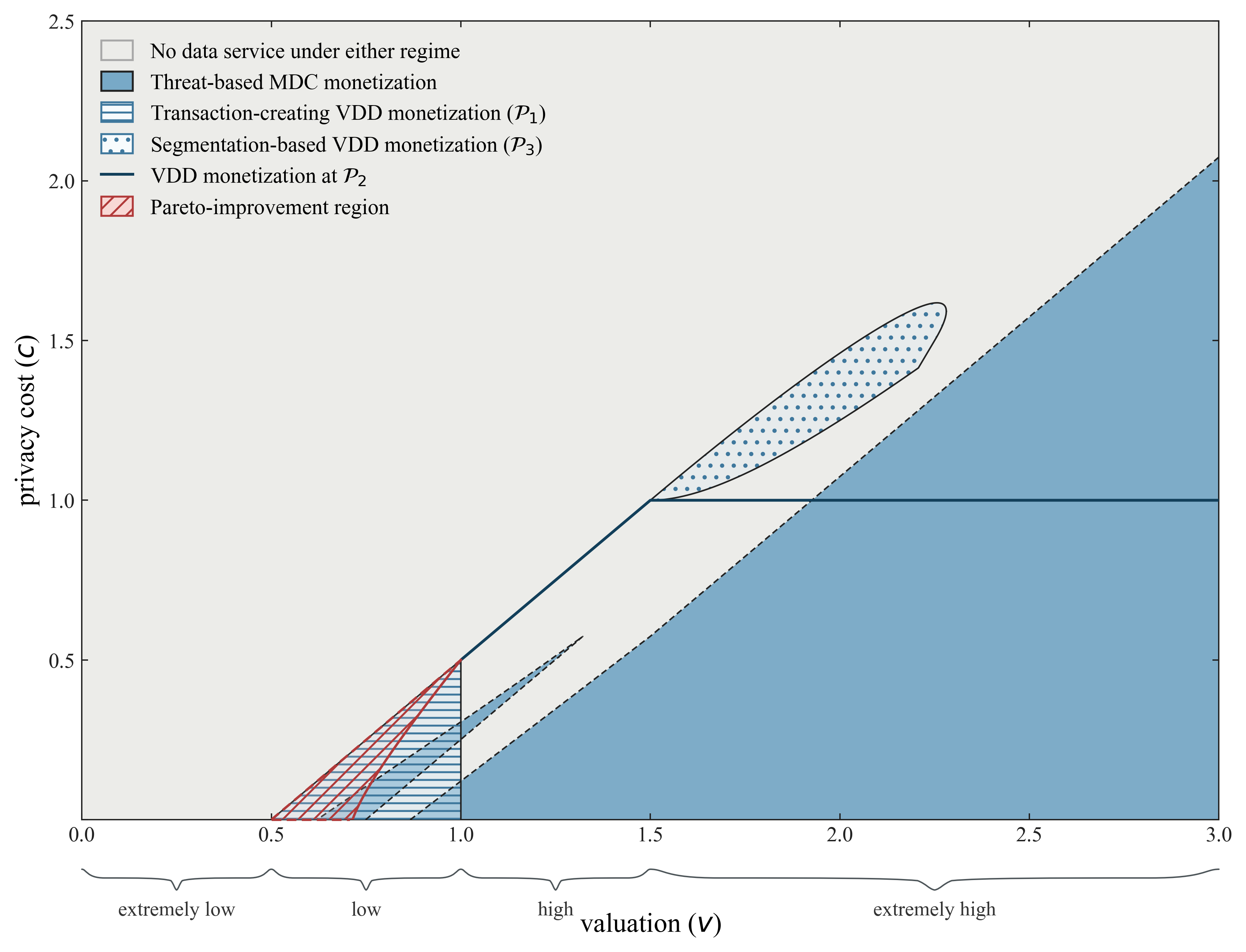}
    \label{fig:data_provision}
    \footnotesize
    \begin{flushleft}
    \footnotesize
    \vspace{-0.1in}
        \textit{Notes.} The ``Pareto-improvement region" compares the platform, each seller, and consumers in aggregate.
    \end{flushleft}
    \vspace{-0.15in}
\end{figure}

\subsubsection{Prices, Coverage, Profits, and Welfare} We now examine the effects of voluntary data disclosure regulations, including the impact on prices and market coverage, profits of platform and sellers, consumer surplus, and social welfare. Withholding the service reproduces NDC. The following comparisons state explicitly when provision is required under one or both regimes.

\begin{proposition}[Impacts on Price and Market Coverage]\label{prop:price_coverage_cmp} Conditional on data provision under VDD and MDC:  
\begin{enumerate}
    \item [{\rm (A)}] VDD opt-in consumers receive the same personalized prices under both regimes 
    ($p_i^{\VDD}(x)=p_i^{\MDC}(x), i\in\{1,2\}$).  If $(v,c)\notin\{(v,c)|v>3/2,c=1\}$, VDD opt-out consumers would not enter under MDC. If $v>3/2,c=1$, only the consumer at $x=\tfrac12$ opts in under VDD, whereas the MDC entry set is $\Omega^{\MDC}=[\max\{0,2-v\},\min\{1,v-1\}]$. 
    \item [{\rm (B)}] Relative to NDC, $p_{i,\varnothing}^{\VDD}\ge p_{i,\varnothing}^{\NDC}$, strictly on
$\{(v,c):2/3<v\le1,\ c<3v/2-1\}\cup\mathcal{P}_3$.
For opt-in consumers, $p_i^{\VDD}(x)\le p_{i,\varnothing}^{\NDC}$ except when $2/3<v\le1$ and $0<c<3v/4-1/2$. In that case, seller 1's personalized price exceeds its NDC price for $x\in[1-(v-c),(2-v)/4)$, and seller 2's exceeds its NDC price for $x\in((2+v)/4,v-c]$; all other opt-in consumers face weakly lower prices.
    \item[{\rm (C)}] Market coverage satisfies $D^{\VDD}=\min\left\{D^{\NDC}+D^{\MDC},\,1\right\}$.
\end{enumerate}
\end{proposition}

Voluntary disclosure changes who remains available for uniform pricing. The anonymous group is more concentrated among strongly matched consumers, which can raise its uniform price. Some consumers may then disclose their purchase because personalization is cheaper than the \emph{new} anonymous price, even though it is more expensive than the \emph{NDC} price. Figure~\ref{fig:price_coverage_mdc_ndc_vdd} illustrates this distinction. Individual optimality within the VDD therefore does not establish a gain relative to a different governance regime.

\begin{figure}[!ht]
\vspace{-0.14in}
    \centering
    \caption{Equilibrium Prices ($v=0.8,c=0.06,\gamma=0.15$)}\label{fig:price_coverage_mdc_ndc_vdd}
     \includegraphics[width=0.5\textwidth]{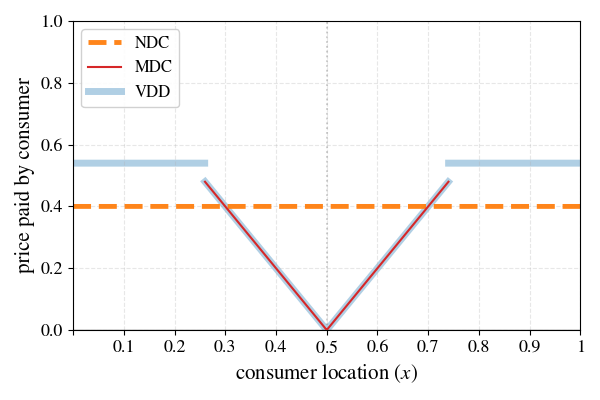}
    \begin{flushleft}
    \vspace{-0.14in}
        \footnotesize
        \textit{Notes.} MDC and VDD prices are conditional on provision, and consumers outside the MDC entry set receive no MDC offer.
    \end{flushleft}
\vspace{-0.25in}
\end{figure}

The coverage identity integrates complementary demand sources.  For $v<1$, NDC sustains anonymous endpoint purchases but leaves central consumers unserved.  MDC can serve central consumers, by offering personalized competition, but loses others due to privacy costs. VDD accommodates both groups, as Figure~\ref{fig:market_coverage} illustrates. In the high-valuation region, full coverage instead means preserving the number of transactions while changing prices and privacy losses.

\begin{figure}[!ht]
\vspace{-0.1in}
    \centering
     \caption{Market Coverages under Alternative Governance Regimes when $v<1$}\label{fig:market_coverage}
     \includegraphics[width=\textwidth]{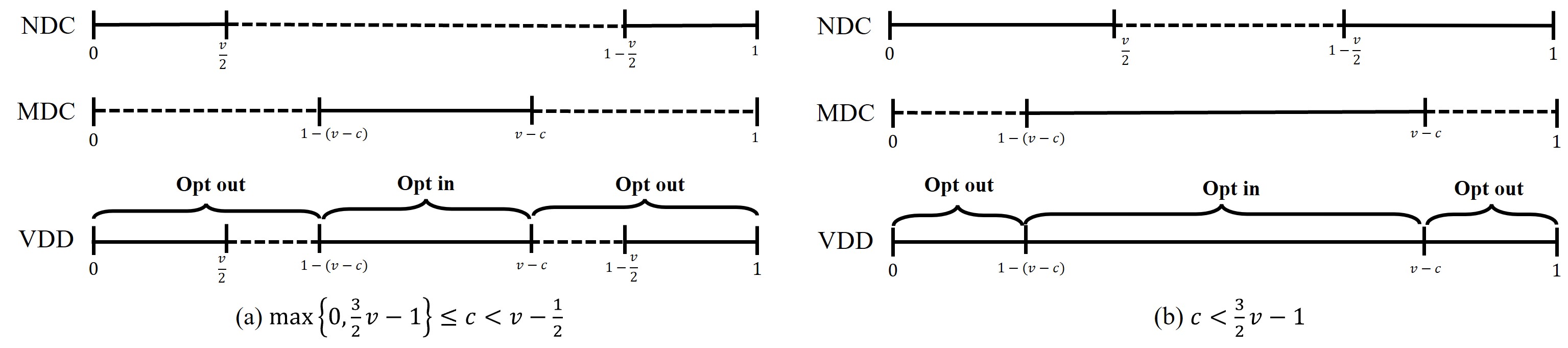}
    \begin{flushleft}
        \footnotesize
        \textit{Notes.} Conditional on provision, solid segments represent purchases; dashed segments represent no purchase.
    \end{flushleft}
    \vspace{-0.25in}
\end{figure}

\begin{proposition}[Impacts on Profits]\label{prop:profit_campare_vdd}
Conditional on data provision under VDD: 
{\rm{(A)}} $\pi_i^{\VDD} = \pi_i^{\NDC} \ge \pi_i^{\MDC}$ for $i\in \{1,2\}$;
{\rm{(B)}} $\Pi_p^{\VDD}\ge \Pi_p^{\NDC}$; {\rm{(C)}}  $\Pi_p^{\VDD}\ge\Pi_p^{\MDC}$ if and only if 
\[
\gamma\ge \gamma^{VM}(v,c):= \max\{0,1-\tfrac{r^{PPV}-r^{PPM}}{T^{\MDC}}\}.
\]
\end{proposition}

VDD protects the seller's disagreement payoff. Because the platform extracts the entire bilateral revenue gain through the fee, seller net profits equal NDC profits; they do not generally increase above that benchmark. The platform's comparison with MDC depends on the revenue gains retained under VDD relative to the defensive fees available under MDC. A higher commission reduces the latter advantage, producing the threshold in Proposition~\ref{prop:profit_campare_vdd}. The platform need not voluntarily prefer effective consumer control at a low commission, so a privately profitable commitment to honor opt-out is not automatic.

Consumer surplus and welfare require separate comparisons with NDC and MDC. The former asks whether data use itself improves outcomes; the latter isolates the consequences of disclosure rights when both regimes use data.

\begin{proposition}[Impacts on Consumer Surplus and Social Welfare]\label{prop:cs_sw_cmp} Conditional on data provision under VDD:
\begin{enumerate}
    \item[\rm(A)] Consumer surplus satisfies \(CS^{\VDD}\ge CS^{\NDC}\) on 
    $(v,c)\in\{(v,c): \tfrac{1}{2}<v\le\tfrac{\sqrt{2}}{2},c\le v-\frac{1}{2}\}\cup  \{(v,c): \tfrac{\sqrt{2}}{2}<v\le 1,v-\tfrac{3}{4}+\tfrac{\sqrt{2v^2-1}}{4}\le c\le v-\frac{1}{2}\}\cup \mathcal P_2 $, and \(CS^{\VDD}<CS^{\NDC}\) otherwise. Similarly, social welfare satisfies \(SW^{\VDD}\ge SW^{\NDC}\) on $(v,c)\in\{(v,c):\frac{1}{2}<v\le\frac{6+\sqrt{6}}{10},c\le v-\frac{1}{2}\} \cup \{(v,c):\tfrac{6+\sqrt{6}}{10}<v\le 1,c\in(0,\tfrac{2v-1-\sqrt{10v^2-12v+3}}{4}]\cup[\tfrac{2v-1+\sqrt{10v^2-12v+3}}{4},
    v-\tfrac{1}{2}]\} \cup \mathcal P_2$, and \(SW^{\VDD}<SW^{\NDC}\) otherwise.
    \item[\rm(B)] If the service is also provided under MDC, $CS^{\VDD}>CS^{\MDC}$ and $SW^{\VDD}>SW^{\MDC}$.
\end{enumerate}
\end{proposition}

Additional trade is the main source of welfare gains over NDC. When baseline coverage is sufficiently low, serving central consumers can outweigh privacy losses. As $v$ approaches one, fewer incremental transactions remain available, and the higher anonymous price can reduce consumer surplus. The welfare effect can be nonmonotone in $c$: low privacy costs make disclosure inexpensive, while sufficiently high costs discourage disclosure and limit aggregate privacy losses; intermediate costs can be most problematic when the remaining coverage gain is small.

The high-valuation segmentation region provides a useful boundary to the main mechanism. In $\mathcal{P}_3$, the market was already covered under NDC, so higher seller revenue does not reflect additional trade. Anonymous-price increases redistribute surplus and disclosure creates real privacy losses. Consumer surplus and welfare are therefore below NDC, notwithstanding $G^{\VDD}\ge0$. The zero-mass cases in $\mathcal{P}_2$ leave aggregates unchanged. This is why eliminating defensive monetization cannot by itself establish welfare improvement.

When both regimes provide the service, VDD preserves the personalized offers of disclosers while allowing other consumers to purchase anonymously instead of exiting. This explains the improvement over MDC. Provision remains endogenous: if only VDD provides, MDC equals NDC and part (A) applies; if only MDC provides, VDD equals NDC and Proposition~\ref{prop:profit_campare_M_N} applies; if neither provides, all outcomes coincide.

\begin{proposition}[Pareto Improvement]\label{prop:3win}
VDD yields a Pareto improvement over both NDC and MDC, evaluated for the platform, each seller, and consumers in aggregate, if and only if $(v,c)\in \{(v,c):\tfrac{1}{2}<v\le \tfrac{\sqrt{2}}{2},c< v-\tfrac{1}{2}\}\cup \{(v,c):\tfrac{\sqrt{2}}{2}<v<1,v-\tfrac{3}{4}+\tfrac{\sqrt{2v^2-1}}{4}\le c<v-\tfrac{1}{2}\}$ and $\gamma\ge\gamma^{\mathrm{win}}(v,c)$, where
$$
\gamma^{\text{win}}(v,c)=
\begin{cases}
\tfrac{(3v-2c-2)^2}{v^2},
&
\text{if } \tfrac{2}{3}<v<1,c\le \tfrac{3v}{2}-1,\\
0,
&
\text{otherwise}.
\end{cases}
$$
\end{proposition}  

These gains occur in the transaction-creating region, as illustrated in Figure \ref{fig:data_provision}. Sellers retain their NDC net profits, the platform captures additional revenue, and the aggregate consumer surplus is not decreasing. The statement does not claim that all individual buyers are better off: Proposition~\ref{prop:price_coverage_cmp} identifies possible price increases for some consumers. The policy implication is consequently conditional, not a universal recommendation for data collection or a blanket preference for VDD over NDC. 

\subsubsection{Protecting the Anonymous Purchase Option}
The distinction between voluntary disclosure and protection of the anonymous channel motivates a separate design benchmark. Impose an enforceable cap $p_{i,\emptyset}\le p_{i,\emptyset}^{\NDC}$, $i\in\{1,2\}$, under VDD. The platform announces the benchmark before seller adoption and enforces it on anonymous offers. All other baseline assumptions remain unchanged. We refer to this rule as `UP' and let $U^g(x)$ denote consumer $x$'s utility.

\begin{proposition}[Anonymous-Price Protection] \label{cor:up}
When VDD is implemented with an enforceable price cap ($p_{i,\varnothing}\le p_{i,\varnothing}^{\NDC}$, $i=1,2$): 
\begin{itemize}
    \item[{\rm(A)}] The platform provides the data service if and only if $(v,c)\in\{(v,c)|\tfrac{1}{2}<v\le 1, c\le v-\tfrac{1}{2},4c^2\ge5v^2-4v\}\cup \mathcal P_2$.
    \item[{\rm(B)}] $U^{\UP}(x)\ge U^{\NDC}(x)$ for $x\in[0,1]$ and $SW^{\UP}\ge SW^{\NDC}$.
    \item[{\rm(C)}] Conditional on service provision under both UP and MDC, $CS^{\UP}\ge\max\{CS^{\NDC},CS^{\MDC}\}$ and $\pi_i^{\UP}=\pi_i^{\NDC}\ge\pi_i^{\MDC}$ for $i=1,2$. Moreover, UP yields a Pareto improvement over both NDC and MDC, evaluated for the platform, each seller, and consumers in aggregate, if and only if service provision under UP induces positive-mass disclosure and $\gamma\ge\gamma^{UM}(v,c):=\max\{0,1-(r^{\UP}-r^{\MDC})/T^{\MDC}\}$.
\end{itemize}
\end{proposition}

The cap preserves every consumer's NDC anonymous-purchase option. Disclosure must compensate for privacy cost relative to that protected alternative. This rules out segmentation-based monetization and confines positive-mass provision to transaction creation, while preserving seller benchmark profits under full extraction. Unlike the aggregate consumer comparison for uncapped VDD, part (B) protects each consumer relative to NDC. UP is an enforceable policy benchmark, not the platform's privately optimal rule: uncapped VDD can support greater extraction. Applying the cap would require specifying and enforcing a defensible no-data price. That counterfactual benchmark is determined inside the model, and the analysis does not estimate it from marketplace observations or resolve its practical implementation.

\section{Extensions and Counterfactuals}\label{sec:extension}
In this section, we first validate the robustness of our key findings by examining three extensions that test specific restrictions of the baseline model, and then conduct model-based counterfactual analyses to further illustrate these results. 

\subsection{Platform Choice of Access Structure}\label{sec:endogenous_exclusive_access}
The baseline commits to bilateral access on the equilibrium path, with exclusivity only after rejection. Broad access serves as a design benchmark motivated by concerns about discriminatory access to platform resources \citep{EUMetaKustomer2022,EUDataAct2023}. We allow bilateral, exclusive, or no access through fees, retaining sellers' acceptance and rejection rights. Full commitment, zero provision costs, and baseline equilibrium selection keep the revenue states unchanged. Alternative continuation selections affect access boundaries through revenues.

\begin{proposition}[Optimal Access Structure]
\label{app:prop_endogenous_exclusive_access}
If the platform can choose whether to induce bilateral or exclusive data access:
\begin{itemize}
    \item[{\rm{(A)}}]
    For $g\in\{MDC,VDD\}$, bilateral data access is optimal if $G^g+(1-\gamma)T^g\ge 0$ and $2R_B^g\ge R_I^g+R_U^g$. Exclusive data access is optimal if $R_I^g-R_U^g\ge 2\gamma T^g$ and $R_I^g+R_U^g> 2R_B^g$. The platform can induce exclusive access by committing to 
    \[
    F_1^g=(1-\gamma)(R_I^g-R_U^g), \quad F_{2\mid RE}^g=(1-\gamma)(R_I^g-R_N^g), \quad F_{2\mid AC}^g>(1-\gamma)(G^g+T^g).
    \]
    \item[{\rm{(B)}}] Under MDC, $R_I^{MDC}+R_U^{MDC}<2R_N^{MDC}$. Hence, neither bilateral nor exclusive data access increases total seller sales revenue relative to NDC. When the platform implements exclusive data access, 
    $\pi_i^{MDC}<\pi_i^{NDC},i\in\{1,2\}$ and $SW^{MDC}<SW^{NDC}$. Moreover, conditional on exclusive data access under MDC, $CS^{MDC}>CS^{NDC}$ only if $(v,c)\in \left\{(v,c)\,|1<v<\tfrac32\right\}\cup \{(v,c)\,|\tfrac32\le v\le\tfrac74,c<v-\tfrac12-\sqrt{2v-5/2}\}\cup\left\{(v,c)\,\middle|\,v>\tfrac74,\,c<\tfrac14\right\}$.
    \item[{\rm{(C)}}] Under VDD, the platform never offers exclusive data access.
\end{itemize}
\end{proposition}

Proposition \ref{app:prop_endogenous_exclusive_access} shows that platform endogenous access allocation preserves the source of value distinction and the incentive of platform data monetization. Under MDC, the platform may prefer exclusive access when it generates more total sales revenue than bilateral personalization. Nevertheless, total sales revenue under either arrangement remains below NDC. Profitable monetization under MDC therefore still relies on the threat effect. Under endogenous data access, seller profit and social welfare are still reduced, and the consumer surplus rise only under large valuation and low privacy cost. Under VDD, consumers refuse disclosure when only one seller has access, making both the rejection-contingent threat and deliberate exclusivity economically ineffective. Thus, the platform can profitably provide the service only through bilateral sales-revenue gains under VDD.

\subsection{Heterogeneous Privacy Cost}\label{sec:hete_cost}
In the baseline model, we assume that all consumers have a homogeneous privacy cost. Here, we extend the model to consider consumers with heterogeneous privacy concerns. Let privacy cost be $c'\in\{c_H,c_L\}$, where $c_H\ge c_L>0$. The high-cost group has mass $\rho$ and the low-cost group mass $1-\rho$. Raising $\rho$ raises average privacy cost, holding the two costs fixed. Locations remain uniformly distributed within each group. The extension retains strictly positive privacy costs and does not encompass arbitrary multidimensional or correlated preferences.

\begin{proposition}[Robustness to heterogeneous privacy costs]
\label{prop:heterogeneous_privacy}
Under heterogeneous privacy costs, whenever the pure-strategy rational-expectations continuation equilibria exist:
\begin{itemize}
    \item[{\rm (A)}] Under MDC, $G^{\MDC-h}<0$ and $T^{\MDC-h}>0$, and the platform provides the data service if and only if $G^{\MDC-h}+(1-\gamma)T^{\MDC-h}\ge0$. Conditional on provision, $D^{\MDC-h}\le D^{\NDC}$ and $SW^{\MDC-h}<SW^{\NDC}$. Moreover, $CS^{\MDC-h}\ge CS^{\NDC}$ if and only if $CS_L^{\MDC-h}\ge CS^{\NDC}$ and $\rho\le\rho'$, where $CS_L^{\MDC-h}$ is consumer surplus per unit mass of low-cost consumers and $\rho'$ is defined in the proof.
    \item[{\rm (B)}] Under VDD, $T^{\VDD-h}=0$ and $V_{data}^{\VDD-h}=2G^{\VDD-h}$. The platform provides the service if and only if $G^{\VDD-h}\ge0$ and the opt-in set is nonempty. Every symmetric bilateral VDD equilibrium satisfies $p_{i,\varnothing}^{PPV-h}\ge p_i^{\NDC}$ for $i\in\{1,2\}$.
    \item[{\rm (C)}] If both regimes provide the service, $D^{\VDD-h}\ge D^{\MDC-h}$, $CS^{\VDD-h}\ge CS^{\MDC-h}$, and $SW^{\VDD-h}\ge SW^{\MDC-h}$. Moreover, $\pi_i^{\VDD-h}=\pi_i^{\NDC}\ge\pi_i^{\MDC-h}$, while $\Pi_p^{\VDD-h}\ge\Pi_p^{\MDC-h}$ if and only if 
    \[\gamma\ge\max\left\{0,1-(G^{\VDD-h}-G^{\MDC-h})/T^{\MDC-h}\right\}.\]
\end{itemize}
\end{proposition}


Privacy-cost heterogeneity changes the distribution of consumer gains and losses without altering the source of data value. Under MDC, aggregate consumer surplus depends on the shares of low- and high-cost consumers, while social welfare remains below NDC. Under VDD, neither privacy type discloses to a single informed seller, so $T^{\VDD-h}=0$. When provision under both MDC and VDD, VDD weakly increases coverage, consumer surplus, and welfare relative to MDC, with the platform comparison governed by the extension's commission threshold. The unilateral nondisclosure result under VDD holds for all strictly positive privacy costs, whereas the bilateral pricing and stakeholder comparisons are conditional on the existence of the selected continuation equilibria.

\subsection{Imperfect Identifiation}\label{sec:imp_quality}
The baseline model assumes that the platform's data perfectly reveals consumer preference $x$. We now relax this assumption and replace certain identification with a reveal-or-no-signal technology. Let $q\in(0,1]$ be the probability that the service reveals a disclosing consumer's exact location to both sellers. Otherwise, it produces no usable signal and the consumer remains in the uniform-price channel. Privacy cost is incurred only following successful identification. This common-knowledge specification is denoted ``VDD-IQ". Thus, $q$ measures the frequency of exact identification, and the extension does not model the precision of a noisy location estimate.

\begin{proposition}[Robustness to imperfect identification]
\label{prop:imperfect_quality}
For any $q\in(0,1]$:
\begin{itemize}
    \item[{\rm (A)}] Under MDC, exclusive access may remain informative. The platform provides the data service if and only if $G^{\MDC}(q)+(1-\gamma)T^{\MDC}(q)\ge0$.
    \item[{\rm (B)}] Under VDD, no consumer discloses under exclusive access and $T^{\VDD}(q)=0$. Thus, $V_{data}^{\VDD}(q)=2G^{\VDD}(q)$, and the platform provides the data service if and only if $G^{\VDD}(q)\ge0$ and the equilibrium opt-in set is nonempty. 
\end{itemize}
\end{proposition}

Under MDC, imperfect identification can leave exclusive access informative because consumers cannot refuse the modeled data use while purchasing anonymously. Under VDD, successful unilateral identification creates a privacy cost without improving the consumer's purchasing opportunity; failed identification preserves the anonymous option. Consumers therefore withhold disclosure from a single informed seller for every positive identification probability. This preserves threat elimination under the specified technology, without implying that all imperfect information structures have the same effect.


\subsection{Model-Based Counterfactuals}
Amazon Brand Tailored Promotions motivates platform control of audience eligibility and targeted discounts, while its US referral-fee schedule motivates $\gamma=0.15$ \citep{toogood2026amazonbrandtailored,AmazonFees2026}. We set $v=0.73$ and $c=0.08$ to illustrate a partially covered NDC market, threat-supported MDC provision, and transaction-creating VDD. These are normalized scenario parameters, \emph{not estimates} of Amazon demand or privacy preferences. No proprietary observations, treatment effects, access fees, or disclosure rates are inferred from the public program description. In particular, it does not document the model's sequential fee schedule or establish that consumers can transact on Amazon under the precise anonymous-consent arrangement assumed here. 

\vspace{-0.2in}
\begin{table}[!ht]
\centering
\caption{Baseline Model-Based Counterfactuals}\label{tab:numerical_counterfactual_results}
\normalsize
\small
\setlength{\tabcolsep}{5pt}
\renewcommand{\arraystretch}{1.2}
\begin{tabular}{lrrrrr}
\toprule
Outcome & NDC & MDC & VDD & MEA & VDD-IQ\\
\midrule
Service provided & No & Yes & Yes & No & Yes\\
Disclosure mass & 0.0000 & 0.3000 & 0.3000 & 0.0000 & 0.3000\\
Market coverage & 0.7300 & 0.3000 & 1.0000 & 0.7300 & 0.8500\\
Seller revenue ($r$) & 0.1332 & 0.0225 & 0.1555 & 0.1332 & 0.1443\\
Seller net profit & 0.1132 & 0.0000 & 0.1132 & 0.1132 & 0.1132\\
Platform profit & 0.0400 & 0.0450 & 0.0845 & 0.0400 & 0.0620\\
Consumer surplus ($CS$) & 0.1332 & 0.0225 & 0.1450 & 0.1332 & 0.1338\\
Social welfare ($SW$) & 0.3997 & 0.0675 & 0.4560 & 0.3997 & 0.4223\\
Sales-revenue effect ($G$) & 0.0000 & $-0.1107$ & 0.0223 & 0.0000 & 0.0110\\
Competition-threat effect ($T$) & 0.0000 & 0.1332 & 0.0000 & 0.0000 & 0.0000\\
\bottomrule
\end{tabular}
\vspace{0.1in}
\begin{flushleft}
    \footnotesize
    \textit{Notes.} Values are reported to four decimal places. Seller revenue and net profit are per seller; platform profit, consumer surplus, and welfare are marketplace totals. Disclosure denotes opt-in/disclosure mass, not successfully identified mass. VDD-IQ uses $q=0.5$. According to Proposition 1, commission-based values for NDC, MDC, VDD, MEA, and VDD-IQ are, respectively, $0.0000$, $-0.0332$, $0.0067$, $0.0000$, and $0.0033$; data-monetization values are $0.0000$, $0.0383$, $0.0379$, $0.0000$, and $0.0187$.
\end{flushleft}
\vspace{-0.2in}
\end{table}

Table \ref{tab:numerical_counterfactual_results} includes mandatory equal access (MEA), under which data access is activated only if both sellers adopt, and VDD-IQ with imperfect identification. The matched 30\% disclosure proportions under MDC and VDD highlight the importance of purchasing alternatives. Under MDC, disclosers constitute the entire purchasing population. At these parameters, the selected exclusive-access continuation has no entry, so the reported threat is participation-driven. Under VDD, nondisclosers continue purchasing and coverage reaches one. Seller net profit remains at NDC, while platform profit rises from $0.0400$ to $0.0845$. MEA also eliminates the exclusive-access threat but creates no positive bilateral revenue gain in this scenario, so the platform withholds the service. Equal access and voluntary disclosure therefore address different margins.

{\bf Conversion versus revenue lift.}
The first diagnostic examines whether higher conversion implies positive sales-revenue value. Let $\tau\in[0,1)$ be the targeted discount and $\ell\ge0$ the proportional purchase-rate lift. Given baseline revenue $R_{\mathrm{full}}=pD$, promotional revenue and its proportional change are 
\[
R_{\mathrm{disc}}(\tau,\ell)=(1-\tau)p(1+\ell)D, \quad \Delta R(\tau,\ell)=(1+\ell)(1-\tau)-1.
\]
The zero revenue change boundary is $\ell^*(\tau)=\tau/(1-\tau)$. For example, a $20\%$ lift with a $20\%$ discount reduces revenue by $4\%$. Conversion lift alone is therefore insufficient to measure the sales-revenue value of data and may overstate the value created by targeting. Data value should instead be evaluated by the resulting revenue change, captured in our framework by $G$. Moreover, if the platform charges a high data-service fee despite $G<0$, monetization may reflect the threat effect $T$ rather than revenue created by the data.

{\bf Threat availability and commission rate.}
We now examine the sensitivity of data monetization to the threat effect and platform commission. Let $\psi\in[0,1]$ denote the share of the baseline MDC threat available for monetization, so $T(\psi)=\psi T^{\MDC}$. Holding the sales-revenue effect fixed, full-commitment data value is
\[V_{\mathrm{data}}^{\MDC}(\psi,\gamma)=2G^{\MDC}+2(1-\gamma)\psi T^{\MDC}.\]
For $G^{\MDC}<0<T^{\MDC}$, provision requires
$\psi\ge-G^{\MDC}/[(1-\gamma)T^{\MDC}]$. This is a reduced-form stress test of threat availability. Its parameter $\psi$ is distinct from the subsidy-credibility parameter $\eta$ in Proposition 2: a concrete change in governance may also alter pricing, participation, and the revenue effect.

Using Table~\ref{tab:numerical_counterfactual_results}'s $G^{\MDC}=-0.1107$ and $T^{\MDC}=0.1332$, Figure~\ref{fig:threat_monetization}  shows that positive data value requires a threat share above approximately $0.978$ at the baseline commission. Thus, a roughly 2.2\% reduction in the available threat eliminates the provision incentive. Even with the full threat, data value becomes negative when $\gamma>0.169$. Higher commissions increase the platform’s exposure to lost product revenue while reducing sellers’ retained revenue available for defensive fees. Threat-supported provision therefore requires both substantial rejection-payoff leverage and a sufficiently low commission rate.\\

\begin{figure}[!ht]
 \vspace{-0.4in}
    \centering
    \caption{Sensitivity of Platform Data Value to Threat Availability and Commission Rate}
    \includegraphics[width=\textwidth]{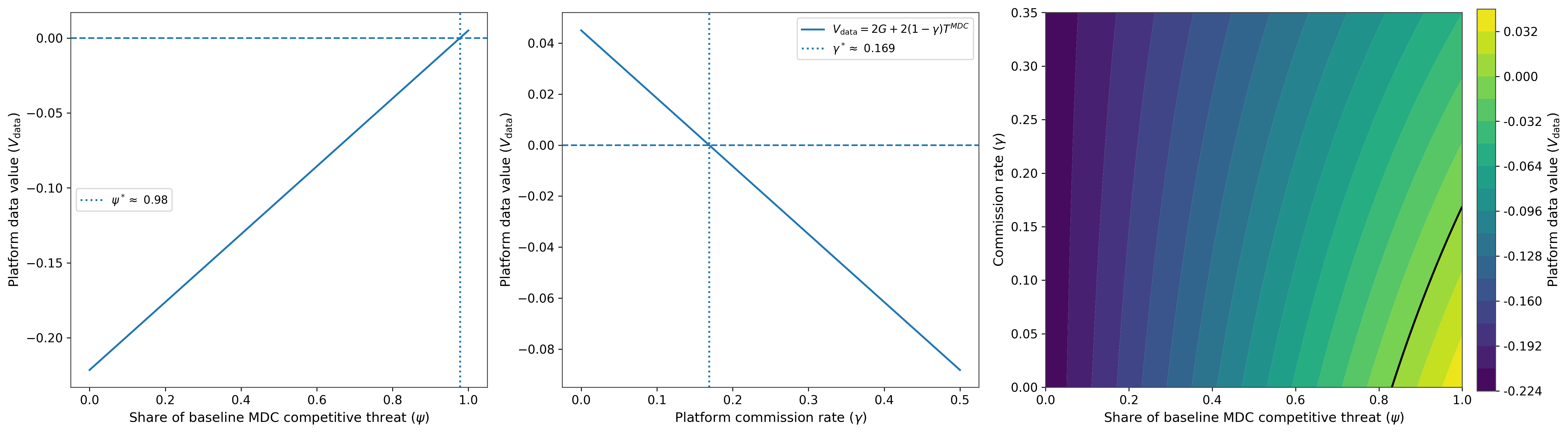}
    \label{fig:threat_monetization}
    \vspace{-0.3in}
\end{figure}

{\bf Identification probability under VDD.} Finally, we examine the identification requirements for profitable data provision and the market outcomes under VDD. Figures~\ref{fig:data_provision_imp}--\ref{fig:data_quality} use the reveal-or-no-signal extension, with privacy cost incurred only upon successful identification. Figure~\ref{fig:data_provision_imp} shows that higher privacy costs need not increase the identification probability required for positive data value, and higher valuation need not facilitate monetization. The bilateral analysis supporting Proposition~\ref{prop:imperfect_quality} establishes a symmetric equilibrium for $\tfrac12<v<1$, $0<c<v-\tfrac12$. Higher privacy costs can discourage consumers from replacing existing anonymous purchases with cheaper personalized offers, limiting revenue losses from personalization. Higher valuation expands NDC coverage and leaves fewer additional transactions to create. Figure~\ref{fig:data_quality}, at $v=0.73$, $c=0.08$, and $\gamma=0.15$, also illustrates that platform and consumer outcomes can move differently as exact identification becomes more likely. Eliminating defensive monetization does not by itself align those outcomes.

\begin{figure}[!ht]
\vspace{-0.1in}
\noindent
    \begin{minipage}[t]{0.49\textwidth}
    \vspace{0pt}
    \centering
    \captionsetup{font=footnotesize}
    \caption{Identification Thresholds for Positive\\ VDD-IQ
    Data Value ($\gamma=0.15$)}
    \label{fig:data_provision_imp}
    \includegraphics[width=\linewidth,height=2.3in]{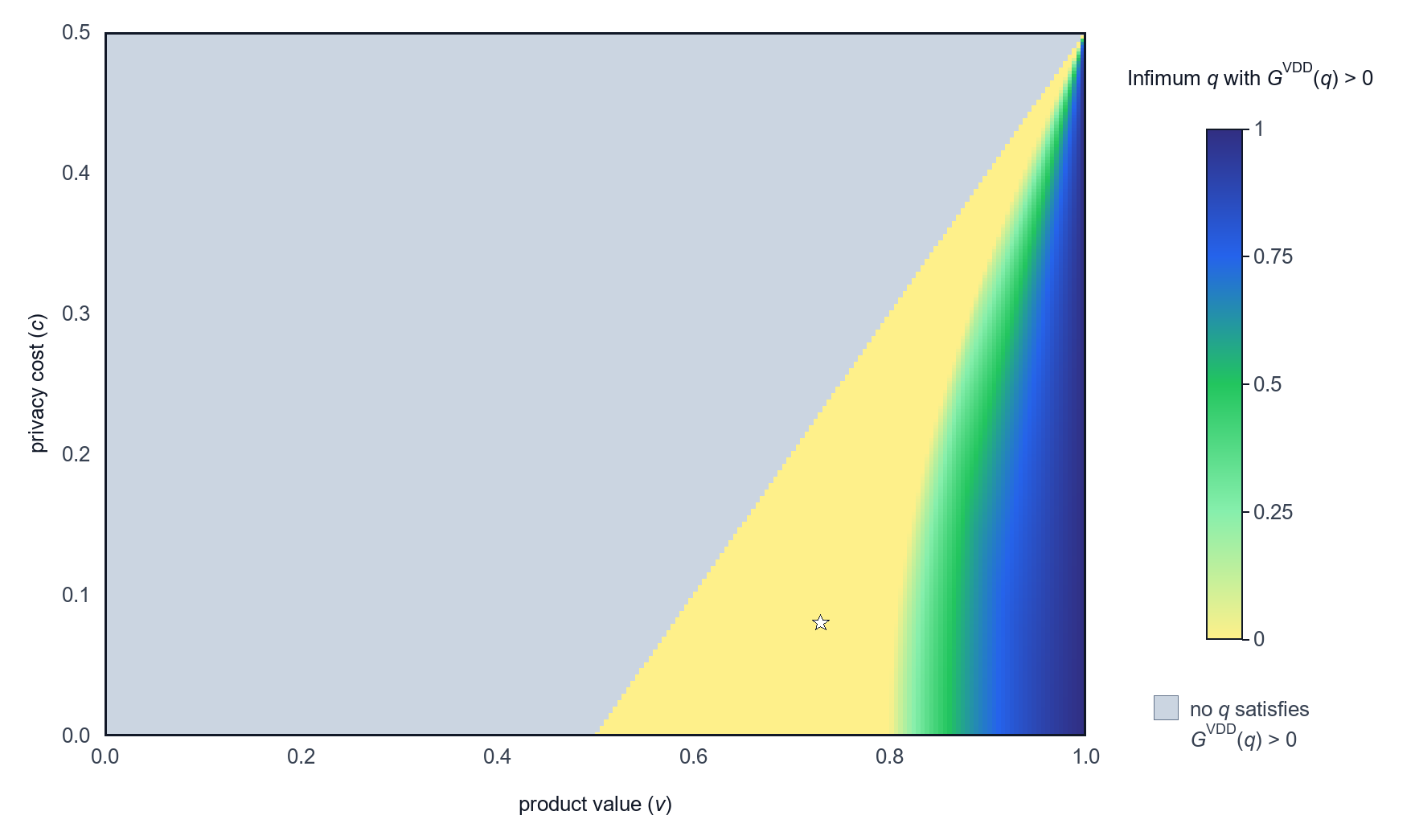}
    \end{minipage}
    \hfill
    \begin{minipage}[t]{0.49\textwidth}
    \vspace{0pt}
    \centering
    \captionsetup{font=footnotesize}
    \caption{Market Outcomes as Exact Identification \\
    Improves ($v=0.73,c=0.08,\gamma=0.15$)}
    \label{fig:data_quality}
    \includegraphics[width=\linewidth,height=2.3in]{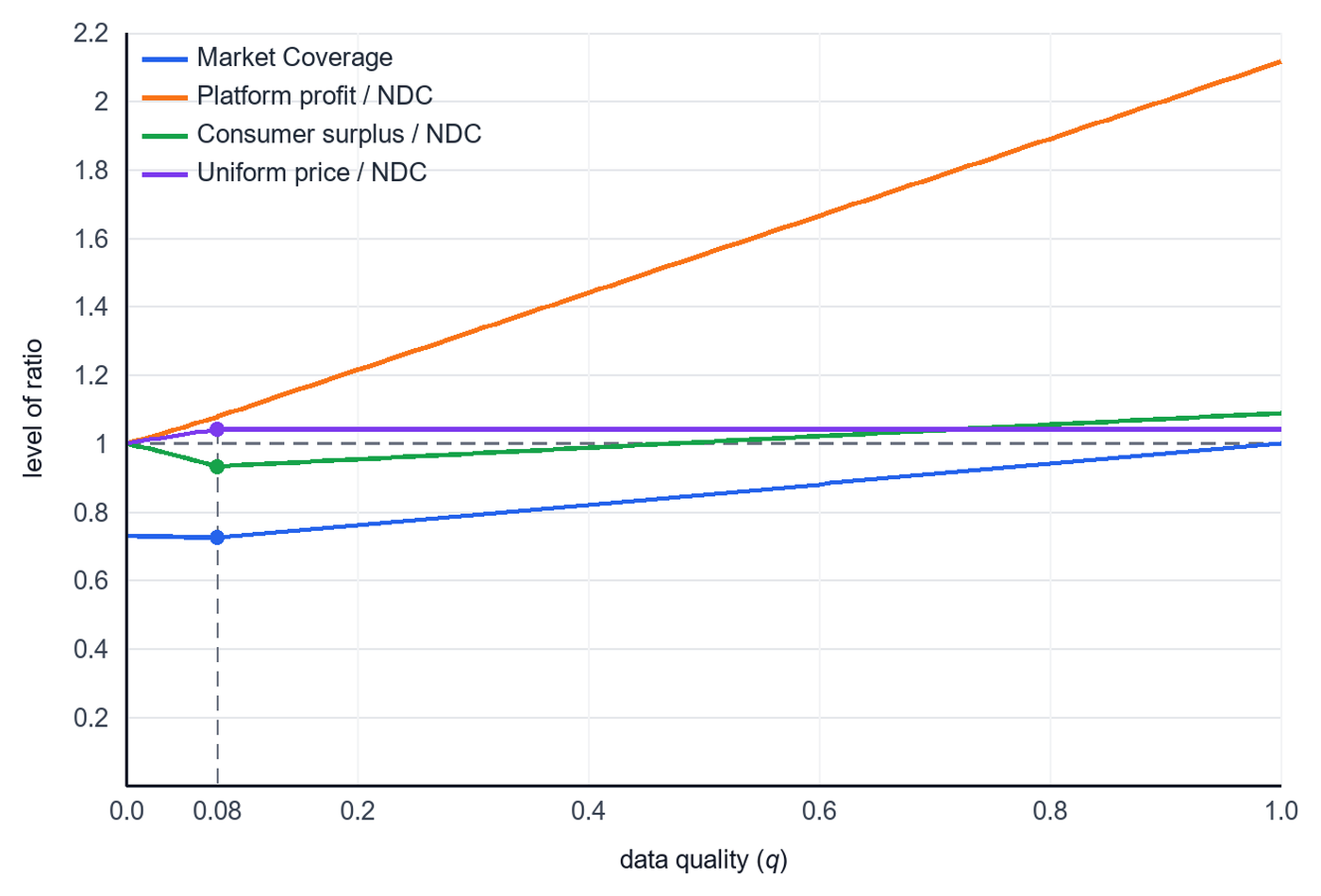}
    \end{minipage}
\vspace{-0.25in}
\end{figure}

\section{Conclusion and Managerial Implications}\label{sec:conclusion}
This paper shows how consumer disclosure rights reshape the disagreement
payoffs supporting data contracts in a commission-based marketplace.
The commercial value of platform-controlled information reflects both
changes in bilateral seller revenue and losses sellers avoid by purchasing
access. Mandatory collection can support defensive fees through asymmetric
competition or reduced participation following rejection, even when
bilateral personalization lowers seller revenue. Under baseline voluntary
disclosure, consumers withhold information from a sole subscriber while
retaining anonymous purchasing, restoring sellers' no-data rejection
payoffs. The platform can then monetize bilateral revenue gains, but not
an exclusive-access threat. This connects consumer control to upstream
contracting and service provision: privacy governance changes
\emph{why sellers pay for data}, not merely how much information is
available.

The analysis yields several managerial implications. \textbf{First, evaluate data services against both no-data and rejection
counterfactuals.}
Managers should distinguish the revenue change under bilateral use from
the loss a seller faces when its rival alone subscribes. Adoption and
conversion should be assessed alongside gross product revenue and retained
profit after commissions and access fees; deeper discounts can increase
purchases while reducing revenue. Sellers should likewise evaluate the
contractual consequences of refusal rather than equate rejecting a
subscription with restoring a no-data market. Where exclusive access
requires a subsidy, its credibility affects the defensive payments the
platform can extract. Under full fee extraction, VDD protects sellers'
no-data net profits but does not allow them to retain their gross revenue
gains.

\textbf{Second, evaluate access fees together with commission exposure.}
Holding downstream revenue states fixed, lower commissions increase
sellers' retained revenue at risk after rejection and strengthen
extractable defensive payments. Higher commissions reduce this fee
component and increase the platform's exposure to lost product revenue.
A data service can therefore generate access fees while shrinking the
commission base. Once VDD eliminates the threat, however, full fee
extraction makes the platform's incremental provision value independent
of the commission rate. Because commissions are exogenous in our model,
these results guide the evaluation of service provision and pricing at
a given rate; they do not identify an optimal commission policy.

\textbf{Third, preserve anonymous purchasing and protect its price.}
VDD can preserve or expand coverage because consumers need not disclose
to purchase. This protection requires effective refusal of pricing-related
data use, not merely a nominal opt-out. Yet anonymity alone does not
preserve no-data prices: in the selected high-valuation segmentation
equilibrium, anonymous prices rise without additional trade. An
enforceable cap at the no-data benchmark addresses this remaining channel
by preserving every consumer's no-data purchasing option. In the model,
the cap protects individual consumer utility and welfare while preserving
sellers' no-data net profits. The platform remains weakly better off than
under no data because it can withhold the service, although its profit can
fall relative to uncapped VDD. Implementation requires a defensible
benchmark announced before seller adoption and enforcement of the cap
on anonymous offers. Disclosure rights and price protection thus address
distinct sources of harm.

\textbf{Fourth, distinguish equal-access commitments from consumer
disclosure rights.}
Offering the service to both sellers does not eliminate defensive demand if rejection leaves the rival exclusively informed. A meaningful equal-access commitment must also prohibit exclusive access after rejection. Although such a rule removes the threat, it does not provide an anonymous purchasing option or ensure that bilateral data use raises revenue. In the reported counterfactual, mandatory equal access leads the platform to withhold the service, whereas VDD supports provision and full coverage. Access rules determine who may use information; disclosure rights also affect whether usable information is generated and whether consumers remain willing to purchase. Platforms and regulators should therefore evaluate these instruments separately, using seller revenue, retained profit, market coverage, consumer surplus, and privacy losses—not adoption or disclosure volume alone.

The results offer a key lesson for marketplace practitioners: the return on a platform's data service should be evaluated against relevant competitive alternatives, rather than solely based on seller revenue without the service. A seller may rationally opt to purchase data even if personalization lowers overall revenue, as refusing access could heighten its disadvantage against an informed rival. Therefore, sellers should distinguish between productive returns and defensive expenditures when evaluating data-service contracts and negotiating terms. Similarly, policymakers should be cautious in interpreting strong seller demand for data as evidence of efficiency, as it may reflect competitive vulnerability rather than productive information use.

The model deliberately isolates the interaction among platform-controlled information, personalized pricing, seller adoption, and consumer disclosure. Future research could incorporate heterogeneous sellers, independent first-party data, direct non-price benefits from disclosure, richer bargaining arrangements, and alternative uses of consumer information such as recommendations, advertising, ranking, and demand forecasting. These extensions may restore some value to unilateral data access under consumer control. More generally, our analysis suggests that the economically relevant measure of privacy governance is not disclosure volume alone, but the amount and quality of information that can be used to create an asymmetric competitive advantage.


\bibliographystyle{informs2020}
\bibliography{reference}

\clearpage
\newpage
\setcounter{page}{1}

\begin{APPENDICES}
\linespread{1.2259}\selectfont
\begin{center}
 \textbf{\large Online Appendix to\\
 \PaperTitle}
\end{center}

\renewcommand{\thesection}{EC.\arabic{section}}
\renewcommand{\thesubsection}{\thesection.\arabic{subsection}}
\renewcommand{\thesubsubsection}{\thesubsection.\arabic{subsubsection}}

\renewcommand{\theHsection}{OA.\arabic{section}}
\renewcommand{\theHsubsection}{\theHsection.\arabic{subsection}}
\renewcommand{\theHsubsubsection}{\theHsubsection.\arabic{subsubsection}}

\counterwithin{lemma}{section}
\counterwithin{proposition}{section}
\counterwithin{equation}{section}
\counterwithin{figure}{section}
\counterwithin{table}{section}

\renewcommand{\thelemma}{\thesection.\arabic{lemma}}
\renewcommand{\theproposition}{\thesection.\arabic{proposition}}
\renewcommand{\theequation}{Eq-\thesection.\arabic{equation}}
\renewcommand{\thefigure}{\thesection.\arabic{figure}}
\renewcommand{\thetable}{\thesection.\arabic{table}}

\vspace{1cm}
This Appendix presents the pricing-subgame equilibria of the baseline model in Appendix \ref{app:sec_subgames}, supplementary Analysis of Platform Data Value in Appendix \ref{app:sec_general_model}, proofs of the main results in Appendix \ref{app:sec_proof_main}, and proofs of additional results and tables in Appendix \ref{app:sec_proof_app}.

\section{Pricing Subgame Equilibria in Baseline Model}\label{app:sec_subgames}

\subsection{MDC Pricing Subgame Equilibria}\label{app:sec_MDC_subgame}

\textbf{1. The ``NN" Subgame.} When neither seller purchases the data service, both use uniform pricing and consumers incur no privacy cost. The equilibrium is identical to the NDC benchmark in Proposition~\ref{prop:without_data}.

\textbf{2. The ``PP" Subgame.} Both sellers offer personalized prices $p_i(x)$ to consumers at $x\in[0,1]$. Let $\Omega^{PPM}$ denote the equilibrium entry set and $p_i^{PPM}(x)$, $i\in\{1,2\}$, the equilibrium pricing strategies.

\begin{lemma}\label{app:lem_mdc_PP}
 Under mandatory data collection, when both sellers purchase the data service, consumers located in $x\in \Omega^{PPM}$ enter the platform, where
$$\Omega^{PPM} = 
 \begin{cases} \emptyset, & \text{if } v\le \tfrac{1}{2} \text{ or if } c> v-\tfrac{1}{2}\\
 [ \max\{0,1-(v-c)\}, \min\{1,v-c\} ], & \text{if } v>\tfrac{1}{2}, c\le v-\tfrac{1}{2}
 \end{cases},$$
while others stay out.
Whenever $\Omega^{PPM}\neq\varnothing$, the equilibrium personalized prices for each $x\in\Omega^{PPM}$ are 
$ p_1^{PPM}(x) = \max\{1-2x,0\}$ and $ p_2^{PPM}(x) = \max\{2x-1,0\}$.
\end{lemma}

Lemma \ref{app:lem_mdc_PP} shows that consumers closer to the market center are more likely to enter under MDC, while those near the endpoints may stay out. Personalized pricing allows the better-matched seller to charge its location advantage over the rival, which is largest at the endpoints and decreases toward the center. Personalized prices therefore fall toward the center, and then consumers enter only when this surplus covers the privacy cost $c$. Hence, no consumer enters when $v-c$ is low; as $v$ increases or $c$ decreases, entry begins around the center and expands toward the endpoints until the entire market is covered.

\textbf{3. The ``PN" and ``NP" Subgames.} In the ``PN" subgame, seller 1 sets the personalized price $p_1(x)$ and seller 2 the uniform price $p_{2,\varnothing}$. Let $\Omega^{PNM}$ denote the equilibrium entry set and $p_1^{PNM}(x)$ and $p_{2,\varnothing}^{PNM}$ the equilibrium prices. If multiple rational-expectations equilibria exist, we select the one maximizing the uninformed seller's sales revenue, capturing optimistic common expectations about its pricing response.

\begin{lemma}\label{app:lem_mdc_PN}
Under MDC, when only seller 1 purchases the data service, consumers in $x\in\Omega^{PNM}$ enter the platform, where 
$$\Omega^{PNM}=
\begin{cases}
\emptyset,& \text{if } v-c<\tfrac{3}{4}\\
\left[\max\{0,\tfrac{3}{2}-(v-c)\},1\right],& \text{if } v-c\ge \tfrac{3}{4}
\end{cases},$$
and others stay out.
Whenever $\Omega^{PNM}\neq\emptyset$, equilibrium prices are $p_1^{PNM}(x)=\max\{\tfrac{3}{2}-2x,0\}$ and $p_{2,\varnothing}^{PNM}=\tfrac{1}{2}$.
\end{lemma}
\begin{remark}\label{app:rem_mdc_PN}
When $\frac{3}{4}<v-c\le1$, the ``PN'' subgame admits an alternative rational-expectations equilibrium with a smaller entry set: consumers in $[v-c,1]$ enter, and the uninformed seller charges $2(v-c)-1$. The threat effect remains strictly positive and is weakly larger than under the maximal-participation equilibrium. Thus, selecting the latter gives a conservative measure of the threat generated by a rival's data access.
\end{remark}

By symmetry, in the ``NP" subgame, seller 2 sets the personalized price $p_2(x)$ and seller 1 the uniform price $p_{1,\varnothing}$. Let $\Omega^{NPM}$, $p_2^{NPM}(x)$, and $p_{1,\varnothing}^{NPM}$ denote the corresponding equilibrium entry set and prices.
\begin{lemma}\label{app:lem_mdc_NP}
Under MDC, when seller 2 adopts the exclusive data access, consumers located in $x\in\Omega^{NPM}$ enter the platform while others stay out, where 
$$\Omega^{NPM}=
\begin{cases}
\emptyset,& \text{if } v-c<\tfrac{3}{4}\\
\left[0,\min\{v-c-\frac{1}{2},1\}\right],& \text{if } v-c\ge \tfrac{3}{4}
\end{cases}.$$
Whenever \(\Omega^{NPM}\neq\emptyset\), equilibrium prices are $p_2^{NPM}(x)=\max\{2x-\tfrac{1}{2},0\}$ and $p_{1,\varnothing}^{NPM} = \tfrac{1}{2}$.
\end{lemma}

Lemma \ref{app:lem_mdc_PN} and Lemma \ref{app:lem_mdc_NP} show that, when only one seller purchases the data service, partial coverage is concentrated on the side of the uninformed seller, and the entry region expands toward the informed seller as $v-c$ increases. Seller 2, the uninformed seller, charges the uniform price $\tfrac12$, which serves as an outside option to seller 1's personalized offer. Near seller 1, this alternative is relatively unattractive, allowing seller 1 to charge a high personalized price and leave little consumer surplus, so consumers may stay out. Toward seller 2, the uniform offer becomes more attractive, constraining seller 1's price and increasing consumer surplus and entry. Thus, the uninformed seller promotes participation by limiting the informed seller's surplus extraction. Seller 1's personalized price falls with $x$ as its location advantage declines, and sufficiently close to seller 2, it cannot profitably match the uniform offer even at zero price, so consumers buy from seller 2. As $v-c$ increases, more consumers can cover the privacy cost, and the entry region expands until the market is fully covered.

\subsection{VDD Pricing Subgame Equilibria}\label{app:sec_VDD_subgame}
{\bf 1. The ``NN" Subgame.} If neither seller adopts the data service, consumer disclosure does not affect the outcome because both sellers can only charge uniform prices.

{\bf 2. The ``PP" Subgame.}
For seller $i\in \{1,2\}$, denote the equilibrium personalized prices and the uniform prices as $p_i^{PPV}(x)$ and $p_{i,\varnothing}^{PPV}$, respectively. We focus on symmetric pure-strategy rational-expectations equilibria. Consumers enter and opt in when indifferent. When multiple such equilibria exist, we select an equilibrium that maximizes each seller's sales revenue. Remark \ref{app:rem_vdd_PP} describes the possible alternative equilibria and their implications for bilateral revenue and service provision.

\begin{lemma}\label{app:lem_vdd_pp}
Under voluntary regulation, when both sellers purchase the data service, all consumers enter the platform. Furthermore, the following results hold:
{\rm{(A)}} If
$(v,c)\in
\left\{(v,c)\,\middle|\,v\le\tfrac12\right\}
\cup
\left\{(v,c)\,\middle|\,
v>\tfrac12,\ 
c\ge\min\left\{v-\tfrac12,1\right\},\
(v,c)\notin\mathcal P_3
\right\}$,
then $p_{i,\varnothing}^{PPV}=p_{i,\varnothing}^{NDC}$
and
\[
\Omega^{PPV}=
\begin{cases}
\{1/2\},
&{\rm if \ } v>\tfrac12,\ 
c=\min\left\{v-\tfrac12,1\right\},\\
\varnothing,
& \text{otherwise}.
\end{cases}
\]
The consumers located at $\tfrac{1}{2}$ face personalized price at \(p_1^{PPV}(1/2)=p_2^{PPV}(1/2)=0\).
{\rm{(B)}} If
$(v,c)\in
\left\{(v,c)\,\middle|\,
v>\tfrac12,\ 
c<\min\left\{v-\tfrac12,1\right\}
\right\}
\cup\mathcal P_3$,
then $\Omega^{PPV}=\left[\max\{0,1-v+c\},\min\{1,v-c\}\right]$,
opt-in consumers face $p_1^{PPV}(x)=\max\{1-2x,0\}$ and $p_2^{PPV}(x)=\max\{2x-1,0\}$, and all other consumers opt-out. If $(v,c)\in\{(v,c)|v-c<1\}$, then
$p_{i,\varnothing}^{PPV}=\max\left\{\tfrac v2,2v-c-1\right\},i=1,2$.
If $(v,c)\in\{(v,c)|v>1,c\le v-1, c<1\}$, then all consumers disclose, and the uniform price satisfying $p_{1,\varnothing}^{PPV}=p_{2,\varnothing}^{PPV}\ge1+c$
supports the selected outcome; \(1+c\) is a representative price.
\end{lemma}

Lemma \ref{app:lem_vdd_pp} shows that all consumers enter under VDD, while disclosure depends on product valuation and privacy cost. When $v$ is very low or $c$ very high, all consumers opt out because personalized pricing cannot offset the privacy cost. For intermediate values, central consumers opt in while those near the endpoints opt out. Because central consumers are similarly matched with both sellers, disclosure intensifies personalized-price competition and lowers prices enough to offset the privacy cost; near the endpoints, the better-matched seller can extract more surplus, making disclosure less attractive. As $v$ increases or $c$ decreases, the opt-in region expands toward the endpoints until all consumers opt in. Consistently, personalized prices fall toward the market center as competition intensifies. Thus, VDD preserves full participation while inducing disclosure only when its competitive benefit exceeds the privacy cost.

\begin{remark}
\label{app:rem_vdd_PP}
The ``PP" subgame admits the following alternative equilibria: {\rm (i)} If $(v,c)\in\{(v,c)|v\ge2,1\le c\le v-1\}$, there is a full-disclosure equilibrium with $\Omega^{PPV}=[0,1]$ and $p_{1,\varnothing}^{PPV}=p_{2,\varnothing}^{PPV}\ge1+c$. {\rm (ii)} If $(v,c)\in\mathcal P_3$, there is an empty-disclosure equilibrium with $\Omega^{PPV}=\varnothing$ and $p_{1,\varnothing}^{PPV}=p_{2,\varnothing}^{PPV}=1$. {\rm (iii)} If $(v,c)\in\{(v,c)|\tfrac32<v<\tfrac{3+\sqrt2}{2},c\ge\max\{1,2v-3\},v-1<c<v-1+\tfrac1{4(v-1)}\}$, there is a partial-disclosure equilibrium with $\Omega^{PPV}=[1-v+c,v-c]$ and $p_{1,\varnothing}^{PPV}=p_{2,\varnothing}^{PPV}=2v-c-1$.

Anticipating any of these equilibria, the platform does not provide the data service. Relative to the selected equilibria, provision changes only on $\mathcal P_3$ and when $v>\tfrac32$ and $c=1$. On $\mathcal P_3$, the alternative equilibrium weakly lowers platform profit and seller sales revenues, leaves seller profits unchanged, and strictly increases consumer surplus. All alternative equilibria yield weakly lower seller sales revenues in the PP subgame. Since data-service fees are sunk, the equilibrium in Lemma \ref{app:lem_vdd_pp} weakly maximizes both sellers' continuation profits and is therefore seller-preferred among symmetric pure-strategy rational-expectations equilibria. The threat effect remains zero under VDD in all alternative equilibria. 
The two equilibria on $\mathcal P_3$ imply
$\Pi_p^{VDD}\in[\Pi_p^{NDC},\Pi_p^{NDC}+2G^{VDD}]$,
$CS^{VDD}\in[CS^{NDC}-2G^{VDD}-c(2v-2c-1),CS^{NDC}]$, and
$SW^{VDD}\in[SW^{NDC}-c(2v-2c-1),SW^{NDC}]$,
with seller net profits unchanged at NDC.
\end{remark}

{\bf 3. The ``PN" and ``NP" Subgames.} In both subgames, only one seller adopts data service.  

\begin{lemma}\label{app:lem_vdd_pn}
With voluntary data disclosure regulations, if only one seller purchases the data service, all consumers enter the platform and opt-out of data disclosure.
\end{lemma}

Lemma \ref{app:lem_vdd_pn} implies that if only one seller purchases the data service, the data market has no data and thus the equilibrium prices are identical to those under NDC.

\section{Supplementary Analysis of Platform Data Value}\label{app:sec_general_model}
\subsection{Credibility Microfoundation}\label{app:sec_commitment_microfoundation}
Section~\ref{subsec:value} treats commitment credibility $\eta^g$ as a reduced-form primitive. We now endogenize this credibility from the platform's continuation payoff and the cost of reneging on a subsidy-backed access commitment. 

We focus on the subsidy region, $R_I^g<R_N^g$. Suppose that providing the data service entails a per-seller operating cost $a\ge0$ and a fixed implementation cost $K\ge0$, incurred upon activation. Exclusive access costs $a+K$, bilateral access costs $2a+K$, and no access incurs neither cost. If the announced exclusive-access continuation is implemented after seller 1's rejection, seller 2 obtains $(1-\gamma)R_I^g-F_{2\mid RE}^g$, while the platform obtains $\gamma(R_I^g+R_U^g)+F_{2\mid RE}^g-a-K$. If the platform instead revises the continuation to no data access, seller 2 obtains $(1-\gamma)R_N^g$, while the platform earns $2\gamma R_N^g$ before any reneging penalty. At the committed fee $F_{2\mid RE}^g=(1-\gamma)(R_I^g-R_N^g)$, seller 2 is indifferent between accepting exclusive access and receiving its no-data payoff. The platform's continuation-profit loss from honoring the subsidy is therefore
$\Delta L^g\equiv R_N^g-R_I^g+\gamma T^g+a+K$.
Let $\kappa\ge0$ denote the contractual, legal, or reputational loss from reneging on the announced subsidy. Its realization occurs after seller 1's adoption decision and is privately observed by the platform, while seller 1 knows only its distribution. We assume that $\kappa$ has a continuous cumulative distribution function $\Phi_\kappa(\cdot)$ with density $\phi_\kappa(\cdot)$. The platform honors the subsidy whenever $\kappa\ge\Delta L^g$, thus the credibility is given by
$\eta^g=\Pr(\kappa\ge\Delta L^g)=1-\Phi_\kappa(\Delta L^g)$.
When $R_I^g\ge R_N^g$, no subsidy is required and the rejection-contingent continuation remains fully credible.

\begin{proposition}[Comparative statics of data value]
\label{app:prop_cred_micro}
In the subsidy region:
{\rm(A)} If $T^g>0$, the platform's data value decreases with $\gamma$, $a$ and $K$, and weakly increases under an upward shift in the distribution of the reneging penalty $\kappa$.
{\rm(B)} If $T^g<0$, data value weakly increases with $a$ if and only if
$\phi_\kappa(\Delta L^g)\ge \frac{2}{(1-\gamma)|T^g|}$, with $K$ if and only if
$\phi_\kappa(\Delta L^g)\ge \frac{1}{(1-\gamma)|T^g|}$, and with $\gamma$ if and only if
$\phi_\kappa(\Delta L^g)\le \frac{1+\eta^g}{(1-\gamma)|T^g|}$. An upward shift in the distribution of $\kappa$ weakly decreases data value.
\end{proposition}

Proposition~\ref{app:prop_cred_micro} endogenizes continuation credibility by linking it to the platform's incentive to honor the subsidy. When $T^g>0$, exclusive access disciplines the rejecting seller, so higher implementation costs or commission weaken credibility and threat-based fee extraction, while a larger reneging penalty strengthens them. Hence, escrow arrangements, third-party verification, and auditable access commitments can reinforce subsidy-backed threats when implementation costs are sufficiently small. When $T^g<0$, exclusive access instead improves the rejecting seller's outside option, so weaker credibility can benefit the platform. Higher $a$ or $K$ may then raise data value if the resulting fee-extraction gain exceeds their direct cost, while the effect of $\gamma$ depends on its direct and credibility effects. The thresholds in part~{\rm(B)} identify when these indirect effects dominate. This case is included for generality and does not arise in an active data-service equilibrium under the focal MDC and VDD regimes.

\subsection{Conditional Bound on Threat Value}
\label{sec:threat_elimination}
The baseline model studies how contractual credibility affects the extraction of an existing threat. We now examine when the threat could be eliminated. Consider a general consumer type space $\Theta$ endowed with an atomless probability measure $\mu$. For any regime $g$, let $\widehat{\Omega}^g\subseteq\Theta$ denote the set of consumers whose information can become usable when only one seller has data access, and let $q(\theta)\in[0,1]$ denote the quality of the corresponding signal. Define the quality-weighted usable unilateral information as $\beta_{\widehat{\Omega}^g,q}=\int_{\widehat{\Omega}^g}q(\theta)d\mu(\theta)$, denoted $\beta^g$ for simplicity.  Let $R_U^g(\widehat{\Omega}^g,q)$ denote the equilibrium sales revenue of the uninformed seller when only its rival can use these signals. Without usable unilateral information, the continuation coincides with the no-data benchmark and $R_U^g(\varnothing,0)=R_N^g$. We assume that unilateral information weakly disadvantages the uninformed seller, $R_U^g(\widehat{\Omega}^g,q)\le R_N^g$, and that its equilibrium revenue is Lipschitz continuous in the usable-signal profile: 
\[|R_U^g(\Omega,q)-R_U^g(\Omega',q')|\le L\int_\Theta|q(\theta)\mathbf{1}_{\Omega}(\theta)-q'(\theta)\mathbf{1}_{\Omega'}(\theta)|d\mu(\theta), \quad \mbox{for $L\ge0$.} \]

\begin{proposition}[Threat bound and elimination]
\label{prop:threat_discipline}
Under the preceding assumptions, the threat effect satisfies $0\le T^g\le L\beta^g$. Hence, $\beta^g=0$ is sufficient to eliminate the threat effect. If, in addition, $R_U^g(\widehat{\Omega}^g,q)<R_N^g$ whenever $\beta^g>0$, then $\beta^g=0$ is also necessary for threat elimination.
\end{proposition}

Proposition \ref{prop:threat_discipline} shows that eliminating threat-based monetization depends not on disclosure volume per se, but on \emph{usable information under unilateral access}. The condition that eliminating the usable information also restores the no-data participation outcome is essential: zero realized records alone need not eliminate a participation-driven threat.

\section{Proofs of Main Results}\label{app:sec_proof_main}
\noindent\textbf{Proof of Lemma \ref{lem:fees}.}
Fix $g\in\{\MDC,\VDD\}$. If seller 1 accepts, seller 2 obtains $(1-\gamma)R_B^g-F_{2\mid AC}^g$ from accepting and $(1-\gamma)R_U^g$ from rejecting. Hence, the maximal fee inducing acceptance is
$F_{2\mid AC}^g=(1-\gamma)(R_B^g-R_U^g)=(1-\gamma)(G^g+T^g)$. If seller 1 rejects, seller 2 obtains $(1-\gamma)R_I^g-F_{2\mid RE}^g$ from accepting and $(1-\gamma)R_N^g$ from rejecting. Thus, $F_{2\mid RE}^g=(1-\gamma)(R_I^g-R_N^g)$. Anticipating these continuation decisions, seller 1 obtains $(1-\gamma)R_B^g-F_1^g$ from accepting and $(1-\gamma)R_U^g$ from rejecting. Therefore, the maximum first-stage fee that induces acceptance is $F_1^g=(1-\gamma)(R_B^g-R_U^g)=(1-\gamma)(G^g+T^g)$. Thus, both sellers adopt, and each earns $(1-\gamma)R_U^g$ on the equilibrium path. \qed

\vspace{0.3cm}
\noindent\textbf{Proof of Proposition \ref{prop:data_value_decomposition}.}
If the platform withholds the data service, its profit is $\Pi_{p,\mathrm{no}}^g=2\gamma R_N^g$. Under provision, Lemma \ref{lem:fees} gives total data-service fees $2(1-\gamma)(G^g+T^g)$, so $\Pi_p^g=2\gamma R_B^g+2(1-\gamma)(G^g+T^g)$. Hence,
$V_{\mathrm{data}}^g=\Pi_p^g-\Pi_{p,\mathrm{no}}^g=2\gamma G^g+2(1-\gamma)(G^g+T^g)=2G^g+2(1-\gamma)T^g$.
Therefore, the platform provides the data service if and only if $V_{\mathrm{data}}^g\ge0$, equivalently, $G^g+(1-\gamma)T^g\ge0$. \qed

\vspace{0.3cm}
\noindent\textbf{Proof of Proposition \ref{prop:limited_commitment}.}
After seller 1 accepts, seller 2's binding participation constraint gives $F_{2\mid AC}^g=(1-\gamma)(R_B^g-R_U^g)=(1-\gamma)(G^g+T^g)$. Following seller 1's rejection, seller 2 accepts exclusive access when $F_{2\mid RE}^g=(1-\gamma)(R_I^g-R_N^g)$. If $R_I^g\ge R_N^g$, this continuation requires no subsidy and is implemented with probability one; if $R_I^g<R_N^g$, it is implemented with probability $\eta^g$. Hence, seller 1's expected rejection payoff is $(1-\gamma)[\xi^g(\eta^g)R_U^g+(1-\xi^g(\eta^g))R_N^g]$. Equating this with its acceptance payoff $(1-\gamma)R_B^g-F_1^g$ yields
$F_1^g=(1-\gamma)[G^g+\xi^g(\eta^g)T^g]$,
proving part {\rm(A)}. On the equilibrium path, fee revenue is
$F_1^g+F_{2\mid AC}^g=2(1-\gamma)G^g+[1+\xi^g(\eta^g)](1-\gamma)T^g$, while commission revenue increases by $2\gamma G^g$. Therefore, $V_{data}^g(\eta^g)=2G^g+[1+\xi^g(\eta^g)](1-\gamma)T^g$, proving part {\rm(B)}. When $R_I^g<R_N^g$, $\xi^g(\eta^g)=\eta^g$, so $\partial V_{data}^g(\eta^g)/\partial\eta^g=(1-\gamma)T^g$. Thus, data value weakly increases with credibility when $T^g\ge0$ and decreases when $T^g<0$. \qed

\vspace{0.3cm}
\noindent\textbf{Proof of Proposition \ref{prop:without_data}.}\label{app:NDC_proof}
Since entering and subsequently not purchasing yields zero utility, all consumers enter. Given $p_1,p_2\ge0$, seller 1's demand is determined by the participation boundary $v-p_1$ and the indifferent consumer $(1-p_1+p_2)/2$; seller 2 is symmetric.

(1) If $p_1+p_2\ge2v-1$, the participation constraints bind before the two demands meet, so each seller locally solves $\max_{p_i\ge0}p_i(v-p_i)$ subject to $p_i+p_j\ge2v-1$. The symmetric candidate is $p_i=v/2$ when $v<1$ and the boundary price $p_i=v-\tfrac12$ when $v\ge1$.

(2) If $p_1+p_2\le2v-1$, the market is covered and seller 1 solves $\max_{p_1\ge0}p_1(1-p_1+p_2)/2$ subject to $p_1+p_2\le2v-1$, with the symmetric problem for seller 2. The symmetric candidate is the boundary price $p_i=v-\tfrac12$ when $v<\tfrac32$ and the unconstrained price $p_i=1$ when $v\ge\tfrac32$. Comparing the feasible candidates across the two regions gives
$p_{i,\varnothing}^{\NDC}=v/2$ for $v<1$, $p_{i,\varnothing}^{\NDC}=v-\tfrac12$ for $1\le v<\tfrac32$, and $p_{i,\varnothing}^{\NDC}=1$ for $v\ge\tfrac32$.

Let $p_{\varnothing}^{\NDC}$ denote this common equilibrium price and $d^{\NDC}=\min\{v-p_{\varnothing}^{\NDC},\tfrac12\}$ each seller's demand. Then $r_i^{\NDC}=p_{\varnothing}^{\NDC}d^{\NDC}$, $\pi_i^{\NDC}=(1-\gamma)r_i^{\NDC}$, $\Pi_p^{\NDC}=2\gamma r_i^{\NDC}$, $CS^{\NDC}=2\int_0^{d^{\NDC}}(v-p_{\varnothing}^{\NDC}-x)dx$, and $SW^{\NDC}=CS^{\NDC}+2r_i^{\NDC}$. Moreover, $D^{\NDC}=2d^{\NDC}=\min\{v,1\}$. 
\qed

\vspace{0.3cm}
\noindent\textbf{Proof of Proposition \ref{prop:equilibrium_mdc}.}\label{app:subsec_proof_mdc}
The MDC subgame equilibria in Online Appendix \ref{app:sec_MDC_subgame} give $R_B^{\MDC}=r^{PPM}$, $R_I^{\MDC}=r_1^{PNM}$, $R_U^{\MDC}=r_2^{PNM}$, and $R_N^{\MDC}=r^{\NDC}$. Substituting the corresponding subgame revenues yields Table \ref{tab:data_fee_mdc}. 

\begin{table}[!ht]
\vspace{-0.25in}
    \small
    \centering
    \caption{$F_1^{\MDC}$ and Its Components in Different Market Conditions}
    \label{tab:data_fee_mdc}
    \renewcommand{\arraystretch}{1.05}
    \resizebox{0.8\textwidth}{!}{
    \begin{tabular}{l c c c c}
        \toprule
        \textbf{Conditions}
        & $v-c\le \tfrac{1}{2}$
        & $\tfrac{1}{2}<v-c<\tfrac{3}{4}$
        & $\tfrac{3}{4}\le v-c\le1$
        & $v-c>1$ \\
        \midrule

        \makecell[l]{$G^{\MDC}$}
        & $-r^{\NDC}$
        & $\tfrac{(2v-2c-1)^2}{4}-r^{\NDC}$
        & $\tfrac{(2v-2c-1)^2}{4}-r^{\NDC}$
        & $\tfrac{1}{4}-r^{\NDC}$ \\

        \makecell[l]{$T^{\MDC}$}
        & $r^{\NDC}$
        & $r^{\NDC}$
        & $r^{\NDC}-\tfrac{1}{8}$
        & $r^{\NDC}-\tfrac{1}{8}$ \\

        \midrule

        $F_1^{\MDC}=(1-\gamma)(G^{\MDC}+T^{\MDC})$
        & $0$
        & $(1-\gamma)\tfrac{(2v-2c-1)^2}{4}$
        & $(1-\gamma)\left[\tfrac{(2v-2c-1)^2}{4}
          -\tfrac{1}{8}\right]$
        & $\tfrac{1-\gamma}{8}$ \\
        \bottomrule
    \end{tabular}
    }

    \vspace{1pt}
    \parbox{\textwidth}{%
        \footnotesize
        \emph{Notes:} Seller revenue under NDC is
        $r^{\NDC}=\tfrac{v^2}{4}$ if $v\le1$,
        $r^{\NDC}=\tfrac{2v-1}{4}$ if $1<v<\tfrac32$,
        and $r^{\NDC}=\tfrac12$ if $v\ge\tfrac32$.
    }
    \vspace{-0.45in}
\end{table}

Hence,
$G^{\MDC}=r^{PPM}-r^{\NDC}<0$ and $T^{\MDC}=r^{\NDC}-r_2^{PNM}>0$. By Lemma \ref{lem:fees}, bilateral adoption is implemented with $F_1^{\MDC}=F_{2\mid AC}^{\MDC}=(1-\gamma)(r^{PPM}-r_2^{PNM})=(1-\gamma)(G^{\MDC}+T^{\MDC})$
and $F_{2\mid RE}^{\MDC}=(1-\gamma)(r_1^{PNM}-r^{\NDC})$, proving part {\rm(A)}. By Proposition \ref{prop:data_value_decomposition}, the platform's incremental profit from provision is $2[G^{\MDC}+(1-\gamma)T^{\MDC}]$. Since $T^{\MDC}>0$, provision occurs if and only if $\gamma\le1+G^{\MDC}/T^{\MDC}$, proving part {\rm(B)}. Conditional on provision, both sellers adopt and the ``PP'' subgame applies. Lemma \ref{app:lem_mdc_PP} then gives the equilibrium entry set and the personalized prices.\qed 

\vspace{0.3cm}
\noindent\textbf{Proof of Proposition \ref{prop:profit_campare_M_N}.}
Based on Proposition \ref{prop:equilibrium_mdc}, when the platform does not monetize data, the equilibrium coincides with NDC. When it monetizes data, the resulting market outcomes are reported in Table \ref{tab:mandatory_equilibrium}.

\begin{table}[!ht]
    \small
    \centering
    \renewcommand{\arraystretch}{1.2}
    \caption{Equilibrium Outcomes when the Platform Provides Data Service under MDC}
    \label{tab:mandatory_equilibrium}
    \resizebox{\textwidth}{!}{
    \begin{tabular}{l c c c c}
        \toprule
        \textbf{Conditions} & $\pi_1^{\MDC}=\pi_2^{\MDC}$ & $\Pi_p^{\MDC}$ & $CS^{\MDC}$ & $SW^{\MDC}$ \\
        \midrule
        $\tfrac{1}{2} < v \le \tfrac{3}{4}, c < v - \tfrac{1}{2}, \gamma \le \tfrac{(2v-2c-1)^2}{v^2}$ 
        & $0$ 
        & $\tfrac{(2v-2c-1)^2}{2}$ 
        & $\tfrac{(2v-2c-1)^2}{4}$ 
        & $\tfrac{3(2v-2c-1)^2}{4}$ \\
        
        $\tfrac{2+\sqrt2}{4} < v \le 1, c \le v - \tfrac{2+\sqrt{2}}{4}, \gamma \le \tfrac{2(2v-2c-1)^2 - 1}{2v^2 - 1}$
        & $\tfrac{1-\gamma}{8}$ 
        & $\tfrac{(2v-2c-1)^2}{2} - \tfrac{1-\gamma}{4}$ 
        & $\tfrac{(2v-2c-1)^2}{4}$ 
        & $\tfrac{3(2v-2c-1)^2}{4}$ \\

        $\tfrac{3}{4} < v \le 1, v - \tfrac{3}{4} < c < v - \tfrac{1}{2}, \gamma \le \tfrac{(2v-2c-1)^2}{v^2}$
        & $0$ 
        & $\tfrac{(2v-2c-1)^2}{2}$ 
        & $\tfrac{(2v-2c-1)^2}{4}$ 
        & $\tfrac{3(2v-2c-1)^2}{4}$ \\

        $1<v<\tfrac{3}{2}, 0 <c <v-1, \gamma \le \tfrac{1}{4v-3}$
        & $\tfrac{1-\gamma}{8}$ 
        & $\tfrac{1+\gamma}{4}$ 
        & $v-c-\tfrac{3}{4}$ 
        & $v-c-\tfrac{1}{4}$ \\

        $1<v<\tfrac{3}{2}, v - \tfrac{3}{4} < c < v - \tfrac{1}{2}, \gamma \le \tfrac{(2v-2c-1)^2}{2v-1}$
        & $0$ 
        & $\tfrac{(2v-2c-1)^2}{2}$ 
        & $\tfrac{(2v-2c-1)^2}{4}$ 
        & $\tfrac{3(2v-2c-1)^2}{4}$ \\

        $1<v<\tfrac{3}{2}, v - 1\le c \le v - \tfrac{2+\sqrt{2}}{4}, \gamma \le \tfrac{2(2v-2c-1)^2 - 1}{4v-3}$
        & $\tfrac{1-\gamma}{8}$ 
        & $\tfrac{(2v-2c-1)^2}{2} - \tfrac{1-\gamma}{4}$ 
        & $\tfrac{(2v-2c-1)^2}{4}$ 
        & $\tfrac{3(2v-2c-1)^2}{4}$ \\

        $v\ge \tfrac{3}{2}, 0 <c <v-1, \gamma\le \tfrac{1}{3}$
        & $\tfrac{1-\gamma}{8}$ 
        & $\tfrac{1+\gamma}{4}$ 
        & $v-c-\tfrac{3}{4}$ 
        & $v-c-\tfrac{1}{4}$ \\

        $v\ge \tfrac{3}{2}, v - 1 \le c \le v - \tfrac{2+\sqrt{2}}{4}, \gamma \le \tfrac{2(2v-2c-1)^2 - 1}{3}$
        & $\tfrac{1-\gamma}{8}$ 
        & $\tfrac{(2v-2c-1)^2}{2} - \tfrac{1-\gamma}{4}$ 
        & $\tfrac{(2v-2c-1)^2}{4}$ 
        & $\tfrac{3(2v-2c-1)^2}{4}$ \\

        $v\ge \tfrac{3}{2}, v-\tfrac{3}{4} < c < v-\tfrac{1}{2}, \gamma \le \tfrac{(2v-2c-1)^2}{2}$
        & $0$ 
        & $\tfrac{(2v-2c-1)^2}{2}$ 
        & $\tfrac{(2v-2c-1)^2}{4}$ 
        & $\tfrac{3(2v-2c-1)^2}{4}$ \\
        \bottomrule
    \end{tabular}
    }
    \vspace{-0.15in}
\end{table}

\noindent\textbullet{\bf Market coverage.} Under NDC, $D^{\NDC}=\min\{v,1\}$. Conditional on data monetization under MDC, Lemma \ref{app:lem_mdc_PP} gives $D^{\MDC}=[\min\{2(v-c)-1,1\}]_+$. If $v<1$, then $D^{\MDC}<v=D^{\NDC}$. If $v\ge1$, then $D^{\NDC}=1$ and $D^{\MDC}\le1$, with equality if and only if $v-c\ge1$. Hence, $D^{\MDC}\le D^{\NDC}$, with equality if and only if $v-c\ge1$.

\noindent\textbullet{\bf Profits of the sellers and the platform.} For seller $i\in\{1,2\}$, under MDC with data monetization, we have $\pi_i^{\NDC}-\pi_i^{\MDC}=(1-\gamma)(r^{\NDC}-r_2^{PNM})=(1-\gamma)T^{\MDC}\ge 0$. Hence, $\pi_i^{\NDC}>\pi_i^{\MDC}$.
The platform monetizes data only if $\Pi_p^{\MDC}\ge\Pi_p^{\NDC}$. Therefore, $\Pi_p^{\MDC}\ge\Pi_p^{\NDC}$ holds.

\noindent\textbullet{\bf Consumer surplus and social welfare.} 
If $v\le1$, $CS^{\MDC}-CS^{\NDC}=[(2v-2c-1)^2-v^2]/4<0$ and $SW^{\MDC}-SW^{\NDC}=3[(2v-2c-1)^2-v^2]/4<0$. If $1<v<\tfrac32$ and $c<v-1$, then $CS^{\MDC}-CS^{\NDC}=v-c-1>0$ and $SW^{\MDC}-SW^{\NDC}=-c<0$; if $v-1\le c<v-\tfrac12$, then $2v-2c-1\le1$, so $CS^{\MDC}\le CS^{\NDC}$ and $SW^{\MDC}-SW^{\NDC}\le1-v<0$. Finally, if $v\ge\tfrac32$ and $c<v-1$, then $CS^{\MDC}-CS^{\NDC}=\tfrac12-c$, while $SW^{\MDC}-SW^{\NDC}=-c<0$; if $v-1\le c<v-\tfrac12$, then $CS^{\MDC}-CS^{\NDC}\le\tfrac32-v\le0$ and $SW^{\MDC}-SW^{\NDC}\le1-v<0$.
Therefore, $CS^{\MDC}>CS^{\NDC}$ exactly when $1<v<\tfrac32$ and $v-c>1$, or $v\ge\tfrac32$ and $c<\tfrac12$, while $SW^{\MDC}<SW^{\NDC}$ throughout the monetization regions. 
\qed

\vspace{0.3cm}
\noindent\textbf{Proof of Proposition \ref{prop:equilibrium_vdd}.}
\label{app:subsec_equilibrium_vdd_proof}
Online Appendix \ref{app:sec_VDD_subgame} gives $R_B^{\VDD}=r^{PPV}$ and $R_I^{\VDD}=R_U^{\VDD}=R_N^{\VDD}=r^{\NDC}$ because unilateral access induces no disclosure. Hence, $T^{\VDD}=0$. By Lemma \ref{lem:fees}, $F_1^{\VDD}=F_{2\mid AC}^{\VDD}=(1-\gamma)(r^{PPV}-r^{\NDC})=(1-\gamma)G^{\VDD}$ and $F_{2\mid RE}^{\VDD}=0$, proving part {\rm(A)}. The platform's incremental profit is therefore $2G^{\VDD}$, so it provides the service if and only if $G^{\VDD}\ge0$ and the selected opt-in set is nonempty.

Define $\mathcal P_1=\{(v,c):\tfrac12<v\le1,\ c\le v-\tfrac12\}$ and $\mathcal P_2=\{(v,c):v>1,\ c=\min\{v-\tfrac12,1\}\}$. Let $h_\pm(v)=[v+3\pm\sqrt{-16v^2+40v-8}]/17$, $\bar v_1=(3+\sqrt2)/2$, and $\bar v_2=(5+\sqrt{17})/4$, and define
\[
\begin{aligned}
\mathcal P_3={}&\left\{(v,c):\tfrac32<v\le\bar v_1,\ 
v-1+\tfrac1{4(v-1)}\le c\le v-1+h_+(v)\right\}\\
&\cup\left\{(v,c):\bar v_1<v\le\tfrac94,\ 
2v-3\le c\le v-1+h_+(v)\right\}\\
&\cup\left\{(v,c):\tfrac94<v\le\bar v_2,\ 
v-1+h_-(v)\le c\le v-1+h_+(v)\right\}.
\end{aligned}
\]

(1) For $\tfrac12<v\le1$ and $0<c\le v-\tfrac12$, $G^{\VDD}=(v-c-\tfrac12)^2\ge0$ when $c\ge\tfrac{3v}{2}-1$, and $G^{\VDD}=(5v-4c-3)(1-v)/4\ge0$ otherwise. The opt-in set is nonempty, including $\{1/2\}$ at $c=v-\tfrac12$, giving $\mathcal P_1$.

(2) For $1<v\le\tfrac32$, $G^{\VDD}=(1-v)/2<0$ when $c\le v-1$ and $G^{\VDD}=(v-1)(1-2v+2c)/2<0$ when $v-1<c<v-\tfrac12$. At $c=v-\tfrac12$, $G^{\VDD}=0$ and $\Omega^{PPV}=\{1/2\}$, giving the corresponding part of $\mathcal P_2$.

(3) For $v>\tfrac32$ and $c<1$, Lemma \ref{app:lem_vdd_pp} gives $G^{\VDD}=-1/4<0$ under full disclosure. If $v-1<c<1$, then $G^{\VDD}=-1/4+(v-1)(1-v+c)<-(v-\tfrac32)^2<0$. Hence, provision does not occur for $c<1$. On $\mathcal P_3$, instead, $\Omega^{PPV}=[1-v+c,v-c]$ and $r^{PPV}=1/4+(v-1)(1-v+c)$; the defining bounds imply $c\ge v-1+1/[4(v-1)]$, so $G^{\VDD}=-1/4+(v-1)(1-v+c)\ge0$. Finally, when $c\ge1$ and $(v,c)\notin\mathcal P_3$, the selected PP equilibrium has $G^{\VDD}=0$, with $\Omega^{PPV}=\{1/2\}$ at $c=1$ and $\Omega^{PPV}=\varnothing$ for $c>1$. Thus, provision occurs exactly on $\mathcal P_1\cup\mathcal P_2\cup\mathcal P_3$, proving part {\rm(B)}.

Conditional on provision, Lemma \ref{app:lem_vdd_pp} gives, for $(v,c)\in\mathcal P_1\cup\mathcal P_3$, $p_{i,\varnothing}^{\VDD}=\max\{v/2,2v-c-1\}$ and $\Omega^{\VDD}=[\max\{0,1-v+c\},\min\{1,v-c\}]$; for $(v,c)\in\mathcal P_2$, $p_{i,\varnothing}^{\VDD}=p_{i,\varnothing}^{\NDC}$ and $\Omega^{\VDD}=\{1/2\}$. Opt-in consumers face $p_1^{\VDD}(x)=\max\{1-2x,0\}$ and $p_2^{\VDD}(x)=\max\{2x-1,0\}$, while all others opt out.

It remains to compare personalized and uniform prices. For $(v,c)\in\mathcal P_1\cup\mathcal P_3$, the largest personalized price charged by seller 1 to an opt-in consumer is $2v-2c-1$, while $p_{1,\varnothing}^{\VDD}\ge2v-c-1>2v-2c-1$ because $c>0$; the argument for seller 2 is symmetric. For $(v,c)\in\mathcal P_2$, the only opt-in consumer is at $x=1/2$, where $p_1^{\VDD}(1/2)=p_2^{\VDD}(1/2)=0<p_{i,\varnothing}^{\NDC}$. Hence, $p_i^{\VDD}(x)<p_{i,\varnothing}^{\VDD}$ for every $x\in\Omega^{\VDD}$ and $i=1,2$.
\qed

\vspace{0.3cm}
\noindent\textbf{Proof of Proposition \ref{prop:price_coverage_cmp}.}
Based on Proposition \ref{prop:equilibrium_vdd}, when the platform does not monetize data, the equilibrium coincides with NDC. When it monetizes data, the resulting market outcomes are reported in Table \ref{tab:voluntary_equilibrium}.
\begin{table}[htbp]
\vspace{-0.2in}
\small
\centering
\renewcommand{\arraystretch}{1.15}
\caption{Equilibrium Outcomes when the Platform Provides Data Service under VDD}
\label{tab:voluntary_equilibrium}
\resizebox{\textwidth}{!}{
\begin{tabular}{l c c c c}
\toprule
\textbf{Condition}
& $\pi_1^{VDD}=\pi_2^{VDD}$
& $\Pi_p^{VDD}$
& $CS^{VDD}$
& $SW^{VDD}$ \\
\midrule

$(v,c)\in\mathcal P_1, c\ge \tfrac{3v}{2}-1$
& $(1-\gamma)\tfrac{v^2}{4}$
& $\tfrac{\gamma v^2+(2v-2c-1)^2}{2}$
& $\tfrac{v^2+(2v-2c-1)^2}{4}$
& $\tfrac{3v^2+3(2v-2c-1)^2}{4}$ \\

$(v,c)\in\mathcal P_1, c<\tfrac{3v}{2}-1$
& $(1-\gamma)\tfrac{v^2}{4}$
& $\tfrac{\gamma v^2+(1-v)(5v-4c-3)}{2}$
& $\tfrac{4(1-v+c)^2+(2v-2c-1)^2}{4}$
& $\tfrac{4v^2+12c^2-(2v-1+2c)^2}{4}$ \\

$(v,c)\in\mathcal P_3$
& $(1-\gamma)\tfrac{1}{2}$
& $\gamma-\tfrac12+2(v-1)(1-v+c)$
& $\tfrac{4(1-v+c)^2+(2v-2c-1)^2}{4}$
& $\tfrac{4v^2+12c^2-(2v-1+2c)^2}{4}$ \\

$(v,c)\in\mathcal P_2$
& $(1-\gamma)r^{NDC}$
& $2\gamma r^{NDC}$
& $CS^{NDC}$
& $SW^{NDC}$ \\

\bottomrule
\end{tabular}
}
\vspace{-0.2in}
\end{table}

{\bf 1. Impact on prices.} 
\emph{\textbullet\ VDD vs. MDC.} By Propositions \ref{prop:equilibrium_mdc} and \ref{prop:equilibrium_vdd}, on $\mathcal P_1\cup\mathcal P_3$, $\Omega^{\VDD}=\Omega^{\MDC}=[\max{0,1-v+c},\min{1,v-c}]$, and opt-in consumers face the same personalized prices $p_1(x)=\max{1-2x,0}$ and $p_2(x)=\max{2x-1,0}$. On $\mathcal P_2$ with $v\le\tfrac32$, $c=v-\tfrac12$ and both entry sets equal ${1/2}$. If $v>\tfrac32$, then $c=1$, $\Omega^{\VDD}={1/2}$, and $\Omega^{\MDC}=[\max{0,2-v},\min{1,v-1}]$. 

\emph{\noindent\textbullet\ VDD vs. NDC.} consider seller 1; seller 2 follows symmetrically. On $\mathcal P_1$, $p_{1,\varnothing}^{\NDC}=v/2$ and $p_{1,\varnothing}^{\VDD}=\max{v/2,2v-c-1}$, so the inequality is strict if and only if $c<3v/2-1$, which requires $v>2/3$. On $\mathcal P_2$, the two uniform prices coincide. On $\mathcal P_3$, $p_{1,\varnothing}^{\VDD}=2v-c-1>1=p_{1,\varnothing}^{\NDC}$. Hence, the stated uniform-price comparison follows. For an opt-in consumer with $x\le1/2$, $p_1^{\VDD}(x)=1-2x$. On $\mathcal P_1$, this exceeds $p_{1,\varnothing}^{\NDC}=v/2$ if and only if $x<(2-v)/4$. Since the left disclosure boundary is $1-v+c$, such consumers exist if and only if $1-v+c<(2-v)/4$, or $c<3v/4-1/2$, which requires $v>2/3$. Thus, the relevant interval is $[1-(v-c),(2-v)/4)$. On $\mathcal P_2$, the only opt-in consumer is at $x=1/2$ and pays zero, while on $\mathcal P_3$, $p_1^{\VDD}(x)\le1=p_{1,\varnothing}^{\NDC}$. Symmetry gives seller 2's interval $((2+v)/4,v-c]$.

{\bf 2. Impact on market coverages.}
The comparisons of the market coverage across alternative privacy regulations are summarized in Table \ref{tab:market_coverage}. It follows that when the platform monetizes data under both MDC and VDD, $D^{\VDD}=\min\{D^{\MDC}+D^{\NDC},1\}$ is satisfied.\qed
\begin{table}[!ht]
\vspace{-0.2in}
\small
\centering
\renewcommand{\arraystretch}{1.2}
\setlength{\tabcolsep}{6pt}
\caption{Market Coverage Across Alternative Privacy Regulations}
\label{tab:market_coverage}
\begin{tabular}{ll}
\toprule
\textbf{Regime and condition}
& \textbf{Market coverage} \\
\midrule

NDC
& $D^{NDC}=\min\{v,1\}$ \\

MDC, service provided
& $D^{MDC}=[\min\{2v-2c-1,1\}]_+$ \\

MDC, service withheld
& $D^{MDC}=\min\{v,1\}$ \\

VDD, service provided
& $D^{VDD}=\min\{1,v+[2v-2c-1]_+\}$ \\

\bottomrule
\vspace{-0.4in}
\end{tabular}
\end{table}

\vspace{0.3cm}
\noindent\textbf{Proof of Proposition \ref{prop:profit_campare_vdd}.}\label{app:proof_cmp_profit}
Conditional on provision under VDD, Lemma \ref{lem:fees} gives $\pi_i^{\VDD}=(1-\gamma)r^{\NDC}=\pi_i^{\NDC}$. Under MDC, if the service is withheld, $\pi_i^{\MDC}=\pi_i^{\NDC}$; if it is provided, $\pi_i^{\MDC}=(1-\gamma)r_2^{PNM}\le(1-\gamma)r^{\NDC}$ because $T^{\MDC}=r^{\NDC}-r_2^{PNM}>0$. Hence, $\pi_i^{\VDD}=\pi_i^{\NDC}\ge\pi_i^{\MDC}$.

Since VDD provision requires $G^{\VDD}\ge0$, $\Pi_p^{\VDD}=\Pi_p^{\NDC}+2G^{\VDD}\ge\Pi_p^{\NDC}$. Under MDC, $\Pi_p^{\MDC}=\Pi_p^{\NDC}+2\max\{G^{\MDC}+(1-\gamma)T^{\MDC},0\}$. Because $G^{\VDD}\ge0$, $\Pi_p^{\VDD}\ge\Pi_p^{\MDC}$ if and only if
$G^{\VDD}\ge G^{\MDC}+(1-\gamma)T^{\MDC}$.
Using $G^{\VDD}-G^{\MDC}=r^{PPV}-r^{PPM}$ and $T^{\MDC}>0$, this is equivalent to
$\gamma\ge1-\frac{r^{PPV}-r^{PPM}}{T^{\MDC}}$.
Since $\gamma\ge0$, the condition can be written as $\gamma\ge\gamma^{VM}(v,c):=\max\{0,1-\frac{r^{PPV}-r^{PPM}}{T^{\MDC}}\}$. 
\qed

\vspace{0.3cm}
\noindent\textbf{Proof of Proposition \ref{prop:cs_sw_cmp}.}

\emph{\noindent\textbullet\ VDD vs. NDC.}
Define $\Delta_{VN}^{CS}=CS^{\VDD}-CS^{\NDC}$ and
$\Delta_{VN}^{SW}=SW^{\VDD}-SW^{\NDC}$. We consider the regions where the platform provides the service under VDD.

(1) If $(v,c)\in\mathcal P_1$ and
$c\ge\tfrac{3v}{2}-1$, then $\Delta_{VN}^{CS}=\frac{(2v-2c-1)^2}{4}\ge0,\Delta_{VN}^{SW}=\frac{3(2v-2c-1)^2}{4}\ge0.$

(2) If $(v,c)\in\mathcal P_1$ and $c<\tfrac{3v}{2}-1$, then
$\Delta_{VN}^{CS}=\frac{h_2(c)}{4}$, where $h_2(c)=8c^2+(12-16v)c+5-12v+7v^2$.
For $\tfrac23<v\le\tfrac{\sqrt2}{2}$, $h_2(c)\ge0$ throughout this region. For
$\tfrac{\sqrt2}{2}<v\le1$, $h_2(c)\ge0$ if and only if
$c\ge\tfrac{4v-3+\sqrt{2v^2-1}}{4}.$
Likewise,
$\Delta_{VN}^{SW}=\frac{h_3(c)}{4}$, where $h_3(c)=8c^2+(4-8v)c-1+4v-3v^2$. For $\tfrac23<v\le\tfrac{6+\sqrt6}{10}$, $h_3(c)\ge0$ throughout this region. For $\tfrac{6+\sqrt6}{10}<v\le1$, $h_3(c)\ge0$ if
$c\le\tfrac{2v-1-\sqrt{10v^2-12v+3}}{4}$ or if $c\ge\frac{2v-1+\sqrt{10v^2-12v+3}}{4}.$ At $v=1$, both differences equal $-c(1-2c)<0$ for $c<\tfrac12$.

(3) If $(v,c)\in\mathcal P_3$, then
$\Delta_{VN}^{CS}=-c(2v-2c-1)-2G^{\VDD}<0$ and 
$\Delta_{VN}^{SW}=c(2c+1-2v)=-c(2v-2c-1)<0.$

(4) If $(v,c)\in\mathcal P_2$, only the midpoint discloses and aggregate outcomes equal those under NDC. 

\emph{\noindent\textbullet\ VDD vs. MDC.}
Suppose the platform provides the service under both regimes. Define
$\Delta_{VM}^{CS}=CS^{\VDD}-CS^{\MDC}$ and
$\Delta_{VM}^{SW}=SW^{\VDD}-SW^{\MDC}$.

(1) If $(v,c)\in\mathcal P_1$ and $c\ge\tfrac{3v}{2}-1$, then
$\Delta_{VM}^{CS}=\frac{v^2}{4}>0$ and $\Delta_{VM}^{SW}=\frac{3v^2}{4}>0.$

(2) If $(v,c)\in\mathcal P_1$ and $c<\tfrac{3v}{2}-1$, or if $(v,c)\in\mathcal P_3$, let
$y=1-v+c$, then $y>0$ and $2v-y>0$. Thus,
$\Delta_{VM}^{CS}=y^2>0$ and $\Delta_{VM}^{SW}=y(2v-y)>0$.

(3) If $(v,c)\in\mathcal P_2$, the outcomes under VDD equal those under NDC. Proposition \ref{prop:profit_campare_M_N} implies that consumer surplus and social welfare under MDC are strictly below their NDC levels whenever MDC data provision. 

Therefore, whenever both regimes provide the service,
$CS^{\VDD}>CS^{\MDC}$ and $SW^{\VDD}>SW^{\MDC}$.
\qed

\vspace{0.3cm}
\noindent\textbf{Proof of Proposition \ref{prop:3win}.}
A Pareto improvement over both NDC and MDC requires $\pi_i^{\VDD}\ge\max\{\pi_i^{\NDC},\pi_i^{\MDC}\}$ for $i=1,2$, $\Pi_p^{\VDD}\ge\max\{\Pi_p^{\NDC},\Pi_p^{\MDC}\}$, and $CS^{\VDD}\ge\max\{CS^{\NDC},CS^{\MDC}\}$.
 Under VDD, $\pi_i^{\VDD}=\pi_i^{\NDC}$. Proposition \ref{prop:profit_campare_vdd} further gives $\pi_i^{\VDD}\ge\pi_i^{\MDC}$. Hence, the seller condition always holds whenever VDD monetizes data.
By Proposition \ref{prop:cs_sw_cmp}, apart from the zero-measure disclosure cases treated below, $CS^{\VDD}\ge CS^{\NDC}$ under VDD provision if and only if
$\tfrac{1}{2}<v\le\tfrac{\sqrt{2}}{2}$ and $c<v-\tfrac{1}{2}$, or
$\tfrac{\sqrt{2}}{2}<v<1$ and
$\tfrac{4v-3+\sqrt{2v^2-1}}{4}\le c<v-\tfrac{1}{2}$. The remaining weak-comparison cases are the upper boundary
$c=v-\tfrac12$ of $\mathcal P_1$ and $\mathcal P_2$. In these cases only the midpoint discloses and all aggregate VDD outcomes equal NDC, so VDD does not yield a Pareto improvement over NDC.
Moreover, $CS^{\VDD}>CS^{\MDC}$ whenever MDC also monetizes data, while $CS^{\MDC}=CS^{\NDC}$ otherwise. Therefore, the
consumer condition is exactly the set of conditions above \((v,c)\).
The Proposition \ref{prop:profit_campare_vdd} shows that $\Pi_p^{\VDD}\ge\Pi_p^{\NDC}$. Comparing with MDC, if MDC does not monetize data, $\Pi_p^{\MDC}=\Pi_p^{\NDC}$ and the platform condition holds automatically. If MDC also monetizes data, Proposition \ref{prop:profit_campare_vdd} implies that, within the consumer conditions above, $\Pi_p^{\VDD}\ge\Pi_p^{\MDC}$ if and only if $\gamma\ge\gamma^{\mathrm{win}}(v,c)$, where
\[
\gamma^{\mathrm{win}}(v,c)=
\begin{cases}
\tfrac{(3v-2c-2)^2}{v^2},
& \tfrac{2}{3}<v<1,\quad c<\tfrac{3v}{2}-1,\\
0,
& \text{otherwise within the above $CS^{\VDD}\ge\max\{CS^{\NDC},CS^{\MDC}\}$ regions}.
\end{cases}
\]

Therefore, VDD yields a Pareto improvement over both NDC and MDC if
$\tfrac{1}{2}<v\le\tfrac{\sqrt{2}}{2}$ and $c<v-\tfrac{1}{2}$, or
$\tfrac{\sqrt{2}}{2}<v<1$ and
$\tfrac{4v-3+\sqrt{2v^2-1}}{4}\le c<v-\tfrac{1}{2}$,
with $\gamma\ge \gamma^{\mathrm{win}}(v,c)$.
\qed

\vspace{0.3cm}
\noindent\textbf{Proof of Proposition \ref{cor:up}.}
At $p_{i,\varnothing}=p_{i,\varnothing}^{\NDC}$, the disclosure conditions in Lemma \ref{app:lem_vdd_pp} give positive-mass disclosure if $c<\min\{v-\tfrac12,1\}$, with $\Omega^{\UP}=[\hat{x}^{\UP},1-\hat{x}^{\UP}]$, where $\hat{x}^{\UP}=\max\{1-v+c,(1+c-p_{\varnothing}^{\NDC})/2\}$. For simplicity, denote $p_{\varnothing}^{\NDC}=p_{i,\varnothing}^{\NDC},i\in\{1,2\}$. These prices are optimal under the cap. If $\hat{x}^{\UP}\ge\min\{v-p_{\varnothing}^{\NDC},\tfrac12\}$, then $p_{\varnothing}^{\NDC}=v/2$, and any feasible price $p$ yields anonymous revenue at most $p(v-p)\le v^2/4$. Otherwise, $\hat{x}^{\UP}=(1+c-p_{\varnothing}^{\NDC})/2\ge c/2$. A deviation with $p\ge c$ yields at most $p\hat{x}^{\UP}\le p_{\varnothing}^{\NDC}\hat{x}^{\UP}$, while one with $p<c$ yields at most $p[\hat{x}^{\UP}+(c-p)/2]\le c\hat{x}^{\UP}<p_{\varnothing}^{\NDC}\hat{x}^{\UP}$. Any symmetric price below $p_{\varnothing}^{\NDC}$ admits a profitable small increase. Hence, $p_{i,\varnothing}^{\UP}=p_{i,\varnothing}^{\NDC}$. Each seller then earns
$r^{\UP}=p_{\varnothing}^{\NDC}\min\{\hat{x}^{\UP},v-p_{\varnothing}^{\NDC},\tfrac12\}+(\tfrac12-\hat{x}^{\UP})^2$.
For $\tfrac12<v<1$ and $0<c<v-\tfrac12$, this gives $G^{\UP}=(v-c-\tfrac12)^2$ if $c\ge\tfrac{3v}{2}-1$, and $G^{\UP}=(4v-5v^2+4c^2)/16$ otherwise. Thus, $G^{\UP}\ge0$ exactly when $4c^2\ge5v^2-4v$. For $v\ge1$, positive-mass disclosure gives $G^{\UP}=[c^2-(p_{\varnothing}^{\NDC})^2]/4<0$. At $v>\tfrac12$ and $c=\min\{v-\tfrac12,1\}$, only $x=\tfrac12$ discloses and $G^{\UP}=0$; otherwise the disclosure set is empty. Since exclusive access still induces no disclosure, the threat effect is zero. The maintained provision rule gives part {\rm(A)}.

The cap preserves every consumer's NDC anonymous-purchase option, while disclosure is voluntary. Hence, $U^{\UP}(x)\ge U^{\NDC}(x)$ for all $x$ and $CS^{\UP}\ge CS^{\NDC}$. Under provision, $SW^{\UP}-SW^{\NDC}=CS^{\UP}-CS^{\NDC}+2G^{\UP}\ge0$; without provision, UP coincides with NDC. This proves part {\rm(B)}.

If MDC also provides, its personalized offers remain available through disclosure under UP, while UP additionally permits anonymous purchasing without privacy cost. Thus, $U^{\UP}(x)\ge U^{\MDC}(x)$ and $CS^{\UP}\ge CS^{\MDC}$. Moreover, Lemma \ref{lem:fees} gives $\pi_i^{\UP}=\pi_i^{\NDC}$ and $\pi_i^{\MDC}=(1-\gamma)(r^{\NDC}-T^{\MDC})\le\pi_i^{\NDC}$, proving the first statement in part {\rm(C)}. For platform profit, endogenous MDC provision gives $\Pi_p^{\MDC}=\Pi_p^{\NDC}+2\max\{G^{\MDC}+(1-\gamma)T^{\MDC},0\}$, whereas UP provision gives $\Pi_p^{\UP}=\Pi_p^{\NDC}+2G^{\UP}$. Since $G^{\UP}\ge0$, $\Pi_p^{\UP}\ge\Pi_p^{\MDC}$ if and only if $r^{\UP}-r^{\MDC}\ge(1-\gamma)T^{\MDC}$, equivalently, $\gamma\ge\gamma^{UM}(v,c)$. Positive-mass UP provision requires $v<1$ and $c<v-\tfrac12$. A positive measure of central consumers then obtains positive utility under UP while being unserved under NDC, so $CS^{\UP}>CS^{\NDC}$. If MDC provides, $\pi_i^{\UP}>\pi_i^{\MDC}$ because $T^{\MDC}>0$; otherwise MDC coincides with NDC. Thus, positive-mass disclosure together with $\gamma\ge\gamma^{UM}(v,c)$ yields a Pareto improvement over both benchmarks. Conversely, without positive-mass disclosure, UP reproduces NDC aggregate payoffs, while if $\gamma<\gamma^{UM}(v,c)$ the platform is worse off than under MDC. The conditions are necessary and sufficient. \qed

\noindent\textbf{Proof of Proposition \ref{app:prop_endogenous_exclusive_access}.}
(1) Optimal data access.
Under no access, platform profit is $\Pi_N^g=2\gamma R_N^g$. Under bilateral access, Lemma \ref{lem:fees} gives $\Pi_B^g=2R_B^g-2(1-\gamma)R_U^g$, so $\Pi_B^g-\Pi_N^g=2[G^g+(1-\gamma)T^g]$. To implement exclusive access for seller 1, set $F_{2\mid AC}^g>(1-\gamma)(R_B^g-R_U^g)$, $F_1^g=(1-\gamma)(R_I^g-R_U^g)$, and $F_{2\mid RE}^g=(1-\gamma)(R_I^g-R_N^g)$. The resulting profit is $\Pi_E^g=R_I^g+(2\gamma-1)R_U^g$, so $\Pi_E^g-\Pi_N^g=R_I^g-R_U^g-2\gamma T^g$. Exclusive access for seller 2 yields weakly less, since the profit difference is $(1-\gamma)(R_N^g-R_U^g)=(1-\gamma)T^g\ge0$. Finally, $\Pi_B^g-\Pi_E^g=2R_B^g-R_I^g-R_U^g$. Comparing $\Pi_N^g,\Pi_B^g,\Pi_E^g$ proves part {\rm(A)}.

(2) MDC.
Under unilateral MDC,
$R_I^{MDC}+R_U^{MDC}=0$ if $v-c<\tfrac34$, $(v-c-\tfrac34)^2+\tfrac18$ if $\tfrac34\le v-c<\tfrac32$, and $\tfrac{11}{16}$ if $v-c\ge\tfrac32$. If $v\le1$, positive unilateral revenue requires $v>\tfrac34$, and $R_I^{MDC}+R_U^{MDC}<\tfrac3{16}<\tfrac{v^2}{2}=2R_N^{MDC}$. If $1<v<\tfrac32$, then $v-c<v$ and $2R_N^{MDC}-R_I^{MDC}-R_U^{MDC}>\tfrac38-(v-\tfrac54)^2>0$. If $v\ge\tfrac32$, $R_I^{MDC}+R_U^{MDC}\le\tfrac{11}{16}<1=2R_N^{MDC}$. Hence, $R_I^{MDC}+R_U^{MDC}<2R_N^{MDC}$. Together with $G^{MDC}<0$, neither bilateral nor exclusive access increases total seller sales revenue relative to NDC. Direct substitution gives $2R_B^{MDC}\ge R_I^{MDC}+R_U^{MDC}$ if and only if $v-c\le(3+\sqrt6)/4$. Thus, exclusive access can be optimal only if $v-c>(3+\sqrt6)/4>1$. Under exclusive access, $F_1^{MDC}=(1-\gamma)(R_I^{MDC}-R_U^{MDC})$, so both sellers earn $\pi_i^{MDC}=(1-\gamma)R_U^{MDC}<(1-\gamma)R_N^{MDC}=\pi_i^{NDC}$ because $T^{MDC}>0$. When $v-c>1$, NDC covers the entire market and assigns each consumer to the closer seller. Exclusive MDC introduces privacy costs and allocates a positive measure of consumers to the more distant seller, while possibly also reducing coverage. Hence, $SW^{MDC}<SW^{NDC}$.
Under exclusive MDC, the entry set is $[\max\{0,\tfrac32-v+c\},1]$, and each entrant obtains utility $v-c-\tfrac32+x$. Therefore, $CS^{MDC}=\tfrac12(v-c-\tfrac12)^2$ if $v-c<\tfrac32$, and $CS^{MDC}=v-c-1$ otherwise. For $1<v<\tfrac32$, the exclusive-access condition $v-c>(3+\sqrt6)/4$ implies $CS^{MDC}>\tfrac14=CS^{NDC}$. For $v\ge\tfrac32$ and $v-c<\tfrac32$, $CS^{MDC}>CS^{NDC}$ is equivalent to $c<v-\tfrac12-\sqrt{2v-\tfrac52}$; for $v-c\ge\tfrac32$, it is equivalent to $c<\tfrac14$. Combining these cases gives the regions in part (B).

(3) VDD.
Under exclusive VDD, no consumer discloses, so $R_I^{\VDD}=R_U^{\VDD}=R_N^{\VDD}$ and $T^{\VDD}=0$. Exclusive access therefore yields $\Pi_E^{\VDD}=2\gamma R_N^{\VDD}$, identical to NDC, while generating no usable information. Under the maintained provision convention, the platform therefore does not implement exclusive access.
\qed

\vspace{0.3cm}
\noindent\textbf{Proof of Proposition \ref{prop:heterogeneous_privacy}.}
When neither seller adopts, consumers incur no privacy cost, so $R_N^{g-h}=r^{\NDC}$ for $g\in\{\MDC,\VDD\}$. The contracting argument in Lemma~\ref{lem:fees} is unchanged and gives $F_1^{g-h}=F_{2\mid AC}^{g-h}=(1-\gamma)(R_B^{g-h}-R_U^{g-h})$. Hence, conditional on bilateral provision, $\pi_i^{g-h}=(1-\gamma)R_U^{g-h}$ and $\Pi_p^{g-h}-\Pi_p^{\NDC}=V_{\mathrm{data}}^{g-h}=2G^{g-h}+2(1-\gamma)T^{g-h}$. Write $[z]_+\triangleq\max\{z,0\}$ and $[z]_a^b\triangleq \min\{\max\{z,a\},b\}$

{(1) MDC.}
The bilateral MDC argument in Proposition~\ref{prop:equilibrium_mdc} and Lemma~\ref{app:lem_mdc_PP} applies separately to each privacy type after replacing $c$ with $c_k$. Privacy costs are sunk at the pricing stage, so personalized prices do not depend on the distribution of privacy types. Because each type is uniformly distributed over locations, aggregate bilateral outcomes are the weighted averages of the corresponding homogeneous-cost outcomes. This aggregation holds conditional on bilateral access, not after the platform's endogenous provision decision.

In particular, $R_B^{\MDC-h}=\rho\left.r^{PPM}\right|_{c=c_H}+(1-\rho)\left.r^{PPM}\right|_{c=c_L}$. Proposition~\ref{prop:equilibrium_mdc} establishes $r^{PPM}<r^{\NDC}$ for every $c>0$, so $R_B^{\MDC-h}<r^{\NDC}$ and $G^{\MDC-h}<0$. The same weighting applies to coverage and welfare. For $v\le1$, the baseline bilateral entry intervals give coverage $2[v-c_k-\tfrac12]_+<v$ and welfare $3[v-c_k-\tfrac12]_+^2<3v^2/4$ per unit mass of type-$k$ consumers. For $v>1$, NDC serves every consumer through the closer seller without privacy costs. Bilateral MDC either preserves a transaction while adding a positive privacy cost or eliminates a transaction with positive social surplus. Thus, each type has weakly lower coverage and strictly lower welfare than NDC. Weighting these comparisons yields $D^{\MDC-h}\le D^{\NDC}$ and $SW^{\MDC-h}<SW^{\NDC}$ whenever the service is provided.

The unilateral revenue comparison requires a separate argument because the uninformed seller charges a common uniform price. Suppose only seller 1 has access, and let $p=p_{2,\varnothing}^{PNM-h}$. A type-$k$ consumer purchasing from seller 2 must satisfy $v-p-(1-x)-c_k\ge0$ by voluntary entry and rational expectations. Since $c_k\ge c_L$, seller 2's demand is at most $[v-c_L-p]_+$. Thus, for $v\le1$, $R_U^{\MDC-h}\le\max_{p\ge0}p[v-c_L-p]_+=[v-c_L]_+^2/4<v^2/4=r^{\NDC}$. For $v>1$, a consumer purchasing from seller 2 must also satisfy $x\ge(1+p)/2$; otherwise, seller 1 can profitably attract her with a nonnegative personalized price. Hence, $R_U^{\MDC-h}\le\max_{p\ge0}p[(1-p)/2]_+=1/8<r^{\NDC}$. Therefore, $T^{\MDC-h}>0$, and the data-value identity gives the stated MDC provision condition.

For consumer surplus, let $CS_k^{\MDC-h}$ denote surplus per unit mass of type-$k$ consumers under bilateral provision. Applying the baseline entry utilities type by type gives $CS_k^{\MDC-h}=\int_0^1[v-\max\{x,1-x\}-c_k]_+\,dx$. Consequently, $CS_H^{\MDC-h}\le CS_L^{\MDC-h}$ and $CS^{\MDC-h}=\rho CS_H^{\MDC-h}+(1-\rho)CS_L^{\MDC-h}$. Define $\rho'=1$ if $CS_H^{\MDC-h}\ge CS^{\NDC}$, $\rho'=0$ if $CS_L^{\MDC-h}<CS^{\NDC}$, and $\rho'=(CS_L^{\MDC-h}-CS^{\NDC})/(CS_L^{\MDC-h}-CS_H^{\MDC-h})$ otherwise. The last case satisfies $CS_H^{\MDC-h}<CS^{\NDC}\le CS_L^{\MDC-h}$, so its denominator is strictly positive. Rearranging the weighted-average inequality proves the stated condition, including cases with equal type-specific surplus.

{(2) VDD.}
The nondisclosure argument in Lemma~\ref{app:lem_vdd_pn} applies separately to each privacy type because it holds for any anticipated uniform prices and requires only $c>0$. Replacing $c$ with $c_k$, anonymity improves each consumer's utility over disclosure by at least $c_k>0$. Hence, both types opt out under unilateral access, and all consumers enter because anonymous entry allows nonpurchase at zero utility. The continuation therefore coincides with NDC, giving $R_I^{\VDD-h}=R_U^{\VDD-h}=R_N^{\VDD-h}=r^{\NDC}$ and $T^{\VDD-h}=0$. By Proposition~\ref{prop:data_value_decomposition}, $V_{\mathrm{data}}^{\VDD-h}=2G^{\VDD-h}$. Thus, the platform provides the service if and only if $G^{\VDD-h}\ge0$ and the opt-in set is nonempty. It remains to establish the uniform-price comparison. Suppose a symmetric bilateral equilibrium satisfies $p_{i,\varnothing}^{PPV-h}<p_{i,\varnothing}^{\NDC}$. For a type-$k$ consumer at $x\le\tfrac12$, disclosure yields $[v-1+x]_+-c_k$, whereas anonymity yields $[v-p_{i,\varnothing}^{PPV-h}-x]_+$. The disclosure gain is nondecreasing in $x$, so anonymous consumers form intervals adjacent to the endpoints. Since $p_{i,\varnothing}^{PPV-h}<p_{i,\varnothing}^{\NDC}\le1$ and $c_k>0$, these intervals have positive length. A small upward deviation in the uniform price leaves disclosure and personalized-price revenue unchanged. For a type whose anonymous interval ends strictly before $\min\{v-p_{i,\varnothing}^{PPV-h},\tfrac12\}$, a sufficiently small increase preserves demand and raises revenue. Otherwise, the local demand response coincides with NDC. The corresponding revenue derivative is $v-2p_{i,\varnothing}^{PPV-h}$ if $v-p_{i,\varnothing}^{PPV-h}<\tfrac12$, $(1-p_{i,\varnothing}^{PPV-h})/2$ if $v-p_{i,\varnothing}^{PPV-h}>\tfrac12$, and $\tfrac12-p_{i,\varnothing}^{PPV-h}$ at equality, using the right derivative. By Proposition~\ref{prop:without_data}, each derivative is strictly positive when $p_{i,\varnothing}^{PPV-h}<p_{i,\varnothing}^{\NDC}$. Thus, a small price increase raises revenue for both types, contradicting equilibrium. Therefore, $p_{i,\varnothing}^{PPV-h}\ge p_{i,\varnothing}^{\NDC}$. This proves part~(B).

The personalized-pricing argument underlying Proposition~\ref{prop:price_coverage_cmp} applies separately to each privacy type. Under bilateral access, disclosure offers the same personalized prices under MDC and VDD, yielding utility $[v-\max\{x,1-x\}]_+-c_k$. VDD preserves this purchasing opportunity while allowing consumers to remain anonymous. Thus, every consumer obtains weakly higher utility under VDD, and every consumer purchasing under MDC continues purchasing under VDD, apart from measure-zero indifference boundaries. Aggregating across types gives $CS^{\VDD-h}\ge CS^{\MDC-h}$ and $D^{\VDD-h}\ge D^{\MDC-h}$.

For social welfare, consumers purchase from the closer seller under either bilateral regime. A consumer purchasing under MDC who continues disclosing under VDD retains the same allocation and privacy cost, whereas one who switches to anonymity avoids the privacy cost without changing the allocation. Any additional VDD buyer purchases anonymously and generates nonnegative social surplus, since voluntary purchase at a nonnegative price implies $v-\min\{x,1-x\}\ge0$. Therefore, $SW^{\VDD-h}\ge SW^{\MDC-h}$. Finally, the contracting argument in Proposition~\ref{prop:profit_campare_vdd} applies unchanged to the heterogeneous revenue states. By Lemma~\ref{lem:fees} and parts {\rm(A)}--{\rm(B)}, $\pi_i^{\VDD-h}=(1-\gamma)r^{\NDC}=\pi_i^{\NDC}$ and $\pi_i^{\MDC-h}=(1-\gamma)(r^{\NDC}-T^{\MDC-h})\le\pi_i^{\NDC}$. When both regimes provide the service, $\Pi_p^{\VDD-h}-\Pi_p^{\MDC-h}=2[G^{\VDD-h}-G^{\MDC-h}-(1-\gamma)T^{\MDC-h}]$. Since $T^{\MDC-h}>0$ and $\gamma\ge0$, the difference is nonnegative if and only if $\gamma\ge\max\{0,1-(G^{\VDD-h}-G^{\MDC-h})/T^{\MDC-h}\}$. \qed

\vspace{0.3cm}
\noindent\textbf{Proof of Proposition \ref{prop:imperfect_quality}.}
{(1) MDC. }Under MDC, an exclusive adopter can identify consumers with positive probability and exclusive access may remain informative. By Lemma \ref{lem:fees}, $F_1^{\MDC}(q)=F_{2\mid AC}^{\MDC}(q)=(1-\gamma)[G^{\MDC}(q)+T^{\MDC}(q)]$. Hence, the platform's profit gain from bilateral provision is $2[G^{\MDC}(q)+(1-\gamma)T^{\MDC}(q)]$, so it provides the service if and only if $G^{\MDC}(q)+(1-\gamma)T^{\MDC}(q)\ge0$.

{(2) VDD. }Under VDD, we first characterize the bilateral VDD-IQ pricing equilibrium, which characterizes the bilateral revenue
$G^{\VDD}(q)$ and also supports the counterfactual results in Figures~\ref{fig:data_provision_imp}--\ref{fig:data_quality}. 

\begin{lemma}[Bilateral VDD equilibrium under imperfect identification]
\label{app:lem_iq_pp_vdd}
Under VDD-IQ with both sellers purchase data service, all consumers enter, and the following results hold: 

{\rm(A)} For a candidate common uniform price $p$, let $\hat x=\max\{0,1-v+c,(1+c-p)/2\}$ and $d=\min\{\hat x,1/2\}$. The opt-in set is $[\hat x,1-\hat x]$ if $\hat x\le1/2$, and is empty otherwise. Identified consumers face $p_1^{PPV}(x)=[1-2x]_+$ and $p_2^{PPV}(x)=[2x-1]_+$.

{\rm(B)} Define $y(z,p)=[\min\{v-z,(1+p-z)/2\}]_0^1$, $A(y;d)=\min\{y,d\}+[y-(1-d)]_+$, and $H(z;p,d,q)=z[(1-q)y(z,p)+qA(y(z,p);d)]$. Then $p$ is a symmetric pure-strategy rational-expectations equilibrium price if and only if $p\in\arg\max_{z\ge0}H(z;p,d,q)$, holding $d$ fixed. Global optimality can be checked at the breakpoints and interior stationary points of the quadratic branches.

{\rm(C)} If $\tfrac12<v<1$ and $0<c<v-\tfrac12$, a symmetric equilibrium exists with a unique common uniform price and disclosure set. Let $a=1-v+c$ and $b_q=v(1-q)/(2-q)$. If $a<b_q$, then $d=(1-q)(2+2c-v)/(4-3q)$ and $p=v/2+q(2+2c-v)/[2(4-3q)]$. If $b_q\le a<v/2$, then $(d,p)=(a,v-a)$. If $a\ge v/2$, then $(d,p)=(a,v/2)$.
\end{lemma}

When multiple symmetric equilibria satisfy part {\rm(B)}, we apply the maintained selection rule, maximizing each seller's sales revenue and breaking ties in favor of greater disclosure. Next, we suppose only seller 1 has data access. If a consumer opts in, she is identified with probability $q$, in which case her utility is $U^{in}(x)=\max\{0,v-p_{2,\varnothing}-(1-x)\}-c$; with probability $1-q$, identification fails and she obtains $U^{out}(x)$, where $U^{out}(x)=\max\{0,v-p_{1,\varnothing}-x,v-p_{2,\varnothing}-(1-x)\}$. At the anticipated equilibrium prices, opting in yields expected utility $q[\max\{0,v-p_{2,\varnothing}-(1-x)\}-c]+(1-q)U^{out}(x)\le U^{out}(x)-qc<U^{out}(x)$ for every $q>0$. Hence, no consumer discloses under unilateral access. Therefore, $R_I^{\VDD}(q)=R_U^{\VDD}(q)=R_N^{\VDD}(q)=r^{\NDC}$ and $T^{\VDD}(q)=0$ for $q\in(0,1]$. At $q=0$, no usable information is generated, so $T^{\VDD}(0)=0$ as well. By Lemma \ref{lem:fees}, $F_1^{\VDD}(q)=F_{2\mid AC}^{\VDD}(q)=(1-\gamma)G^{\VDD}(q)$ and $F_{2\mid RE}^{\VDD}(q)=0$. Hence, $V_{data}^{\VDD}(q)=2G^{\VDD}(q)$. Under the maintained convention that the service is not activated when the equilibrium opt-in set is empty, the platform provides the service if and only if $G^{\VDD}(q)\ge0$ and the opt-in set is nonempty. 
\qed

\section{Proofs of Results in Appendix}\label{app:sec_proof_app}
\noindent{\bf Proof of Lemma \ref{app:lem_mdc_PP}.}
We solve the subgame by backward induction. Conditional on entry, the privacy cost $c$ is sunk, so not purchasing yields utility $-c$. Consider $x\le\tfrac12$. Seller 1 has mismatch-cost advantage $1-2x$. If seller 2 can generate positive gross surplus at zero price, Bertrand competition yields $p_2(x)=0$ and $p_1(x)=1-2x$, leaving gross surplus $v-(1-x)$; otherwise, seller 1 extracts the full gross surplus $v-x$. Thus, continuation utility from entry is $\max{v-(1-x),0}-c$. Since $c>0$, entry occurs if and only if $x\ge1-(v-c)$. Symmetrically, for $x\ge\tfrac12$, entry requires $x\le v-c$. Hence, the entry set is empty if $v-c<\tfrac12$, $[1-(v-c),v-c]$ if $\tfrac12\le v-c<1$, and $[0,1]$ if $v-c\ge1$, equivalent to the regions in the lemma. For every entering consumer with $x\le\tfrac12$, $v-(1-x)\ge c>0$, so both sellers can profitably serve the consumer and Bertrand competition gives $p_1^{PPM}(x)=1-2x$ and $p_2^{PPM}(x)=0$. Symmetrically, for $x\ge\tfrac12$, $p_1^{PPM}(x)=0$ and $p_2^{PPM}(x)=2x-1$. Therefore, for all $x\in\Omega^{PPM}$, $p_1^{PPM}(x)=\max\{1-2x,0\}$ and $p_2^{PPM}(x)=\max\{2x-1,0\}$.  \qed

\vspace{0.3cm}
{\bf Proof of Lemma \ref{app:lem_mdc_PN}.}
Fix an anticipated entry set $[\hat{x}^{PNM},1]$. Under the maintained tie-breaking rule, the informed seller can attain the highest personalized price that leaves the consumer weakly preferring its offer to the uninformed alternative. Given seller 2's uniform price $p_{2,\varnothing}$, seller 1's optimal personalized price is
$\hat p_1(x,p_{2,\varnothing})=\max\{0,\min\{v-x,1+p_{2,\varnothing}-2x\}\}$.
Hence, a consumer enters if and only if $x\ge1-(v-c-p_{2,\varnothing})$. For an interior REE, rational expectations therefore require
$p_{2,\varnothing}=\hat{x}^{PNM}+v-c-1$.

Conditional on $\hat{x}^{PNM}$, seller 2's demand at price $p$ is
$D_2(p,\hat{x}^{PNM})=[1-\max{\hat{x}^{PNM},(1+p)/2,1-v+p}]_+$.
At an interior fixed point, $1-v+p=\hat{x}^{PNM}-c<\hat{x}^{PNM}$, so the participation constraint is locally slack. If $\hat{x}^{PNM}<v-c$, seller 2's local revenue is $p(1-p)/2$, which requires $p=\tfrac12$ and hence $\hat{x}^{PNM}=\tfrac32-(v-c)$. This fixed point exists if and only if $v-c\ge\tfrac34$. Moreover, when $\hat{x}^{PNM}\le\tfrac34$, $pD_2(p,\hat{x}^{PNM})\le p(1-p)/2\le\tfrac18$, and $p=\tfrac12$ attains the bound, so it is globally optimal. If $\hat{x}^{PNM}>v-c$, seller 2's revenue is locally increasing in $p$, so no fixed point exists. If $\hat{x}^{PNM}=v-c$, the fixed-point price is $p=2(v-c)-1$, which is optimal only when $v-c\ge\tfrac34$. Thus, for $\tfrac34<v-c\le1$, there is also an equilibrium with entry set $[v-c,1]$; the largest-entry selection chooses $\hat{x}^{PNM}=\tfrac32-(v-c)$. If $v-c\ge\tfrac32$, this threshold is nonpositive and all consumers enter. Therefore, $\Omega^{PNM}=[\max{0,\tfrac32-(v-c)},1]$ whenever $v-c\ge\tfrac34$, and is empty otherwise.

Finally, substituting $p_{2,\varnothing}^{PNM}=\tfrac12$ into seller 1's best response gives $p_1^{PNM}(x)=\max\{0,\min\{v-x,\tfrac32-2x\}\}$. For every entrant with $x\le\tfrac34$, $x\ge\max\{0,\tfrac32-(v-c)\}>\tfrac32-v$, so $\tfrac32-2x<v-x$; for $x\ge\tfrac34$, the optimal price is zero. Hence, $p_1^{PNM}(x)=\max\{\tfrac32-2x,0\}$. \qed

\vspace{0.3cm}
\noindent{\bf Proof of Remark \ref{app:rem_mdc_PN}.}
Consider the candidate entry threshold $\hat{x}^{PNM}=v-c$. Since $v-c>\tfrac34$, seller 2's best response is $\widetilde p_{2,\varnothing}^{PNM}=2(v-c)-1$. Substituting this price into the entry threshold $1-(v-c-p_{2,\varnothing})$ returns $v-c$, establishing rational-expectations consistency. Seller 2 then serves consumers in $(v-c,1]$ and earns $\widetilde r_2^{PNM}=[2(v-c)-1][1-(v-c)]$. Under the maximal-participation equilibrium, $p_{2,\varnothing}^{PNM}=\tfrac12$ and $r_2^{PNM}=\tfrac18$. Hence,
$r_2^{PNM}-\widetilde r_2^{PNM}=2(v-c-\tfrac34)^2>0$. Moreover, $c>0$ and $v-c>\tfrac34$ imply $v>\tfrac34$. Thus, $r^{\NDC}>\tfrac18\ge\widetilde r_2^{PNM}$: if $v\le1$, $r^{\NDC}=v^2/4>9/64>1/8$; if $1<v<\tfrac32$, $r^{\NDC}=(2v-1)/4>1/4$; and if $v\ge\tfrac32$, $r^{\NDC}=1/2$. Therefore, $\widetilde T^{\MDC}=r^{\NDC}-\widetilde r_2^{PNM}>0$ and
$\widetilde T^{\MDC}-T^{\MDC}=r_2^{PNM}-\widetilde r_2^{PNM}=2(v-c-\tfrac34)^2>0$.
Since provision requires $G^{\MDC}+(1-\gamma)T^{\MDC}\ge0$, the maximal-participation equilibrium yields the more demanding monetization condition. \qed

\vspace{0.3cm}
\noindent{\bf Proof of Lemma \ref{app:lem_mdc_NP}.}
This subgame is symmetric to the ``PN'' subgame in Lemma \ref{app:lem_mdc_PN} after interchanging the two sellers and replacing $x$ by $1-x$. Hence, $\Omega^{NPM}=\emptyset$ if $v-c<\tfrac34$ and $\Omega^{NPM}=[0,\min\{v-c-\tfrac12,1\}]$ otherwise, with $p_2^{NPM}(x)=\max\{2x-\tfrac12,0\}$ and $p_{1,\varnothing}^{NPM}=\tfrac12$. \qed

\vspace{0.3cm}
\noindent{\bf Proof of Lemma \ref{app:lem_vdd_pp}.}
We restrict attention to symmetric pure-strategy rational-expectations equilibria and select the equilibrium maximizing each seller's sales revenue, breaking ties in favor of more disclosure. Since entry yields at least zero utility, all consumers enter.

Let $p_{i,\varnothing}^e$ be seller $i$'s anticipated uniform price. Consider an opt-in consumer with $x\le\tfrac12$. Once she discloses, $c$ is sunk. For $x\le \tfrac{1}{2}$, if $v\ge1-x$, seller 2 can serve her at price zero, so Bertrand competition gives $p_1(x)=1-2x$ and $p_2(x)=0$, yielding utility $v-(1-x)-c$. If $v<1-x$, seller 2 cannot profitably serve her, so seller 1 charges $v-x$ whenever a purchase occurs, yielding utility $-c$. Thus, $U^{opt-in,e}(x)=[v-(1-x)]_+-c$. Symmetrically, $U^{opt-in,e}(x)=[v-x]_+-c$ for $x\ge\tfrac12$, so $U^{opt-in,e}(x)=[v-\max\{x,1-x\}]_+-c$. Opting out yields $U^{opt-out,e}(x)=\max\{0,v-p_{1,\varnothing}^e-x,v-p_{2,\varnothing}^e-(1-x)\}$. For $x\le\tfrac12$, comparing these utilities shows that disclosure is optimal if and only if $x\ge1-v+c$, $x\ge(1+c-p_{1,\varnothing}^e)/2$, and $p_{2,\varnothing}^e\ge c$. The conditions for $x\ge\tfrac12$ are symmetric. Hence, when $p_{1,\varnothing}^e,p_{2,\varnothing}^e\ge c$, define $\hat x_i^{PPV,e}=\max\{0,1-v+c,(1+c-p_{i,\varnothing}^e)/2\}$. The anticipated disclosure set is $[\hat x_1^{PPV,e},1-\hat x_2^{PPV,e}]$ if $\hat x_1^{PPV,e}+\hat x_2^{PPV,e}\le1$, and is empty otherwise.

Fix an interior disclosure set $[\hat x_1^{PPV},1-\hat x_2^{PPV}]$. Since realized uniform prices are unobserved when disclosure is chosen, a uniform-price deviation leaves this set and personalized-price revenue fixed. Let $q_i\le\hat x_i^{PPV}$ be seller $i$'s anonymous demand adjacent to its endpoint. Its price cannot exceed $v-q_i$ or $1+p_{-i,\varnothing}-2q_i$, so anonymous revenue is at most $q_i\min\{v-q_i,1+p_{-i,\varnothing}-2q_i\}$. The competition constraint cannot bind in equilibrium. If it bound for both sellers, their two equalities would imply $q_1+q_2=1$; if it bound only for seller $i$ while the rival's valuation constraint bound, $p_{i,\varnothing}=1+v-q_{-i}-2q_i\le v-q_i$ would again imply $q_i+q_{-i}\ge1$. Both contradict a positive disclosure interval. Thus, seller $i$ maximizes $q_i(v-q_i)$ over $q_i\le\hat x_i^{PPV}$, giving $q_i=\min\{\hat x_i^{PPV},v/2\}$ and $p_{i,\varnothing}^{PPV}=\max\{v/2,v-\hat x_i^{PPV}\}$.

Rational expectations require $p_{i,\varnothing}^e=p_{i,\varnothing}^{PPV}$. Substituting the preceding price into $\hat x_i^{PPV}=\max\{0,1-v+c,(1+c-p_{i,\varnothing}^{PPV})/2\}$ yields $\hat x_i^{PPV}=1-v+c$. To see this, if $\hat x_i^{PPV}\ge v/2$, then $p_{i,\varnothing}^{PPV}=v/2$; if the price-dependent boundary exceeded $1-v+c$, it would require $c<3v/2-1$, whereas $\hat x_i^{PPV}\ge v/2$ would require $c\ge3v/2-1$, a contradiction. If $\hat x_i^{PPV}<v/2$, then $p_{i,\varnothing}^{PPV}=v-\hat x_i^{PPV}$ and consistency becomes $\hat x_i^{PPV}=\max\{1-v+c,(1-v+c+\hat x_i^{PPV})/2\}$, whose only solution is $\hat x_i^{PPV}=1-v+c$. Therefore, $\hat x_1^{PPV}=\hat x_2^{PPV}=1-v+c$ and $p_{i,\varnothing}^{PPV}=\max\{v/2,2v-c-1\}$.

For an opt-in consumer with $x\le\tfrac12$, $x\ge1-v+c$ implies $v-x-(1-2x)=v-1+x\ge c>0$, so the competition constraint binds and $p_1^{PPV}(x)=1-2x$; symmetrically, $p_2^{PPV}(x)=2x-1$ for $x\ge\tfrac12$. Hence, $p_1^{PPV}(x)=\max\{1-2x,0\}$ and $p_2^{PPV}(x)=\max\{2x-1,0\}$.

Let $\hat x^{PPV}=1-v+c$. It remains to rule out price reductions attracting anonymous consumers from the opposite side. If $\hat x^{PPV}\ge v/2$, any deviation yields anonymous revenue at most $\max_{p\ge0}p(v-p)=v^2/4$, which is attained in equilibrium. If $0<\hat x^{PPV}<\min\{v/2,1/2\}$, then $p_{\varnothing}^{PPV}=v-\hat x^{PPV}$ and revenue from consumers on the seller's own side cannot increase because $q(v-q)$ is increasing for $q\le\hat x^{PPV}$. A deviation crossing the disclosure interval and serving $q\in[\hat x^{PPV},2\hat x^{PPV}]$ yields at most $q[p_{\varnothing}^{PPV}+4\hat x^{PPV}-1-2q]$. Maximizing this quadratic and comparing it with $p_{\varnothing}^{PPV}\hat x^{PPV}$ gives $p_{\varnothing}^{PPV}\le2$ if $\hat x^{PPV}\le\tfrac14$, and $p_{\varnothing}^{PPV}\le1+\sqrt{8\hat x^{PPV}(1-2\hat x^{PPV})}$ if $\hat x^{PPV}>\tfrac14$. These conditions always hold for $v\le\tfrac32$: the first follows from $p_{\varnothing}^{PPV}<2$, while for $\hat x^{PPV}>\tfrac14$ and $p_{\varnothing}^{PPV}>1$, $p_{\varnothing}^{PPV}-1\le\tfrac12-\hat x^{PPV}$ and $8\hat x^{PPV}(1-2\hat x^{PPV})-(\tfrac12-\hat x^{PPV})^2=(\tfrac12-\hat x^{PPV})(17\hat x^{PPV}-\tfrac12)\ge0$. They also hold for $v>\tfrac32$ and $c<1$, because $p_{\varnothing}^{PPV}=1+c-2\hat x^{PPV}<2$ and, when $\hat x^{PPV}>\tfrac14$ and $p_{\varnothing}^{PPV}>1$, $p_{\varnothing}^{PPV}-1=c-2\hat x^{PPV}<1-2\hat x^{PPV}\le\sqrt{8\hat x^{PPV}(1-2\hat x^{PPV})}$.

If $\hat x_1^{PPV}=\hat x_2^{PPV}=0$, the disclosure conditions require $c\le v-1$ and $p_{i,\varnothing}^{PPV}\ge1+c$. Conversely, these conditions induce full disclosure. Since a realized uniform-price deviation cannot create anonymous demand once all consumers have disclosed, any symmetric price weakly above $1+c$ supports the outcome; we take $1+c$ as a representative price.

If the disclosure set is empty or has zero measure, the uniform-price game coincides with NDC, with $p_{\varnothing}^{NDC}=v/2$ for $v\le1$, $p_{\varnothing}^{NDC}=v-\tfrac12$ for $1<v\le\tfrac32$, and $p_{\varnothing}^{NDC}=1$ for $v>\tfrac32$. At these prices, the maximal gain from disclosure is attained at $x=\tfrac12$ and equals $\min\{v-\tfrac12,1\}-c$. Hence, no consumer discloses if $v\le\tfrac12$; for $v>\tfrac12$, the disclosure set is $\{1/2\}$ when $c=\min\{v-\tfrac12,1\}$ and is empty when $c>\min\{v-\tfrac12,1\}$. At $x=\tfrac12$, $p_1^{PPV}(1/2)=p_2^{PPV}(1/2)=0$.

For $v\le\tfrac32$, these conditions yield full disclosure when $v>1$ and $c\le v-1$, interval disclosure $[\max\{0,1-v+c\},\min\{1,v-c\}]$ when $v>\tfrac12$ and $\max\{v-1,0\}<c<v-\tfrac12$, singleton disclosure at $c=v-\tfrac12$, and no disclosure otherwise. For $v>\tfrac32$ and $c<1$, full disclosure occurs when $c\le v-1$, while the same interval disclosure occurs when $v-1<c<1$.

Finally, let $v>\tfrac32$ and $c\ge1$. The NDC-like equilibrium gives each seller revenue $1/2$, while any full-disclosure equilibrium gives $1/4$ and is not selected. An interior-disclosure equilibrium has $\hat x^{PPV}=1-v+c\in(0,1/2)$ and revenue $(\tfrac12-\hat x^{PPV})^2+(v-\hat x^{PPV})\hat x^{PPV}=\tfrac14+(v-1)\hat x^{PPV}$, so it is selected if and only if $\hat x^{PPV}\ge1/[4(v-1)]$, with equality selected because it induces more disclosure. The no-deviation condition additionally requires $\hat x^{PPV}\ge v-2$ when $\hat x^{PPV}\le\tfrac14$, and $17(\hat x^{PPV})^2-2(v+3)\hat x^{PPV}+(v-1)^2\le0$ when $\hat x^{PPV}>\tfrac14$. The latter has roots $(v+3\pm\sqrt{-16v^2+40v-8})/17$. Since $1/[4(v-1)]=v-2$ at $v=(3+\sqrt2)/2$, the smaller root equals $1/4$ at $v=9/4$, and the roots are real up to $v=(5+\sqrt{17})/4$, substituting $\hat x^{PPV}=1-v+c$ gives exactly $\mathcal P_3$. Thus, on $\mathcal P_3$, $\Omega^{PPV}=[1-v+c,v-c]$ and $p_{i,\varnothing}^{PPV}=2v-c-1$; otherwise the selected outcome is that in part {\rm(A)}. 
\qed

\vspace{0.3cm}
\noindent{\bf Proof of Remark \ref{app:rem_vdd_PP}.}
All three regions satisfy $v>\tfrac32$, so $p_{\varnothing}^{NDC}=1$ and $r^{NDC}=1/2$.

For {\rm(i)}, $v\ge2$, $1\le c\le v-1$, and $p_{i,\varnothing}^{PPV}\ge1+c$ imply $1-v+c\le0$ and $(1+c-p_{i,\varnothing}^{PPV})/2\le0$, so the disclosure boundaries are zero and $\Omega^{PPV}=[0,1]$. A realized uniform-price deviation cannot change disclosure and hence cannot generate anonymous demand. Thus, each seller earns $r^{PPV}=\int_0^{1/2}(1-2x)dx=1/4<r^{NDC}$, so $G^{\VDD}<0$ and the platform does not provide the service.

For {\rm(ii)}, let $(v,c)\in\mathcal P_3$, where $c>1$. At $p_{1,\varnothing}^{PPV}=p_{2,\varnothing}^{PPV}=1$, for $x\le\tfrac12$ the gain from disclosure is $2x-c\le1-c<0$; the case $x>\tfrac12$ follows symmetrically. Hence, $\Omega^{PPV}=\varnothing$. The uniform-price game then coincides with NDC, so $p_{i,\varnothing}^{PPV}=1$ and $r^{PPV}=1/2$. Thus, $G^{\VDD}=0$, but the platform does not provide because the opt-in set is empty.

For {\rm(iii)}, the stated restrictions imply $0<1-v+c<1/2$. Letting $\hat x^{PPV}=1-v+c$, the interior-equilibrium conditions established in Lemma \ref{app:lem_vdd_pp} give $p_{i,\varnothing}^{PPV}=v-\hat x^{PPV}=2v-c-1$. The no-deviation condition is $\hat x^{PPV}\ge v-2$ when $\hat x^{PPV}\le1/4$, which follows from $c\ge2v-3$; when $\hat x^{PPV}>1/4$, the stated parameter bounds place $\hat x^{PPV}$ between the roots of $17(\hat x^{PPV})^2-2(v+3)\hat x^{PPV}+(v-1)^2\le0$. Hence, $\Omega^{PPV}=[1-v+c,v-c]$ is a rational-expectations equilibrium. Each seller earns $r^{PPV}=1/4+(v-1)(1-v+c)<1/2$, where the inequality follows from $c<v-1+1/[4(v-1)]$. Thus, $G^{\VDD}<0$ and the platform again does not provide.

Relative to the selected equilibrium, cases {\rm(i)} and {\rm(iii)} change provision only when $c=1$: the selected equilibrium then has $\Omega^{PPV}=\{1/2\}$ and $G^{\VDD}=0$, so provision changes but all payoffs remain unchanged. For $c>1$, both selections induce non-provision.

On $\mathcal P_3$, the selected equilibrium has $r^{PPV}=1/4+(v-1)(1-v+c)\ge1/2$, so the platform provides and earns $\Pi_P^{VDD}=\gamma+2G^{\VDD}\ge\gamma$. Under the alternative empty-disclosure equilibrium it does not provide and earns $\gamma$, so platform profit and seller sales revenues are weakly lower. Seller profits are unchanged because the selected-equilibrium fee leaves each seller with $(1-\gamma)r^{PPV}-(1-\gamma)(r^{PPV}-r^{NDC})=(1-\gamma)r^{NDC}$. Moreover, the selected equilibrium gives $CS^{PPV}=(1-v+c)^2+(v-c-\tfrac12)^2<1/4$, whereas non-provision gives $CS^{NDC}=v-\tfrac54>1/4$; hence consumer surplus is strictly higher under the alternative equilibrium.

In the PP subgame, the equilibria in {\rm(i)} and {\rm(iii)} yield sales revenues $1/4$ and $1/4+(v-1)(1-v+c)<1/2$, respectively, while {\rm(ii)} yields $1/2\le r^{PPV}$ under the selected equilibrium. Since the preceding characterization exhausts the symmetric pure-strategy rational-expectations equilibria and data-service fees are sunk in this subgame, the equilibrium selected in Lemma \ref{app:lem_vdd_pp} weakly maximizes both sellers' continuation profits. Finally, these alternative PP equilibria do not affect the exclusive access subgames, so $T^{\VDD}=0$ holds. \qed

\vspace{0.3cm}
\noindent{\bf Proof of Lemma \ref{app:lem_vdd_pn}.}
By symmetry, we assume without loss of generality that seller 1 makes the purchase. Let $p_{i,\varnothing}^e$ denote the anticipated uniform price of seller $i$. If a consumer at $x$ opts in, seller 1 observes $x$ and can extract any surplus above the consumer's outside option from seller 2 or no purchase. Hence, under rational expectations, her utility from opting in is
$U^{in,e}(x)=\max\{v-p_{2,\varnothing}^e-(1-x),0\}-c$.
If she opts out, both sellers charge uniform prices, giving
$U^{out,e}(x)=\max\{v-p_{1,\varnothing}^e-x,v-p_{2,\varnothing}^e-(1-x),0\}$.
Therefore, $U^{out,e}(x)\ge U^{in,e}(x)+c>U^{in,e}(x)$ for every $x$, so no consumer discloses. Since entering, opting out, and not purchasing yields zero utility, all consumers enter.
\qed

\vspace{0.3cm}
\noindent\textbf{Proof of Proposition \ref{app:prop_cred_micro}.}
In the subsidy region, the platform's net data value is $\Delta\Pi_P^g=2G^g+(1+\eta^g)(1-\gamma)T^g-2a-K$. Thus, $\tfrac{\partial\Delta\Pi_P^g}{\partial\eta^g}=(1-\gamma)T^g$. Since $\eta^g=1-\Phi_\kappa(\Delta L^g)$ and $\Delta L^g=R_N^g-R_I^g+\gamma T^g+a+K$, $\tfrac{\partial\eta^g}{\partial a}=\tfrac{\partial\eta^g}{\partial K}=-\phi_\kappa(\Delta L^g)$ and $\tfrac{\partial\eta^g}{\partial\gamma}=-T^g\phi_\kappa(\Delta L^g)$. If $T^g>0$, higher $a$ or $K$ lowers credibility and data value, while $\tfrac{\partial\Delta\Pi_P^g}{\partial\gamma}=-(1+\eta^g)T^g-(1-\gamma)(T^g)^2\phi_\kappa(\Delta L^g)<0$. An upward shift in the distribution of $\kappa$ raises $\eta^g$ and thus weakly increases data value. If $T^g<0$, $\tfrac{\partial\Delta\Pi_P^g}{\partial a}=(1-\gamma)|T^g|\phi_\kappa(\Delta L^g)-2$, $\tfrac{\partial\Delta\Pi_P^g}{\partial K}=(1-\gamma)|T^g|\phi_\kappa(\Delta L^g)-1$, and $\tfrac{\partial\Delta\Pi_P^g}{\partial\gamma}=(1+\eta^g)|T^g|-(1-\gamma)|T^g|^2\phi_\kappa(\Delta L^g)$. Hence, these derivatives are nonnegative if and only if $\phi_\kappa(\Delta L^g)\ge\tfrac{2}{(1-\gamma)|T^g|}$, $\phi_\kappa(\Delta L^g)\ge\tfrac{1}{(1-\gamma)|T^g|}$, and $\phi_\kappa(\Delta L^g)\le\tfrac{1+\eta^g}{(1-\gamma)|T^g|}$, respectively. An upward shift in $\kappa$ raises $\eta^g$ and therefore weakly decreases data value when $T^g<0$.
\qed

\vspace{0.3cm}
\noindent\textbf{Proof of Proposition \ref{prop:threat_discipline}.}
By definition, $T^g=R_N^g-R_U^g(\widehat{\Omega}^g,q)$ and $R_N^g=\mathcal R_U^g(\varnothing,0)$. Since unilateral information weakly disadvantages the uninformed seller, $T^g\ge0$ holds. The $L$-Lipschitz continuity assumption further implies $T^g=\left|\mathcal R_U^g(\varnothing,0)-\mathcal R_U^g(\widehat{\Omega}^g,q)\right|\le L\int_\Theta q(\theta)\mathbf 1_{\widehat{\Omega}^g}(\theta)d\mu(\theta)=L\beta^g$. Hence, $0\le T^g\le L\beta^g$. Therefore, $\beta^g=0$ implies $T^g=0$, establishing sufficiency. For necessity, suppose $T^g=0$ and $R_U^g(\widehat{\Omega}^g,q)<R_N^g$ whenever $\beta^g>0$. Then $\beta^g>0$ would imply $T^g=R_N^g-R_U^g(\widehat{\Omega}^g,q)>0$, a contradiction. Hence, under the additional condition, threat elimination requires $\beta^g=0$.
\qed

\vspace{0.3cm}
\noindent{\bf Proof of Table \ref{tab:mandatory_equilibrium}.}
Conditional on data provision, $F_1^{\MDC}=F_{2\mid AC}^{\MDC}=(1-\gamma)(r^{PPM}-r_2^{PNM})$. Hence, each seller earns $\pi_i^{\MDC}=(1-\gamma)r^{PPM}-F_1^{\MDC}=(1-\gamma)r_2^{PNM}$, which equals $0$ if $r_2^{PNM}=0$ and $(1-\gamma)/8$ if $r_2^{PNM}=1/8$. The platform earns $\Pi_p^{\MDC}=2\gamma r^{PPM}+2F_1^{\MDC}=2r^{PPM}-2(1-\gamma)r_2^{PNM}$, which reduces to $2r^{PPM}$ and $2r^{PPM}-(1-\gamma)/4$, respectively. If entry is partial, $CS^{\MDC}=2\int_{1-(v-c)}^{1/2}(v-1+x-c)dx=(2v-2c-1)^2/4$; if all consumers enter, $CS^{\MDC}=2\int_0^{1/2}(v-1+x-c)dx=v-c-\tfrac34$. Since prices and platform payments are transfers, $SW^{\MDC}=CS^{\MDC}+2r^{PPM}$. Substituting the corresponding $r^{PPM}$ and $r_2^{PNM}$ gives the outcomes in all provision regions. Whenever the platform does not provide the data service, the equilibrium coincides with NDC. This yields Table \ref{tab:mandatory_equilibrium}. \qed

\vspace{0.3cm}
{\bf Proof of Table \ref{tab:voluntary_equilibrium}.}
Conditional on provision, $F^{\VDD}=F_1^{\VDD}=F_{2\mid AC}^{\VDD}=(1-\gamma)(r^{PPV}-r^{\NDC})$. Hence, $\pi_i^{\VDD}=(1-\gamma)r^{\NDC}$, $\Pi_p^{\VDD}=2r^{PPV}-2(1-\gamma)r^{\NDC}$, and $SW^{\VDD}=CS^{\VDD}+2r^{PPV}$.

When $\tfrac12<v\le\tfrac23$ and $c<v-\tfrac12$, or $\tfrac23<v\le1$ and $\tfrac{3v}{2}-1\le c<v-\tfrac12$, $r^{\NDC}=v^2/4$. Substitution gives $\pi_i^{\VDD}=(1-\gamma)v^2/4$ and $\Pi_p^{\VDD}=[\gamma v^2+(2v-2c-1)^2]/2$. Moreover,
$CS^{\VDD}=2\int_0^{v/2}(v-v/2-x)dx+2\int_{1-(v-c)}^{1/2}(v-1+x-c)dx=[v^2+(2v-2c-1)^2]/4$,
so $SW^{\VDD}=3[v^2+(2v-2c-1)^2]/4$.

When $\tfrac23<v\le1$ and $c<\tfrac{3v}{2}-1$, $p_{i,\varnothing}^{PPV}=2v-c-1$, $r^{\NDC}=v^2/4$, and the market is fully covered. Hence, $\pi_i^{\VDD}=(1-\gamma)v^2/4$ and $\Pi_p^{\VDD}=[-(5-\gamma)v^2+4vc+8v-4c-3]/2$. Consumer surplus is
$CS^{\VDD}=2\int_0^{1-(v-c)}[v-(2v-c-1)-x]dx+2\int_{1-(v-c)}^{1/2}(v-1+x-c)dx=[4(1-v+c)^2+(2v-2c-1)^2]/4$,
and therefore $SW^{\VDD}=(4v-1-8vc+8c^2+4c)/4$.

For $(v,c)\in\mathcal P_3$, $r^{\NDC}=1/2$ and $r^{PPV}=[-4v^2+4vc+8v-4c-3]/4$. Thus, $\pi_i^{\VDD}=(1-\gamma)/2$ and $\Pi_p^{\VDD}=[-4v^2+4vc+8v-4c-(5-2\gamma)]/2$. Since the uniform price, disclosure allocation, and full coverage coincide with the preceding case, $CS^{\VDD}$ and $SW^{\VDD}$ are given by the same expressions. 

Whenever the platform does not provide the data service, the equilibrium coincides with NDC. Substituting these outcomes gives Table \ref{tab:voluntary_equilibrium}. \qed

\vspace{0.3cm}
\noindent{\bf Proof of Lemma \ref{app:lem_iq_pp_vdd}.}
Let $p_{i,\varnothing}^e$ denote seller $i$'s anticipated uniform price. All consumers enter because they can remain anonymous and decline to purchase. For $x\le1/2$, successful identification gives gross utility $u^I(x)=[v-1+x]_+$, whereas anonymity gives $u^A(x)=[v-p^e-x]_+$ under symmetric anticipated uniform prices. Since opting in yields $q[u^I(x)-c]+(1-q)u^A(x)$ and opting out yields $u^A(x)$, disclosure is optimal if and only if $u^I(x)-c\ge u^A(x)$ for $q>0$. This requires $x\ge1-v+c$ and $x\ge(1+c-p^e)/2$. The right half-market is symmetric. Conditional on identification, the privacy cost is sunk and both sellers observe $x$, giving the personalized prices in part {\rm(A)}.

Fix the symmetric anticipated disclosure set, but allow actual uniform prices to differ. If seller $i$ charges $z$ against the rival's price $p_j$, its potential anonymous demand has mass $y(z,p_j)$. Opt-out consumers remain anonymous with probability one, whereas opt-in consumers remain anonymous with probability $1-q$. Thus, anonymous demand is $(1-q)y(z,p_j)+qA(y(z,p_j);d)$. Personalized revenue is unaffected by this uniform-price deviation because the anticipated disclosure set remains fixed. Each seller therefore maximizes $H(z;p_j,d,q)$ holding $p_j$ and $d$ fixed. Imposing $p_1=p_2=p$ and rational-expectations consistency gives the necessary and sufficient condition in part {\rm(B)}.

For completeness, this global maximization requires only finitely many evaluations. Prices $z\ge v$ generate zero anonymous revenue. On $[0,v]$, form the breakpoint set $\mathcal Z_0=(\{0,v,2v-1-p\}\cup\{v-t,1+p-2t:t\in\{0,d,1-d,1\}\})\cap[0,v]$. Between consecutive breakpoints, anonymous demand is affine, say $\alpha-\beta z$ with $\beta\ge0$, and revenue is $z(\alpha-\beta z)$. Its maximum occurs at an endpoint or at $z=\alpha/(2\beta)$ when $\beta>0$ and this point lies inside the interval. Comparing these values establishes global optimality without assuming global concavity.

Now suppose $\tfrac12<v<1$ and $0<c<v-\tfrac12$. Every symmetric equilibrium has $v/2\le p<v$. Indeed, a price below $v/2$ admits a profitable small increase. On the participation-constrained branch, the right derivative is either $v-2p$ or $(1-q)(v-2p)+qd$, both positive. On the competition-constrained branch, it is either $(1-p)/2$ or $(1-q)(1-p)/2+qd$, also positive. A price at least $v$ generates zero anonymous revenue and is dominated by a positive price below $v$, since $d\ge a>0$. Consequently, anonymous demand at symmetry is participation-constrained, and $p\ge v/2>c$ ensures $0<d<1/2$.

For fixed $d$, define $\bar H(z;d)=z\{(v-z)_+-q[(v-z)_+-d]_+\}$. Near a symmetric candidate, this equals the actual anonymous revenue. For $q<1$, its derivatives on the branches $z<v-d$ and $z>v-d$ are $(1-q)(v-2z)+qd$ and $v-2z$, respectively. The derivatives immediately below and above the kink are $(2-q)d-(1-q)v$ and $2d-v$. Hence, a symmetric candidate must satisfy $p=v/2+qd/[2(1-q)]$ when $d<b_q$, $p=v-d$ when $b_q\le d<v/2$, and $p=v/2$ when $d\ge v/2$. At $q=1$, the first branch is absent. Combining these conditions with $d=\max\{a,(1+c-p)/2\}$ gives exactly the three cases in part {\rm(C)}.

It remains to verify global deviations. If $d\ge v/2$, the candidate $p=v/2$ earns anonymous revenue $v^2/4$, while $H(z;p,d,q)\le z(v-z)_+\le v^2/4$ for every $z\ge0$. Suppose instead that $d<v/2$. The candidate satisfies $p+d\le v<1$ and globally maximizes $\bar H$, whose derivative decreases on each branch and drops at the kink. For any deviation with $y(z,p)\le1-d$, monotonicity of $t-q[t-d]_+$ gives $H(z;p,d,q)\le\bar H(z;d)\le\bar H(p;d)=H(p;p,d,q)$. For a deviation reaching the opposite anonymous tail, $y(z,p)>1-d$ requires $0\le z<t:=p+2d-1<d$. Since $t<1/2<(1+p)/2$, we have $H(z;p,d,q)\le z(1+p-z)/2\le t(1-d)<pd\le H(p;p,d,q)$. The strict inequality follows from $pd-t(1-d)=(1-2d)(1-p-d)>0$. Thus, no deviation is profitable. The preceding local conditions exhaust symmetric candidates, establishing existence and uniqueness of the common uniform price and disclosure set in part {\rm(C)}.\qed

\end{APPENDICES}

\end{document}